\PassOptionsToPackage{unicode}{hyperref}
\PassOptionsToPackage{hyphens}{url}
\documentclass[12pt]{article}

\usepackage[T1]{fontenc}
\usepackage[utf8]{inputenc}
\usepackage{textcomp}
\usepackage[english]{babel}
\usepackage{lmodern}
\usepackage{amsmath,amssymb,amscd,amsthm,amsfonts,amsbsy,bm,mathrsfs,bbm}
\usepackage{graphicx,epsfig,epstopdf}
\usepackage{booktabs,array,longtable,makecell,multirow,diagbox,adjustbox}
\usepackage[flushleft]{threeparttable}
\usepackage{rotating,pdflscape}
\usepackage{algorithm}
\usepackage{algpseudocode}
\usepackage[font=footnotesize,labelfont=bf]{caption}
\usepackage{subcaption}
\usepackage{comment,float,afterpage,enumitem,pifont,xcolor,url}
\usepackage{natbib}
\usepackage{xr-hyper}
\usepackage[colorlinks,citecolor=blue,linkcolor=blue,urlcolor=blue]{hyperref}
\usepackage{bookmark}
\makeatletter
\let\saved@bibcite\bibcite
\let\bibcite\@gobbletwo
\let\bibcite\saved@bibcite
\makeatother

\graphicspath{{../}}

\makeatletter
\def\maxwidth{\ifdim\Gin@nat@width>\linewidth\linewidth\else\Gin@nat@width\fi}
\def\maxheight{\ifdim\Gin@nat@height>\textheight\textheight\else\Gin@nat@height\fi}
\setkeys{Gin}{width=\maxwidth,height=\maxheight,keepaspectratio}
\def\fps@figure{htbp}
\makeatother

\def\spacingset#1{\renewcommand{\baselinestretch}{#1}\small\normalsize}

\newtheorem{theorem}{{Theorem}}
\newtheorem{lemma}{{Lemma}}

\newtheorem{proposition}{{Proposition}}
\newtheorem{corollary}{{Corollary}}

\newtheorem{assumption}{{Assumption}}
\newtheorem*{assumption*}{{Assumption}}
\def\sym#1{\ifmmode^{#1}\else\(^{#1}\)\fi}

\def\A{\bm A}
\def\B{\bm B}

\def\I{\bm I}

\def\P{\bm P}
\def\SS{\mathbb S}
\def\Set{\mathcal S}
\def\J{\mathcal J}
\def\X{\bm X}
\def\Y{\bm Y}

\def\W{\bm W}
\def\bbeta{\bm \beta}
\def\ttheta{\bm \theta}
\def\eeta{\bm \eta}

\def\eepsilon{\bm \epsilon}
\def\Eepsilon{\mathbf{\mathcal{E}}}

\def\Ssigma{\bm \Sigma}

\def\pphi{\bm \phi}
\def\vvarphi{\bm \varphi}
\def\Ppi{\bm \Pi}

\def\oomega{\bm \omega}
\def\rrho{\bm \rho}
\def\xxi{\bm \xi}
\def\Xxi{\bm \Xi}
\def\Ggamma{\bm \Gamma}

\def\e{\bm e}
\def\g{\bm g}
\def\h{\bm h}
\def\x{\bm x}
\def\y{\bm y}
\def\r{\bm r}

\def\u{\bm u}
\def\v{\bm v}

\def\z{\bm z}

\def\o{\bm o}
\def\Stilde{\tilde S}
\def\stilde{\tilde s}
\newcommand{\PP}{\mathbb P}
\newcommand{\R}{\mathbb R}
\newcommand{\EE}{\mathbb E}
\newcommand{\LL}{\mathscr L}
\newcommand{\cmark}{\ding{51}}
\newcommand{\xmark}{\ding{55}}
\newcommand{\vertiii}[1]{{\vert\kern-0.25ex\vert\kern-0.25ex\vert #1 
    \vert\kern-0.25ex\vert\kern-0.25ex\vert}}
\newcommand{\bigvertiii}[1]{{\left\vert\kern-0.25ex\left\vert\kern-0.25ex\left\vert #1 
    \right\vert\kern-0.25ex\right\vert\kern-0.25ex\right\vert}}
\def\diag{\text{diag}}
\DeclareMathOperator*{\argmin}{arg\,min}

\DeclareMathOperator{\tr}{\text{tr}}
\DeclareMathOperator{\vvec}{\text{vec}}

\newcommand{\anon}{1}

\begin{document}

\spacingset{1}


\if1\anon
{
  \title{\bf Moment-Based Selection of Multiresponse Linear Mixed-Effects Models}
  \author{Yifan Chen  \thanks{
    Email address: ychen300@ucsb.edu} \\
    Department of Statistics and Applied Probability,\\
    University of California, Santa Barbara\\
    and \\
    Yuedong Wang \\
    Department of Statistics and Applied Probability,\\
    University of California, Santa Barbara\\
    and \\
    Guo Yu \\
    Department of Statistics and Applied Probability,\\
    University of California, Santa Barbara}
  \maketitle
} \fi

\if0\anon
{
  \bigskip
  \bigskip
  \bigskip
  \begin{center}
    {\LARGE\bf Moment-Based Selection of Multiresponse Linear Mixed-Effects Models}
  \end{center}
  \medskip
} \fi

\bigskip
\begin{abstract}
We propose \texttt{MOMENT} (\textbf{MO}ment-Based \textbf{M}ixed-\textbf{E}ffects Selectio\textbf{N} and Es\textbf{T}imation), a stage-wise moment-based framework that exploits second-order cross-moment identities to select and estimate the random-effects covariance matrix and fixed-effects coefficients. By inducing sparsity through its diagonal under a positive semidefinite constraint, the random-effects selection problem reduces to a smooth constrained convex optimization problem that can be solved efficiently by projected gradient descent. We further establish finite-sample theoretical guarantees for the proposed procedure, including random-effects selection consistency and fixed-effects selection consistency under joint sub-Weibull errors. Simulation studies show that \texttt{MOMENT} performs competitively overall and can substantially outperform separate univariate analyses when responses are correlated. An application to the hemodialysis dataset demonstrates that the proposed method yields an interpretable and flexible approach for multivariate longitudinal data.
\end{abstract}

\noindent%
\vfill

\newpage
\spacingset{1.8} 

\section{Introduction}\label{sec-intro}
Linear mixed-effects models (LMMs) are essential for analyzing non-independent data, such as repeated measurements, clustered, or longitudinal data, which are ubiquitous across the physical, biomedical, and social sciences. They allow researchers to decompose variability into population-level trends (i.e., fixed effects) and subject-specific deviations (i.e., random effects). Compared with the well-studied problem of fixed-effects selection, a central challenge in mixed-effects modeling is identifying relevant random effects that capture subject-specific deviations from the population trajectory.  

The existing literature on model selection for linear mixed-effects models can be broadly classified into \emph{likelihood-based} and \emph{moment-based} methods, with the former comprising the majority of the work. Despite their widespread use, likelihood-based methods, together with the restricted maximum likelihood (REML) approaches \citep{verbyla1993modelling, lindstrom1988newton, jiang2007linear}, often face substantial computational challenges: the full negative log-likelihood is typically nonconvex, complicating the development of efficient, scalable algorithms and limiting the development of sharp finite-sample theoretical guarantees. 
This challenge is even more pronounced in settings where the number of random effects or the number of responses is large, which is increasingly common in modern applications.
As a result, many existing approaches impose strong structural assumptions on the model that may be restrictive in practice. 
Specifically, a prominent line of work leverages Cholesky-based parameterization of the covariance matrix of the random effects. Building on the modified Cholesky reparameterization used in earlier Bayesian random-effects selection work \citep{chen2003random, kinney2007fixed}, \cite{bondell2010joint} proposed a penalized likelihood procedure for simultaneous selection of fixed and random effects.
\cite{ibrahim2011fixed} extended the standard Cholesky decomposition to generalized linear mixed-effects models (GLMMs), and \cite{li2018doubly} employed the same decomposition within the REML framework for LMMs, establishing asymptotic selection and estimation consistency. Nevertheless, likelihood formulations based on Cholesky parameterization remain nonconcave. More importantly, the Cholesky decomposition hinges on the order of the random effects and is thus not permutation invariant \citep{bickel2008regularized}.

Other likelihood-based methods focus on structured penalties and alternative selection strategies under stronger model assumptions.
\cite{hui2017hierarchical} introduced a penalty enforcing a \emph{hierarchical constraint}, whereby a random effect may be selected only if the corresponding fixed effect is included. An alternative approach to encouraging the hierarchy is to directly consider a subset of fixed effects as the potential random effects \citep{li2018doubly, heiling2024glmmpen}. \cite{hui2017joint} developed a penalized quasi-likelihood approach with adaptive group lasso regularization for joint selection of fixed and random effects. Separately, \cite{fan2012variable} studied penalized profile likelihood methods that decouple selection tasks: fixed effects are selected using a proxy for the random-effects covariance matrix, while random effects are selected via group regularization. More recently, \cite{heiling2024glmmpen} developed a Monte Carlo expectation conditional minimization (MCECM) algorithm for selecting GLMMs, which is efficient when the covariance matrix of the random effects is diagonal, but faces significant computational challenges otherwise. Under the same strong assumption of the diagonal covariance matrix for random effects in LMMs, \cite{sholokhov2024relaxation} proposed a proximal gradient algorithm for a class of penalized estimators with additional assumption of a known variance of the random error, and \cite{thompson2025scalable} developed a scalable hierarchical joint best-subset selection method using an $\ell_0$ penalty under the additional assumption of the identical design matrix for fixed and random effects.

In contrast, moment-based methods avoid distributional assumptions on the random effects and do not rely on the full or profile likelihood. Instead, they exploit low-order moment identities, which often lead to improved computational efficiency and simpler optimization procedures. For instance, \cite{peng2012model} proposed an iterative scheme that alternates between the selection and estimation of fixed effects and those of random effects. \cite{ahn2012moment} further decoupled the estimation and selection of random effects by first constructing a moment-based estimator of the covariance structure, and then performing refined selection of important random effects either through thresholding \citep{bickel2008covariance} or by solving a nonconvex optimization problem arising from a sandwich reparameterization.
See Table~\ref{table-literature} in Supplementary Material~\ref{appendix-comparison-table} for a concise comparison of the various methods discussed above.

Despite the extensive literature on model selection for LMMs, several important gaps remain. First, nearly all existing methods focus on models with univariate responses, whereas multivariate responses are common in real-world applications. For instance, electronic health record (EHR) data often contain numerous laboratory and clinical variables. Fitting univariate LMMs individually to each response ignores correlations among responses, preventing the estimation and interpretation of clinically meaningful cross-response associations. Extending mixed-effects models to the multivariate-response setting is inherently challenging, as it requires accurately characterizing covariance structures across multiple dimensions. Specifically, in addition to the covariance induced by multiple random effects, one must also account for the covariance structure of the multivariate error terms --- an issue absent in the univariate case. Properly modeling these cross-response dependencies is essential for both selection accuracy and interpretability, yet it introduces substantial additional complexity. Although \cite{duan2024estimating} proposed a procedure that separately estimates the covariance among multiple random effects and that of the multivariate error terms, their method is restricted to random-effects-only models and does not address the more general mixed-effects setting considered in this paper.
Second, many existing approaches either lack rigorous theoretical guarantees due to the nonconvex nature of the underlying optimization problems or provide only asymptotic results that require the number of subjects to diverge to infinity.

To address these gaps, we propose \texttt{MOMENT} (\textbf{MO}ment-Based \textbf{M}ixed-\textbf{E}ffects Selectio\textbf{N} and Es\textbf{T}imation), a convex, moment-based framework for variable selection in multi-response linear mixed-effects models (MLMMs).
Our contributions and the organization of the paper are summarized as follows. In Section \ref{sec-est}, we develop a novel moment-based procedure for model selection in the multi-response setting (with background reviewed in Section \ref{sec-method}). The proposed approach explicitly leverages cross-response correlations that are ignored when the responses are analyzed separately. The formulation does not impose structural constraints on the covariance or design matrices, thereby enabling flexible and fully general modeling. By relying solely on second-order moment equations rather than full likelihood evaluation, our method avoids strong distributional assumptions and remains computationally efficient.
In Section \ref{sec-theory}, we establish finite-sample guarantees for both random- and fixed-effects selection consistency under joint sub-Weibull errors, explicitly characterizing the relevant scaling and effective sample size. Finally, Sections \ref{sec-sim} and \ref{sec-apply} demonstrate the favorable empirical performance of our approach through simulation and real-data applications.

Throughout this paper, we adopt the following notational conventions unless otherwise specified. Scalars are denoted by regular (non-bold) letters, vectors by bold lowercase letters, and matrices by bold uppercase letters. Greek letters are reserved for unknown but fixed model parameters to be estimated. For any index $m$, denote $[m] = \{1,2,\ldots,m\}$. Proofs of all theoretical results are provided in the Supplementary Material.

\section{Multiresponse linear mixed-effects models}
\label{sec-method}

Suppose we observe repeated measurements from $m$ subjects, where for the $i$-th subject, the $j$-th observation (with $j \in [n_i]$) follows the linear mixed-effects model:
    \begin{equation}\label{model_eq}
        \y_{ij} = \bbeta^\top \x_{ij} + \B_{i}^\top\z_{ij} +\eepsilon_{ij},
    \end{equation}
where $ \y_{ij} \in \mathbb{R}^d$ is a $d$-dimensional response vector, $\x_{ij} \in \mathbb{R}^p$ and $\z_{ij} \in \mathbb{R}^q$ denote the covariates associated with the fixed and random effects, respectively, and $\eepsilon_{ij} \in \mathbb{R}^d$ is a $d$-dimensional random error vector. The matrix $\bbeta \in \mathbb{R}^{p \times d}$ is the fixed-effects coefficient matrix, while $\B_i\in \mathbb{R}^{q \times d}$ denotes the subject-specific random-effects matrix.
For generality, we assume only that
$\B_i$ are independent and identically distributed random matrices with mean $\mathbf{0} \in \mathbb{R}^{q \times d}$ across $i \in [m]$, and that $\eepsilon_{ij}$ are independent and identically distributed random vectors with mean $\mathbf{0} \in \mathbb{R}^{d}$ and covariance matrix $\Ssigma_{\epsilon}$ across $i \in [m]$ and $j \in [n_i]$.
As is standard in the mixed-effects literature, we further assume that $\eepsilon_{ij}$'s and $\B_i$'s are mutually independent.
Throughout this paper, the covariates $\x_{ij}$ and $\z_{ij}$ are treated as fixed. Our framework is designed to accommodate a broad range of dimensional regimes. We allow $d$, $p$, $q$, $m$, and $n_i$ to vary jointly, subject only to the scaling conditions needed for estimation and selection consistency (see Section~\ref{sec-theory}).

Model \eqref{model_eq} involves two distinct covariance structures: the covariance of the random effects $\B_i$ and that of the random errors $\eepsilon_{ij}$. 
We characterize the covariance structure of the matrix-valued random effect $\B_i$ through its vectorized representation, i.e., 
\begin{equation}
\small
     \Ssigma_{B} := \mathrm{Cov}  \left[\mathrm{vec} \left( \B_i \right) \right] 
    = \begin{bmatrix}
      \Ssigma^{(1, 1)}   & \Ssigma^{(1, 2)} & \ldots & \Ssigma^{(1, d)} \\
      \Ssigma^{(2, 1)}   & \Ssigma^{(2, 2)} & \ldots & \Ssigma^{(2, d)} \\
      \vdots   & \vdots & \ddots & \vdots \\
      \Ssigma^{(d, 1)}   & \Ssigma^{(d, 2)} & \ldots & \Ssigma^{(d, d)}
    \end{bmatrix} \in \mathbb{R}^{qd \times qd},
    \label{eq:covB}
\end{equation}
where each block $\Ssigma^{(r,u)} \in \mathbb{R}^{q \times q}$ represents the covariance between the $q$ random effects associated with the $r$-th response and those associated with the $u$-th response. This block structure explicitly captures both within-response dependence ($r = u$) and cross-response dependence ($r \neq u$) among the random effects. 
The two covariance components $\Ssigma_B$ and $\Ssigma_{\epsilon}$ jointly model the covariance structure on the responses in model \eqref{model_eq}. In particular, for any $i \in [m]$ and $j, k \in [n_i]$, the within-subject covariance is characterized by the following second-order moment equation: 
\begin{equation}
\label{equ-moment}
\small
    \text{Cov}(\y_{ij},\y_{ik}) = \mathbb{E} \Big[ \big( \y_{ij} - \bm{\beta}^\top \bm{x}_{ij} \big) \big( \y_{ik} - \bm{\beta}^\top \bm{x}_{ik} \big)^\top \Big] = 
\begin{cases} 
    (\bm{I}_d \otimes \bm{z}_{ij}^\top)\Ssigma_B (\bm{I}_d \otimes \bm{z}_{ik}) + \Ssigma_{\epsilon}, & j = k, \\
    (\bm{I}_d \otimes \bm{z}_{ij}^\top)\Ssigma_B (\bm{I}_d \otimes \bm{z}_{ik}), & j \neq k.
\end{cases} 
\end{equation}
In contrast to the univariate-response case, where the random-effects covariance is a single $q \times q$ matrix, the multivariate-response setting entails estimating a substantially more complex $qd \times qd$ covariance structure $\Ssigma_B$. Moreover, one must also estimate the error covariance matrix $\Ssigma_{\epsilon}$. These additional layers of complexity necessitate new methodological developments for scalable variable selection and estimation in MLMMs.

\section{The proposed \texttt{MOMENT} method} \label{sec-est}
In this section, we formally introduce the regularized estimation procedure for model \eqref{model_eq}, which involves three unknown parameters to be estimated: $\bbeta$, $\Ssigma_B$, and $\Ssigma_{\epsilon}$. Our approach is moment-based and builds upon the key second-order moment identity \eqref{equ-moment}. Given that a reliable working estimate of $\bbeta$ can be obtained and subsequently refined with relative ease, we adopt a stage-wise procedure with a proxy initial estimate of $\bbeta$. 
For the remainder of this section, we outline the general procedure in Algorithm \ref{alg-5step} and provide a detailed discussion of each step in the subsequent subsections.

\begin{algorithm}[!htbp]
    \spacingset{1.2}
	\renewcommand{\algorithmicrequire}{\textbf{Input:}}
	\renewcommand{\algorithmicensure}{\textbf{Output:}}
	\caption{\texttt{MOMENT}: a stage-wise procedure}
	\label{alg-5step}
	\begin{algorithmic}[1]
		\State Compute an initial working estimate of the fixed-effects coefficient matrix $\bbeta$.
	    \State Perform random-effects selection by solving a convex penalized optimization problem in $\Ssigma_B$ that promotes sparsity in its diagonal entries.
        \State (Optional) Refine the estimate of $\Ssigma_B$ based on the selected random effects.
        \State Estimate the error covariance matrix $\Ssigma_{\epsilon}$.
        \State Refine the estimation and selection of the fixed-effects coefficient matrix $\bbeta$.
	\end{algorithmic}  
\end{algorithm}

\subsection{Selection of random effects}
\label{subsec-random-sel}
In this subsection, we focus on the most challenging task of Algorithm~\ref{alg-5step}, namely the selection (Step 2) and estimation (Step 3) of $\Ssigma_B$. Random-effects selection is equivalent to recovering the support of $\diag(\Ssigma_B)$, since a zero diagonal entry corresponds to a zero-variance random-effect component, which vanishes almost surely under the zero-mean assumption. We propose to induce sparsity in  $\Ssigma_B$ directly by penalizing its diagonal elements. 
The idea of exploiting sparsity in $\Ssigma_B$ for random-effects selection is not new. For example, \cite{ahn2012moment} and \cite{bondell2010joint} adopted nonconvex sandwich parameterizations that effectively enforce sparsity across entire rows and columns of $\Ssigma_B$. Our approach is closer in spirit to \cite{peng2012model}, \cite{lin2013fixed}, and \cite{williams2015covariance}, who also penalized the diagonal elements of  $\Ssigma_B$. However, the methods of \cite{peng2012model} and \cite{williams2015covariance} do not guarantee the positive definiteness of the resulting estimator for the random-effects covariance matrix, which may adversely affect both selection and estimation performance. Meanwhile, \cite{lin2013fixed} formulated the problem as a penalized REML, yielding a nonconvex optimization problem. In contrast, our framework incorporates a positive semidefinite constraint on $\Ssigma_B$ within a convex optimization formulation. Consequently, the resulting estimator is a global optimizer and a valid covariance matrix, thereby improving both selection accuracy and estimation efficiency.

Given a working estimate $\hat{\bbeta}$ obtained from Step 1 in Algorithm \ref{alg-5step}, we construct the random-effects selection procedure by quantifying the discrepancy between the empirical and population second moment in the moment equation\eqref{equ-moment} where $j \neq k$ . This focus on the ``off-diagonal'' component in \eqref{equ-moment} is common in the covariance estimation literature \citep{ahn2012moment, yao2005functional}. Specifically, for any fixed $\hat{\bbeta}$, we consider the following optimization problem:
\begin{equation}
\small
       \min_{\Ssigma_B \succeq 0} \frac{1}{2} \sum_{i=1}^m \sum_{j \neq k}^{n_i} \left\| (\y_{ij}- \hat{\bbeta}^\top\x_{ij})(\y_{ik}-\hat{\bbeta}^\top\x_{ik})^\top - (\bm{I}_d \otimes \bm{z}_{ij}^\top)\Ssigma_B (\bm{I}_d \otimes \bm{z}_{ik}) \right\|_F^2 + \mathcal{P}_B(\mathrm{diag}(\Ssigma_B)),
    \label{equ-random-sel} 
\end{equation}
where $\mathrm{diag} (\Ssigma_B)$ denotes the vector of diagonal elements of $\Ssigma_B$, and $\mathcal{P}_B$ is a generic convex sparsity-inducing penalty. For the rest of this paper, we consider an (adaptive) $\ell_1$-penalty \citep{zou2006adaptive} of the form 
$$\mathcal{P}_B(\mathrm{diag}(\Ssigma_B)) = \lambda \sum_{t = 1}^{qd} \text{w}_{t} |\Ssigma_{B_{tt}}| = \lambda \|\W\cdot \diag(\Ssigma_B)\|_1,
$$ where $\lambda \geq 0$ is a tuning parameter, $\text{w}_{t}$ is a pre-specified weight for $\Ssigma_{B_{tt}},$ the $t$-th diagonal element of $\Ssigma_B$, and $\W\in\R^{dq \times dq}$ is a diagonal weight matrix consisted of $\text{w}_{t}$ for $t\in [dq].$ The following proposition establishes that, under the positive semidefinite constraint, the nonsmooth $\ell_1$-penalty on the diagonal is equivalent to a smooth trace penalty. This equivalence facilitates the development of an efficient optimization algorithm to solve \eqref{equ-random-sel}.
\begin{proposition}
\label{prop-equivalence}
    Problem \eqref{equ-random-sel} with $\mathcal{P}_B = \lambda \|\W\cdot \mathrm{diag}(\Ssigma_B)\|_1$ is equivalent to the following optimization problem that has a smooth objective function:
\begin{equation}
\small
       \min_{\Ssigma_B \succeq 0} \frac{1}{2} \sum_{i=1}^m \sum_{j \neq k}^{n_i} \left\| (\y_{ij}- \hat{\bbeta}^\top\x_{ij})(\y_{ik}-\hat{\bbeta}^\top\x_{ik})^\top - (\bm{I}_d \otimes \bm{z}_{ij}^\top)\Ssigma_B (\bm{I}_d \otimes \bm{z}_{ik}) \right\|_F^2 + \lambda \cdot \mathrm{tr}(\W \cdot\Ssigma_B).
    \label{equ-equivalent}
\end{equation}
\end{proposition}
In contrast to the nonsmooth objective in \eqref{equ-random-sel}, the objective function in the equivalent formulation \eqref{equ-equivalent} is smooth and therefore substantially easier to optimize. The problem can be efficiently solved using a proximal gradient descent algorithm \citep{nesterov2013gradient,parikh2014proximal}, which in this setting is also commonly referred to as projected gradient descent. Each iteration requires evaluating the proximal operator associated with the constraint set $\{\Ssigma_B \succeq 0\}$, which admits a closed-form solution given by the projection onto the positive semi-definite cone. Full algorithmic details, including the choice of the weights $\text{w}_t$ and a cross-validation procedure to select the tuning parameter $\lambda$, are provided in Supplementary Material~\ref{appendix-algo}.

One may further refine the estimate of $\Ssigma_B$ by refitting \eqref{equ-equivalent}, with $\lambda = 0$, given the selected random effects $\hat{S} = \{t \in [qd]: \hat{\Ssigma}_{B_{tt}} \neq 0\}$. Specifically, 
a solution can be obtained by first solving the optimization problem in a reduced dimension corresponding to the selected active index set using proximal gradient descent, and then embedding the resulting estimator back into the original matrix dimension by assigning zeros to the rows and columns associated with inactive random effects.

\subsection{Estimation of the random error covariance}
\label{subsec-sigmae}
For generality, we do not impose any structural assumptions on the error covariance matrix $\Ssigma_{\epsilon}$. In Step 4 of Algorithm \ref{alg-5step}, the optimization problem for estimating $\Ssigma_{\epsilon}$ can be formulated using the ``diagonal'' blocks where $j = k$ in the moment equation \eqref{equ-moment}, given the current estimates of $\bbeta$ and $\Ssigma_B$. Specifically, we consider the following optimization problem:
\begin{equation}
\label{equ-opt-sigmae}
\min_{\Ssigma_{\epsilon}\succeq 0} \frac{1}{2} \sum_{i=1}^m \sum_{j = 1}^{n_i} \left\| (\y_{ij}- \hat{\bbeta}^\top\x_{ij})(\y_{ij}- \hat{\bbeta}^\top\x_{ij})^\top - (\bm{I}_d \otimes \bm{z}_{ij}^\top)\hat{\Ssigma}_B (\bm{I}_d \otimes \bm{z}_{ij}) - \Ssigma_{\epsilon} \right\|_F^2, 
\end{equation}
which is an unpenalized problem since we do not impose any structural assumptions on $\Ssigma_{\epsilon}$. Then, \eqref{equ-opt-sigmae} has a closed-form solution
\begin{equation}
    \hat{\Ssigma}_\epsilon = \mathrm{proj}_{\mathbb{S}^d_+} \left( \frac{1}{N} \sum_{i=1}^m \sum_{j=1}^{n_i} \left[(\y_{ij}-\hat{\bbeta}^\top \x_{ij})(\y_{ij}-\hat{\bbeta}^\top \x_{ij})^\top -(\I_d \otimes \z_{ij}^\top)\hat{\Ssigma}_{B}(\I_d \otimes \z_{ij})\right] \right),
    \nonumber
\end{equation}
where $N=\sum_{i=1}^m n_i$ is the total number of observations and $\mathrm{proj}_{\mathbb{S}^d_+}(\bm{M})$ corresponds to projecting a symmetric matrix $\bm{M} \in \mathbb{R}^{d \times d}$ onto the positive semidefinite cone $\mathbb{S}^d_+$ and admits a closed-form solution via eigen-decomposition
$\sum_{h=1}^d \max(v_h,0)\u_h \u_h^\top,$
where $\bm{M}=\sum_{h=1}^d v_h \u_h \u_h^\top$ is the eigen-decomposition of $\bm{M}$.

A class of stage-wise approaches to variable selection in mixed models proceeds by first projecting out the random effects in \eqref{model_eq} to estimate auxiliary quantities, such as $\bbeta$ \citep{wu2017new} or $\Ssigma_{\epsilon}$ \citep{peng2012model}, which are then used to construct an estimator of $\Ssigma_B$. Such procedures typically rely on strong conditions on the subject-specific sample sizes, in particular requiring $n_i > q$, to ensure identifiability of the intermediate estimators. In contrast, our method avoids these requirements by directly estimating $\Ssigma_B$.

\subsection{Initial estimation and refinement of fixed-effects coefficients}
\label{subsec-beta}
In Step 1 of Algorithm \ref{alg-5step}, a common choice for the initial estimate of $\bbeta$ is the ordinary least squares (OLS) estimator obtained by fitting the working homoscedastic multiresponse linear model $\y_{ij} = \x_{ij}^\top \bbeta + \eepsilon_{ij}$, effectively omitting the random effects from \eqref{model_eq}. 


Algorithm \ref{alg-5step} finishes with a refinement of the fixed-effects coefficients that can be considered. With the estimates of $\Ssigma_B$ and $\Ssigma_{\epsilon}$, covariance structures between and within responses can be estimated using \eqref{equ-moment}. Consequently, a natural approach to refine the estimate of $\bbeta$ is the feasible generalized least squares (FGLS) \citep{baltagi2008econometrics}. 

Consider the stack form of model \eqref{model_eq} on the subject level. Specifically, for any $i \in [m]$, we have $\bm{Y}_i = \left[\bm{y}_{i1}^\top,  \bm{y}_{i2}^\top, ..., \bm{y}_{in_i}^\top \right]^\top \in \mathbb{R}^{n_i \times d}$ and $\bm{X}_i = \left[\bm{x}_{i1}^\top,  \bm{x}_{i2}^\top, ..., \bm{x}_{in_i}^\top \right]^\top \in \mathbb{R}^{n_i \times p}$. By vectorizing the response matrix $\Y_i$, we have $\EE[\mathrm{vec}(\Y_i^\top)] = \mathrm{vec}(\X_i \bbeta)$ and
\begin{equation}
\label{equ-sigmai}
  \Ssigma_{i}:= \mathrm{Cov} [\mathrm{vec}(\Y_i^\top)] = \mathcal{Z}_i \Ssigma_B \mathcal{Z}_i^\top + \I_{n_i} \otimes \Ssigma_{\epsilon} \in \mathbb{R}^{n_i d \times n_i d},
\end{equation}
where $\mathcal{Z}_i = [\I_d \otimes \z_{i1}, ..., \I_d \otimes \z_{i n_i}]^\top \in \R^{n_i d \times qd}$. And $\Ssigma_i$ can be estimated by plugging in the estimates of $\Ssigma_B$ and $\Ssigma_{\epsilon}$ obtained in the previous steps.
Consequently, the penalized FGLS estimator of $\bbeta$ is given by solving the following optimization problem:
\begin{equation}\label{equ-fgls}
    \min_{\bbeta} \frac{1}{2N} \sum_{i=1}^m \left\| \hat{\Ssigma}_{i}^{-1/2} \mathrm{vec}(\Y_i^\top) - \hat{\Ssigma}_{i}^{-1/2} (\X_i \otimes \I_d) \mathrm{vec}(\bbeta^\top) \right\|_2^2 + \mathcal{P}_\beta(\bbeta),
\end{equation}
where the penalty function $\mathcal{P}_\beta$ can be chosen to promote various structural patterns in $\bbeta$ depending on the specific application and scientific context.
For instance, in the literature of multiresponse linear regression models, various structural patterns for $\bbeta$ have been considered \citep{liu2025estimation}, leading to different choices of $\mathcal{P}_\beta$. These include elementwise sparsity via an $\ell_1$ penalty \citep{rothman2010sparse}, row-wise sparsity via a group lasso penalty \citep{argyriou2008convex, obozinski2011support}, and low-rank structure via a nuclear norm penalty \citep{yuan2007dimension, chen2013reduced, gregory2022multivariate}, among many other possibilities and combinations thereof. 

Note that although both the estimated covariance matrices $\hat{\Ssigma}_B$ and $\hat{\Ssigma}_\epsilon$ are guaranteed positive semidefinite, the resulting $\hat{\Ssigma}_i$ is not necessarily positive definite and may still be singular, making \eqref{equ-fgls} ill-posed. 
A natural strategy to address this issue is to add a small ridge-type perturbation to the covariance matrix, which guarantees invertibility and improves numerical stability, i.e., 
\begin{equation}
    \min_{\bbeta} \frac{1}{2N} \sum_{i=1}^m \left\| (\hat{\Ssigma}_{i} + \gamma \I_{n_id})^{-1/2} \mathrm{vec}(\Y_i^\top) - (\hat{\Ssigma}_{i} + \gamma \I_{n_id})^{-1/2} (\X_i \otimes \I_d) \mathrm{vec}(\bbeta^\top) \right\|_2^2 + \mathcal{P}_\beta(\bbeta),
\end{equation}
for some small $\gamma > 0$. Note that this procedure can be viewed as a penalized FGLS estimator of $\bbeta$ where the responses are perturbed by an additional homoscedastic noise with a pre-specified variance $\gamma$. In practice, one can repeat the process multiple times and then aggregate the resulting estimates to reduce variance and stabilize variable selection. For the particular choice of $\mathcal{P}_\beta (\bbeta) = \tau \|\vvec(\bbeta^\top)\|_1$ for tuning parameter $\tau > 0$, this procedure coincides with the bootstrap lasso \citep{bach2008bolasso}, which has been shown to be effective in improving the stability of variable selection in high-dimensional regression problems.

\section{Theoretical analysis}
\label{sec-theory}
We begin by introducing notation useful for theoretical analysis. Let $N_i = \sum_{j\neq k}^{n_i} 1 = n_i(n_i - 1)$ and $M =  \sum_{i=1}^m N_i$. For a sub-Weibull($\alpha$) random variable $x$, define the associated Orlicz norm by $\|x\|_{\psi_\alpha} = \inf \left\{\varphi>0:\EE\left[\exp\left(\frac{|x|}{\varphi}\right)^\alpha\right] \leq 2\right\}$, and for a random vector $\X$ following a joint sub-Weibull($\alpha$) distribution, define its Orlicz norm as $\|\X\|_{J,\psi_{\alpha}} = \sup_{\|\eeta\|_2 = 1} \|\eeta^\top \X \|_{\psi_{\alpha}}$. Let $\Lambda_{\min}(\B)$ and $\Lambda_{\max}(\B)$ denote the minimum and maximum singular value of matrix $\B$, respectively. We denote the matrix $\ell_\infty/\ell_\infty$ operator norm by $\vertiii{\B}_{\infty} = \max_i \sum_j |\B_{ij}|$ and matrix $\ell_2/\ell_2$ operator norm (i.e., the spectral norm) by $\vertiii{\B}_2 = \Lambda_{\max}(\B)$. We write $a \vee b = \max\{a,b\}$. 
Finally, for sequences $f(n)$ and $g(n)$, we write $f(n) = \mathcal O(g(n))$ if there exists a constant $C>0$ such that $f(n) / g(n) \leq C$ as $n\rightarrow\infty$. For a deterministic positive sequence $r_n$, we write $X_n=\mathcal{O}_p(r_n)$ if $X_n/r_n$ is bounded in probability, i.e., for every $\epsilon>0$, there exist constants $M_\epsilon<\infty$ and $N_\epsilon$ such that $\PP(|X_n|>M_\epsilon r_n)<\epsilon$ for all $n\ge N_\epsilon$.

The remainder of this section is devoted to establishing finite-sample selection consistency for both random and fixed effects.

\subsection{Random-effects selection problem reformulation}
We establish sign selection consistency using the primal-dual witness (PDW) framework of \citep{wainwright2009sharp}, applied to an unconstrained surrogate of the random-effects selection problem \eqref{equ-random-sel}. Specifically, we remove the positive semidefinite (PSD) constraint and treat the fixed-effect component as known by replacing the working estimator $\hat \bbeta$ with the true $\bbeta^\star$. A full analysis of problem \eqref{equ-random-sel} with $\hat{\bbeta}$ plugged in would require additional conditions ensuring that the estimation error bound for $\hat{\bbeta}-\bbeta^\star$ in Step 1 is sufficiently small. For simplicity, we use the oracle $\bbeta^\star$ to isolate the random-effects support recovery mechanism induced by the diagonal penalty. Under these simplifications, the problem reduces to 
\begin{equation}
\label{equ-theory-original}
    \min_{\Ssigma_{B}} \frac{1}{2}\sum_{i=1}^m \sum_{\substack{j,k \\ j\neq k}}^{n_i} \left\| (\I_d \otimes \z_{ij}^\top) \Ssigma_{B} (\I_d \otimes \z_{ik})-(\y_{ij}-\bbeta^{\star\top} \x_{ij})(\y_{ik}-\bbeta^{\star\top} \x_{ik})^\top\right\|_F^2 + \lambda \|\W\cdot \diag(\Ssigma_B)\|_1.
\end{equation}
This formulation can be rewritten as a linear regression problem with an unpenalized nuisance parameter corresponding to the off-diagonal entries of $\Ssigma_B$. Specifically, let $\ttheta = \diag(\Ssigma_B) \in \R^{qd}$ and $\rrho = (\Ssigma_{B_{12}},...,\Ssigma_{B_{qd,qd-1}})^\top \in \R^{qd(qd-1)}$ collect the diagonal and off-diagonal entries of $\Ssigma_B$,  with corresponding true values $\ttheta^\star$ and $\rrho^\star$, respectively. Then \eqref{equ-theory-original} corresponds to the linear model:
\begin{equation} \label{equ-theory-linreg}
\r = \A_{\mathcal{I}} \ttheta^\star + \A_{\mathcal{I}^c} \rrho^\star + \xxi,
\end{equation}
where $\mathcal{I} = \{(j-1)qd+j : j \in [qd]\}$ indexes the diagonal entries of $\Ssigma_B$. For the $i$-th subject, $(j, k)$-th measurement pair, and the $(\ell, \ell')$-th response components, the response vector $\r \in \mathbb{R}^{Md^2}$ has entries $\r_{\pi(i,j,k,\ell,\ell')} = (y_{ij\ell}-{\bbeta^\star_{\ell}}^\top \x_{ij})(y_{ik\ell'}-{\bbeta^\star_{\ell'}}^\top \x_{ik})$, and the design matrix $\A\in\R^{Md^2\times q^2d^2}$ by rows $\A_{\pi(i,j,k,\ell,\ell'),\cdot} = [(\e_{\ell} \otimes \z_{ik}) \otimes (\e_{\ell'} \otimes \z_{ij})]^\top$, where $\pi$ is a bijective index mapping. 

Before proceeding, we highlight several technical challenges in applying the PDW framework.
Although the random error vector $\xxi$ has mean zero, its entries are neither Gaussian nor independent. From \eqref{model_eq}, each entry of $\r$ takes the form
$(\B_{i,\ell}^\top \z_{ij}+\epsilon_{ij\ell})
(\B_{i,\ell'}^\top \z_{ik}+\epsilon_{ik\ell'})$,
for $i\in[m]$, $j\neq k\in[n_i]$, and $\ell,\ell'\in[d]$.
These products of random variables are generally non-sub-Gaussian, with tail behavior determined by, and typically heavier than, that of the underlying random effects and random error.
Moreover, the dependence structure is nonstandard: while observations across subjects are independent, within each subject the entries indexed by $(j,k,\ell,\ell')$ are dependent. 
This blockwise independence structure among $\xxi = (\xxi_1^\top, \ldots, \xxi_m^\top)^\top$ distinguishes our setting from the classical lasso setting.

Since the off-diagonal component $\rrho$ is unpenalized, it can be profiled out via orthogonal projection. We further assume $\W$ is the identity matrix, since the general case can be reduced to a standard lasso formulation through an appropriate rescaling of the covariates \citep{buhlmann2011statistics}. Consequently, \eqref{equ-theory-original} is equivalent to the lasso problem
\begin{equation}
\label{equ-theory-lasso}
    \min_{\ttheta} \frac{1}{2} \|\mathcal{Y} - \mathcal{X}\ttheta\|_2^2 + \lambda\|\ttheta\|_1,
\end{equation}
where $ \mathcal{Y} = \P_{\A_{\mathcal{I}^c}^\perp}\r$, $\mathcal{X} = \P_{\A_{\mathcal{I}^c}^\perp} \A_{\mathcal{I}}$, and $\P_{\A_{\mathcal{I}^c}^\perp} = \I - \A_{\mathcal{I}^c}({\A^\top_{\mathcal{I}^c}} \A_{\mathcal{I}^c})^{+}{\A^\top_{\mathcal{I}^c}} \in \R^{Md^2\times Md^2}$ denotes the projection onto the orthogonal complement of $\A_{\mathcal{I}^c}$.

\subsection{Random-effects selection consistency} \label{sec-theory-ran}
Denote $S=\mathrm{supp}(\ttheta^\star) = \{j:\theta_j^{\star} \neq 0\}$ as the support of $\ttheta^\star$ and $s = |S|$ as the cardinality of set $S$. The following assumptions are used for proving random-effects selection consistency:

\begin{assumption}[Irrepresentability]
\label{a1}
There exists $\nu \in (0,1]$ such that
$$\vertiii{\mathcal{X}^\top_{S^c} \mathcal{X}_S (\mathcal{X}_S^\top \mathcal{X}_S)^{-1}}_{\infty} \leq 1-\nu.$$
\end{assumption}

\begin{assumption}[Minimum eigenvalue]
\label{a2}
There exists $C_{\min}>0$ such that
$$\Lambda_{\min} \left(\frac{1}{Md^2}\mathcal{X}_S^\top \mathcal{X}_S \right) \geq C_{\min}.$$
\end{assumption}

\begin{assumption}[Tail behavior]
\label{a3}
The random effects $\vvec(\B_{i})$ and random error $\eepsilon_{ij}$ are both joint sub-Weibull($\alpha$) random vectors with $\alpha > 0$. 
\end{assumption}

\begin{theorem}[Random-effects selection consistency]
\label{thm-sel-consist}
Assume Assumption \ref{a1} - \ref{a3} hold and that the columns are normalized so that $(Md^2)^{-\frac{1}{2}} \max_{j \in S^c} \|\A_{\mathcal{I}_j}\| \leq 1$. 
For any $\delta>0,$ suppose
$$\lambda \geq \frac{2d\sqrt{M} \max_i K_i}{c \nu}\max\left\{\sqrt{\log\frac{2(qd-s)}{\delta}}, \left(\log\frac{2(qd-s)}{\delta}\right)^{2/\alpha}\right\},$$
where $K_i = \|\xxi_i\|_{J,\psi_{\alpha/2}} = \sup_{\|\eeta\|_2 = 1} \|\eeta^\top \xxi_i \|_{\psi_{\alpha/2}}$ is the Orlicz norm of joint sub-Weibull($\frac{\alpha}{2}$) random vector $\xxi_i,$ $c$ is some constant, $\nu$ and $C_{\min}$ are constants from assumption \ref{a1} and \ref{a2} respectively. Then the following properties hold with probability greater than $1-\delta$:
\begin{enumerate}
    \item [(a)] The problem \eqref{equ-theory-lasso} has a unique solution $\hat{\ttheta} \in \R^{qd}$ with its support contained within the true support (i.e. $\mathrm{supp}(\hat{\ttheta}) \subseteq \mathrm{supp}(\ttheta^\star)$), and 
\begin{equation*}
        \|\hat{\ttheta}_S - \ttheta^\star_S\|_{\infty} \leq \underbrace{\frac{\lambda}{Md^2} \bigvertiii{\left(\frac{1}{Md^2}\mathcal{X}_S^\top \mathcal{X}_S \right)^{-1}}_{\infty} + \frac{\max_i K_i}{c \sqrt{C_{\min} M} d}\max \left\{\sqrt{\log\frac{2s}{\delta}}, \left(\log\frac{2s}{\delta}\right)^{2/\alpha}\right\}}_{f(\lambda)}.
\end{equation*}
    \item [(b)] If $\min_{j \in S} |\theta_j^\star| > f(\lambda)$, then $\hat{\ttheta}$ recovers the correct signed support.
\end{enumerate}
\end{theorem}
Assumption \ref{a1} and \ref{a2} are the standard in lasso literature \citep{wainwright2009sharp}. Assumption~\ref{a1} is the irrepresentability condition for the projected design matrix $\mathcal{X} = \P_{\A_{\mathcal{I}^c}^\perp} \A_{\mathcal{I}}$, which guarantees that, after removing the nuisance contribution associated with the off-diagonal entries of $\Ssigma_B$, the inactive diagonal coordinates are not too strongly correlated with the active ones. This condition is exactly what allows us to verify strict dual feasibility in the PDW construction. Assumption~\ref{a2}, in turn, is a restricted positive-definiteness condition on the active set. It ensures that the normalized Gram matrix $(Md^2)^{-1}\mathcal{X}_S^\top \mathcal{X}_S$ is well conditioned and guarantees uniqueness of the PDW construction.

Assumption~\ref{a3} controls the tail behavior of the projected noise in \eqref{equ-theory-linreg}. The sub-Weibull family \citep{kuchibhotla2022moving} includes sub-Gaussian ($\alpha \geq 2$), sub-exponential ($\alpha = 1$), and heavier-tailed distributions ($0 < \alpha < 1$). 
The tail parameter $\alpha$ determines the magnitude of the regularization parameter $\lambda$, which controls the stochastic error in the PDW analysis and therefore influences the convergence rate of the estimator.
For simplicity, we assume the same $\alpha$ for the random effects and the random errors, although the argument extends routinely to different tail parameters.

By focusing on a realistic setting of $0 < \alpha \leq 4$ for the rest of this section, with 
\begin{align}
\lambda \ge 
\frac{2\sqrt{M}d\,\max_i K_i}{c\nu}
\left(\log\frac{2 [s \vee (qd-s)]}{\delta}\right)^{2/\alpha},
\label{eq:lambda}
\end{align}
the $\ell_\infty$-bound in Theorem~\ref{thm-sel-consist}(a) can be simplified as
$$\|\hat{\ttheta}_S-\ttheta_S^\star\|_\infty
\leq 
\frac{\lambda}{Md^2}
\left(
\bigvertiii{\left(\frac{1}{Md^2}\mathcal{X}_S^\top \mathcal{X}_S\right)^{-1}}_\infty
+
\frac{\nu}{\sqrt{C_{\min}}}
\right).$$
We finish this subsection with the following corollary, which makes this bound more explicit by quantifying the Orlicz norms $K_i$ (see Supplementary Material~\ref{appendix-corollary-rate} for details). 
\begin{corollary}
\label{corollary-rate}
    Assume in addition that there exists constants $C_z, \sigma_{\epsilon} > 0$ such that $\max_{ij \ell}\|\epsilon_{ij\ell}\|_{\psi_{\alpha}} \leq \sigma_{\epsilon}$ and $\max_{i,j} \|\z_{ij}\|_2 \leq C_z.$ Let $\kappa = \max_{i,\ell}\|\pphi_{i,\ell}\|_{J,\psi_{\alpha}}$, where $\B_{i,\ell} = \Ssigma_B^{\star(\ell,\ell)^{1/2}} \pphi_{i,\ell}$ and $\Ssigma_B^{\star (\ell,\ell)}$ denotes the true covariance matrix of $\B_{i,\ell}$ as in \eqref{eq:covB}. With $\lambda$  in \eqref{eq:lambda}, it holds with probability at least $1 - \delta$ that
    \begin{align}
      \|\hat{\ttheta}_S-\ttheta_S^\star\|_\infty \leq
      \begin{cases}\displaystyle
        \frac{\tilde{C}\kappa^2 \max_{\ell} \vertiii{\Ssigma_B^{\star (\ell,\ell)}}_2 \left(\log [s \vee (qd-s)]\right)^{2/\alpha}}{\sqrt{M / \max_i N_i}}, \quad &  2 \leq \alpha \leq 4, \\[1.2em]\displaystyle
       \frac{\tilde{C}\kappa^2 \max_{\ell} \vertiii{\Ssigma_B^{\star (\ell,\ell)}}_2 d^{4/\alpha - 2} \left(\log [s \vee (qd-s)]\right)^{2/\alpha}}{\sqrt{M/(\max_i N_i)^{4/\alpha - 1}}}, \quad & 0 < \alpha < 2 
      \end{cases}
      \label{equ-corollary-rate}
    \end{align}
    for some universal constant $\tilde{C} > 0$.
\end{corollary}
The rates depend on the conditioning of the true covariance matrix $\Ssigma_B^\star$ through $\max_{\ell} \vertiii{\Ssigma_B^{\star (\ell,\ell)}}_2$, the largest spectral norm among the diagonal blocks of $\Ssigma_B^\star$. This quantity plays a similar role to the noise level in standard lasso results. With bounded entries of $\Ssigma_B^\star$, this quantity scales linearly with the maximum number of relevant random effects across responses.
Although $\kappa$ may in principle grow with the ambient dimension $q$, it remains constant in important special cases. In particular, when $\B_{i,\ell}$ is Gaussian, $\pphi_{i,\ell}$ is standard Gaussian, so $\kappa$ is a universal constant.

The rates in \eqref{equ-corollary-rate} exhibit a phase transition at $\alpha = 2$. The reason is technical yet important: when $\alpha < 2$, the functional $\|\cdot\|_{\psi_{\alpha/2}}$ is only a quasi-norm, so the usual triangle inequality is no longer available to control $K_i$. As a result, the light-tailed and heavy-tailed regimes yield qualitatively different rates. For $\alpha \geq 2$, the bound resembles the familiar sub-Gaussian with a logarithmic dependence on the dimension $qd$, while for $0 < \alpha < 2$, the rate deteriorates and acquires an additional polynomial dependence on $d$.

This distinction is also reflected in the effective sample size. When the random effects and random errors are sub-Gaussian or have lighter tails (i.e., $\alpha \geq 2$), the effective sample size is of order $M / \max_i N_i$. In the balanced setting, where all $n_i$'s are equal, this is of the same order of $m$, the number of subjects, leading to the standard rate of $m^{-1/2}$, which is consistent with existing results for covariance estimation in repeated measurements \citep{duan2024estimating}. 
In contrast, for heavier tails (with $0 < \alpha < 2$), the effective sample size deteriorates to $M / (\max_i N_i)^{4/\alpha - 1}$. Thus, heavier tails reduce the amount of usable information contributed by each subject-specific block. 

Finally, taking $\alpha \geq 2$ as an example, consider the lasso problem \eqref{equ-theory-lasso} under the homoskedastic model $\mathcal{Y} = \mathcal{X} \ttheta + \bm{\varepsilon}$, where entries of $\bm{\varepsilon} \in \mathbb{R}^{Md^2}$ are i.i.d. sub-Weibull ($\alpha$) random errors. A standard $\ell_\infty$-bound in this setting is 
$ \sigma M^{-1/2} d^{-1} (\log (s \vee (qd - s))^{1/\alpha}$, where $\sigma$ is the Orlicz norm of entries of $\bm{\varepsilon}$.
In comparison, the rate established in \eqref{equ-corollary-rate} is strictly slower. This deterioration stems from the two challenges discussed earlier. First, the moment equations involve products of sub-Weibull($\alpha$) random variables, which exhibit heavier tails and thus lead to a larger numerator in \eqref{equ-corollary-rate}. Second, accounting for both the dependence induced by repeated measurements within each subject and the dependence across $d$ responses further worsens the denominator.
We do not claim optimality of the rate in \eqref{equ-corollary-rate}. Whether this rate can be improved, and more generally, the minimax rate for this problem, remain important directions for future research.

\subsection{Fixed-effects selection consistency}
We next establish fixed-effects selection consistency for the FGLS estimator in \eqref{equ-fgls} with the entrywise penalty
$\mathcal{P}_\beta(\bbeta)=\tau \|\mathrm{vec}(\bbeta^\top)\|_1$ with tuning parameter $\tau > 0$. The heterogeneous covariance estimator $\hat{\Ssigma}_i$ in \eqref{equ-sigmai} is constructed from the refitted covariance estimators $\hat{\Ssigma}_B$ and $\hat{\Ssigma}_{\epsilon}$ obtained in Steps~3--4 of Algorithm~\ref{alg-5step}. 
Recall that the corresponding linear model is
$$\vvec(\Y_i^\top) = (\X_i \otimes \I_d) \vvec(\bbeta^{\star \top}) + \bm \varepsilon_i,$$
where $\EE(\bm \varepsilon_i) = \bm 0$ and $\text{Cov}(\bm \varepsilon_i) = \Ssigma_i^\star.$ Let $\Stilde = \mathrm{supp}[\vvec(\bbeta^{\star\top})] = \{j: \vvec(\bbeta^{\star\top})_j \neq 0\}$ denote the support of $\vvec(\bbeta^{\star \top})$, and $\stilde = |\Stilde|$. Define the population Gram matrix $\Ggamma = \frac{1}{N}\sum_{i=1}^m (\X_i \otimes \I_d)^\top \Ssigma_i^{\star^{-1}} (\X_i \otimes \I_d) \in \R^{pd \times pd}.$ Finally, set $\o_i = \Ssigma_i^{\star^{-1/2}} \bm \varepsilon_i \in \R^{n_i d}$, so that $\o_i$'s are independent decorrelated noise vectors with zero mean and identity covariance. We impose the following assumptions.
\begin{assumption}
    \label{a-fix2}
    There exists a constant $\nu_x \in (0,1]$ such that
    $$\bigvertiii{\Ggamma_{\Stilde^c \Stilde} \Ggamma_{\Stilde \Stilde}^{-1}}_{\infty} \leq 1 - \nu_x.$$
\end{assumption}

\begin{assumption}
    \label{a-fix3}
    There exists a constant $C_e> 0$ such that 
    $$\bigvertiii{\Ggamma_{\Stilde \Stilde}^{-1}}_\infty \leq C_e^{-1}$$
\end{assumption}

\begin{assumption}
    \label{a-fix4}
    The random vectors $\o_i$'s are joint sub-Weibull($\alpha$) with $\alpha > 0.$
\end{assumption}

\begin{theorem}[Fixed-effects selection consistency]
    \label{thm-fixed-sel-consistency}
    Suppose that the columns are standardized such that $\max_{j \in [pd]} \frac{1}{N} \sum_{i=1}^m \|\hat \Ssigma_i^{-1/2} (\X_i \otimes \I_d)_j\|_2^2 \leq 1$, where $(\X_i \otimes \I_d)_j \in \R^{n_i d}$ is the $j$-th column of $\X_i \otimes \I_d$. Furthermore, suppose
    $\max_i \vertiii{\hat \Ssigma_i^{1/2} (\hat{\Ssigma}_i^{-1} - \Ssigma_i^{\star^{-1}}) \hat \Ssigma_i^{1/2}}_2 \leq a_m$ holds with probability at least $1 - \delta_{\Sigma}$,
    where $a_m \leq \frac{\sqrt{5} - 1}{2}$ and $\stilde a_m C_e^{-1} \leq \frac{\nu_x}{8}$. Then for any $\delta_{\beta} > 0$, suppose
    \begin{equation*}
        \tau \geq \frac{8 (1+a_m) \max_i \tilde{K}_i}{c \nu_x \sqrt{N}}\max \left\{\sqrt{\log \frac{2pd}{\delta_{\beta}}}, \left(\log \frac{2pd}{\delta_{\beta}}\right)^{1/\alpha}\right\},
    \end{equation*}
    where $\tilde{K}_i = \|\o_i\|_{J,\psi_{\alpha}} = \sup_{\|\eeta\|_2 = 1} \|\eeta^\top \o_i\|_{\psi_{\alpha}}$ is the Orlicz norm of the joint sub-Weibull($\alpha$) vector $\o_i.$ Then with probability greater than $1-\delta_{\Sigma} - \delta_{\beta}$, the following properties hold:
    \begin{enumerate}
    \item [(a)] The problem \eqref{equ-fgls} with $\mathcal{P}_\beta (\bbeta) =\tau  \|\mathrm{vec}(\bbeta^\top)\|_1$ has a unique solution $\hat{\bbeta} \in \R^{p\times d}$ whose support contained in the true support, i.e. $\mathrm{supp}(\hat{\bbeta}) \subseteq \mathrm{supp}(\bbeta^\star)$, and 
    \begin{equation*}
        \|\mathrm{vec}(\hat{\bbeta}^\top)_{\Stilde} - \mathrm{vec}(\bbeta^{\star \top})_{\Stilde}\|_{\infty} \leq \underbrace{2 (1 + \frac{\nu_x}{8}) \tau / C_e}_{\tilde f(\tau)}.
    \end{equation*}
    \item [(b)] If $\min_{j \in \Stilde} |\mathrm{vec}(\bbeta^{\star \top})_j| > \tilde f(\tau)$, then $\hat{\bbeta}$ recovers the correct signed support.
\end{enumerate}
\end{theorem}

Under sub-Gaussian settings for random errors, where $\alpha = 2,$ the $\ell_{\infty}$ error bound in Theorem~\ref{thm-fixed-sel-consistency} has rate
$\mathcal{O}_p \left\{\max_i \tilde K_i \sqrt{N^{-1}\log (pd)} (1 + a_m)\right\}$,
which differs from the standard lasso rate for homoskedastic linear models \citep{wainwright2009sharp} in two ways.
First, a stronger upper bound assumption on $\vertiii{\Ggamma_{\Stilde \Stilde}^{-1}}_{\infty}$ is required instead of the standard minimum eigenvalue condition on $\Ggamma_{\Stilde \Stilde}$.
Second, fixed-effects selection consistency additionally requires the control of the subject-level covariance estimation error $a_m$, reflecting the stage-wise nature of the proposed method. The rate of $a_m$ follows from the estimation error bounds for the refitted $\hat{\Ssigma}_B$ in Step 3 and $\hat{\Ssigma}_{\epsilon}$ in Step 4, as detailed in Supplementary Material~\ref{appendix-am}. For simplicity, we state the resulting rate in the balanced case where $n_i = n$ for all $i\in[m]$:
$$a_m = \mathcal{O}_p \left\{\sqrt{\frac{d^3}{N}} + \left(\max_{\ell}\vertiii{\Ssigma_B^{\star^{(\ell,\ell)}}}_2 +  \sqrt{\max_{\ell}\vertiii{\Ssigma_B^{\star^{(\ell,\ell)}}}_2} \right) \sqrt{\frac{d^3}{m}} + \max_{\ell} \bigvertiii{\Ssigma_B^{\star^{(\ell,\ell)}}}_2 \sqrt{\frac{s^2}{m}} \right\}.$$
The first term depends on the total sample size $N$ and corresponds to the error rate of an unpenalized estimator of $\Ssigma_\epsilon$ when the random effects $\B_i$'s 
are observed. 
Since $\Ssigma_B$ is unknown and must be estimated in practice, the additional error terms arise, all of which depend on the number of subjects $m$, as discussed in Section~\ref{sec-theory-ran}. 
Finally, the dependence of the first two terms on $d$ can potentially be improved by incorporating structural assumptions on $\Ssigma_\epsilon$ and using a penalized version of \eqref{equ-opt-sigmae}; see \cite{duan2023sparse}. Combining this rate with Theorem~\ref{thm-fixed-sel-consistency} gives the following rate:
{
$$\mathcal{O}_p \left\{\max_i \tilde K_i \sqrt{\frac{\log (pd)}{N}} \left[1 + \left(1 + \sqrt{\max_{\ell}\vertiii{\Ssigma_B^{\star^{(\ell,\ell)}}}_2} \right)^2 \sqrt{\frac{d^3}{m}} + \max_{\ell} \bigvertiii{\Ssigma_B^{\star^{(\ell,\ell)}}}_2 \sqrt{\frac{s^2}{m}} \right]\right\}.$$
}
This rate quantifies the price of using estimated covariance matrices in the fixed-effects selection step. In particular, the procedure still allows $p$ to grow exponentially with $N$, provided that $s$, $d$, and $\max_{\ell}\vertiii{\Ssigma_B^{\star(\ell,\ell)}}_2$ grow sufficiently slowly. Thus, estimating $\Ssigma_i^\star$ inflates the fixed-effects error bound but does not alter the basic selection-consistency mechanism.

\section{Simulation studies}
\label{sec-sim}
This section compares the empirical performance of \texttt{MOMENT}\if1\anon\footnote{An \texttt{R} implementation of our proposed methods, together with the code for reproducible numerical studies, is available at \href{https://github.com/YifanChen3/MOMENT}{https://github.com/YifanChen3/MOMENT}.}\fi{} with several existing methods. Synthetic data are simulated following model \eqref{model_eq}, where the covariates associated with the fixed and random effects are generated as $\x_{ij} \sim \mathcal{N}(\mathbf{0},\I_p)$ and $\z_{ij} \sim \mathcal{N}(\mathbf{0},\I_q)$ respectively. 

Throughout the simulation studies, we fix $d=5$, $p=20$, $m=100$, and $n_i=10$, and consider two values of $q$, representing settings where $q < n_i$ (Section \ref{subsec:low}) and $q > n_i$ (Section \ref{subsec:high}). Each setting is further characterized by a sparsity level $s$, which will be specified in the corresponding subsections. The true coefficient matrix $\bbeta \in \mathbb{R}^{p \times d}$ associated with the fixed effects has the first $5$ rows as $\bbeta_{1,.} = [-2,2,-1,-2,1], \bbeta_{2,.} = [-1,1,2,2,2], \bbeta_{3,.} = [0,1,1,1,2], \bbeta_{4,.} = [1,-2,2,1,-2], \bbeta_{5,.} = [2,-1,-2,2,-1]$, and
all remaining rows set to zero. The random errors are simulated as $\eepsilon_{ij}\stackrel{iid}{\sim} \mathcal{N}(\mathbf{0},\mathbf{\Ssigma_{\epsilon}})$, and the random effects $\B_i$ are simulated as $\text{vec}(\B_i) \stackrel{iid}{\sim} \mathcal{N}(\mathbf{0},\mathbf{\Ssigma_B})$.
To assess the robustness beyond Gaussian random effects, we consider the setting where $\text{vec}(\B_i)$ are simulated from a multivariate Laplace distribution, specifically as $\text{vec}(\B_i) = \bm L \u \sqrt{v}$, where $\Ssigma_B = \bm L\bm L^\top$ denotes its Cholesky decomposition, $\u\sim\mathcal{N}(\bm 0, \I_{dq})$, and $v\sim Exp(1)$, with $v$ and $\u$ being independent.
We focus on two different designs for $\Ssigma_B$, as in \eqref{eq:covB}, and $\Ssigma_{\epsilon}$:

\textbf{Design 1 (\textit{Independent responses})}: We let $\Ssigma_B$ be a block-diagonal matrix, where its $\ell$-th diagonal block matrix $\Ssigma_B^{(\ell,\ell)} \in \mathbb{R}^{q \times q}$ is specified as $(\Ssigma_B^{(\ell,\ell)})_{jk}=\rho_{\ell} ^{|j-k|}$ for $j, k \in [s_{\ell}]$ and $\rho_{\ell}=0.4+0.1\times \ell$ for $\ell \in [d]$. The error covariance is also a diagonal matrix as $\Ssigma_{\epsilon}=\text{diag}(1,0.9,0.8,0.7,0.6) \in \mathbb{R}^{d \times d}$.

\textbf{Design 2 (\textit{Correlated responses})}: For each $r,u \in [d]$, the $(r, u)$-th block submatrix of $\Ssigma_B$ is specified as $(\Ssigma_B^{(r,u)})_{j,k}=0.5^{|j-k|}$ for $j, k \in [s_{\ell}]$. The error covariance matrix is specified as $(\Ssigma_{\epsilon})_{j,k}=0.75^{|j-k|}$ for $j,k \in [d]$.

Under Design~1, the block-diagonal structure of $\Ssigma_B$ and $\Ssigma_{\epsilon}$ implies that different responses are independent. In this case, methods that fit each response separately correctly specify the model and are expected to achieve advantageous performance. In contrast, under Design~2, the presence of correlations between responses means that methods that ignore this structure may fail to capture important dependencies, leading to suboptimal performance. The proposed \texttt{MOMENT} method, which is designed to leverage such correlations, is expected to outperform competing methods in this setting.
To further demonstrate the importance of accounting for cross-response correlations, and for fair comparison with the all the competing univariate methods, we also consider a variant of our method, \texttt{MOMENT (marginal)}, which fits and refits each response separately as a univariate mixed model and assembles the resulting response-specific covariance estimates into block diagonal covariance matrices, thereby setting all cross-response covariance components to zero. This is the same response-by-response fitting strategy used by all competing methods in our simulations, since they are univariate methods.

\subsection{$q < n_i$ setting}
\label{subsec:low}
We begin with a regime where $q < n_i$, under which several existing univariate methods can be applied in a response-by-response manner. Specifically, we set $q=8$ and $s_{\ell}=4$ for each response $\ell \in [d]$. We compare \texttt{MOMENT} with the methods of \cite{ahn2012moment}, \cite{peng2012model}, and \cite{hui2017joint}. 

\begin{table}[!htbp]
  \centering
  \spacingset{1}
  \footnotesize
  \setlength{\tabcolsep}{3pt}
  \renewcommand{\arraystretch}{0.88}
  \caption{Performance comparison in scenario where $q < n_i$, evaluated in terms of Frobenius norm (F-norm) of the difference between the estimates and the true parameters, F1 score for the selection of relevant fixed effects and random effects, and computational time (Time). Reported values are averaged over 50 replications, with standard deviations shown in parentheses. In each configuration, the best performance is highlighted in bold.}
  \label{table-sim-low}
  \begin{tabular}{@{}lllrrrr@{}}
    \toprule
           &   &     & \multicolumn{2}{c}{Independent Responses} & \multicolumn{2}{c}{Correlated Responses} \\
    \cmidrule(lr){4-5} \cmidrule(lr){6-7}
                 &         &           & \multicolumn{1}{c}{Gaussian} & \multicolumn{1}{c}{Laplace}  & \multicolumn{1}{c}{Gaussian} & \multicolumn{1}{c}{Laplace} \\
    \midrule
    \multirow{5}{*}{\makecell[l]{\texttt{MOMENT}}}
      & \multirow{2}{*}{$\hat{\Ssigma}_B$}
        & F-norm  & 3.29(0.35) & 4.68(1.16)  & \textbf{3.66(0.91)} & \textbf{5.52(2.09)}   \\   
      & 
        & F1 score  & 0.995(0.02) & 0.95(0.10)  & \textbf{1(0)} & 0.98(0.05)     \\
      & \multirow{2}{*}{$\hat{\bbeta}$}
        & F-norm    & 0.20(0.03) & 0.22(0.05)  & \textbf{0.19(0.05)} & 0.21(0.05)    \\
      & 
        & F1 score  & \textbf{1(0)} & \textbf{1(0)}  & \textbf{1(0)} & \textbf{1(0)}    \\
    & \multirow{1}{*}{Time}
    & Seconds & 101(14) & 104(19)  & 107(20) & 112(21)  \\
    \midrule
    \multirow{5}{*}{\makecell[l]{\texttt{MOMENT}\\\texttt{(marginal)}}}
      & \multirow{2}{*}{$\hat{\Ssigma}_B$}
        & F-norm  & 1.75(0.26) & 2.55(0.77)  & 10.89(0.07) & 11.07(0.25)   \\   
      & 
        & F1 score  & \textbf{0.997(0.009)} & \textbf{0.99(0.02)}  & 0.999(0.007) & 0.98(0.05)     \\
      & \multirow{2}{*}{$\hat{\bbeta}$}
        & F-norm    & 0.19(0.03) & 0.20(0.03)  & 0.21(0.05) & 0.23(0.05)    \\
      & 
        & F1 score  & 0.999(0.003) & \textbf{1(0)}  & \textbf{1(0)} & \textbf{1(0)}    \\
    & \multirow{1}{*}{Time}
    & Seconds & \textbf{72(8)} & \textbf{70(5)}  & \textbf{65(9)} & \textbf{67(8)}  \\
    \midrule
    \multirow{5}{*}{\citet{ahn2012moment}}
      & \multirow{2}{*}{$\hat{\Ssigma}_B$}
        & F-norm    & 1.96(0.35) & 2.62(0.55)  & 10.91(0.08) & 11.21(0.71)     \\
      & 
        & F1 score  & 0.97(0.03) & 0.96(0.03)  & 0.97(0.05) & 0.96(0.06)    \\
      & \multirow{2}{*}{$\hat{\bbeta}$}
        & F-norm    & \textbf{0.18(0.03)} & \textbf{0.19(0.03)}  & 0.20(0.05) & \textbf{0.20(0.07)}    \\
      & 
        & F1 score  & 0.98(0.03) & 0.98(0.03)  & 0.99(0.01) & 0.97(0.10)    \\
        & \multirow{1}{*}{Time}
    & Seconds & 7088(492) & 6712(837)  & 6199(1217) & 6078(1780)   \\
    \midrule
    \multirow{5}{*}{\citet{peng2012model}}
      & \multirow{2}{*}{$\hat{\Ssigma}_B$}
        & F-norm    & 1.56(0.72) & 1.89(0.62)  & 10.85(0.12) & 10.94(0.13)    \\
      & 
        & F1 score  & 0.99(0.03) & 0.99(0.04)  & 0.99(0.05) & \textbf{0.98(0.04)}  \\
      & \multirow{2}{*}{$\hat{\bbeta}$}
        & F-norm    & 0.22(0.07) & 0.22(0.05)  & 0.25(0.07) & 0.25(0.08)    \\
      & 
        & F1 score  & 0.89(0.10) & 0.87(0.09)  & 0.88(0.11) & 0.87(0.12)    \\
        & \multirow{1}{*}{Time}
    & Seconds & 842(109) & 850(136)  & 831(134) & 839(139)   \\
        \midrule
    \multirow{5}{*}{\citet{hui2017joint}}
      & \multirow{2}{*}{$\hat{\Ssigma}_B$}
        & F-norm    & \textbf{1.33(0.26)} & \textbf{1.76(0.38)}  & 10.82(0.03) & 10.91(0.11)    \\
      & 
        & F1 score  & 0.96(0.04) & 0.95(0.05)  & 0.95(0.06) & 0.94(0.08)   \\
      & \multirow{2}{*}{$\hat{\bbeta}$}
        & F-norm    & \textbf{0.18(0.03)} & \textbf{0.19(0.03)}  & 0.20(0.05) & 0.21(0.05)    \\
      & 
        & F1 score  & 0.95(0.06) & 0.94(0.06)  & 0.94(0.07) & 0.93(0.09)    \\
        & \multirow{1}{*}{Time}
    & Seconds & 481(46) & 492(30)  & 501(53) & 531(53)   \\
    \bottomrule
  \end{tabular}
\end{table}

Table \ref{table-sim-low} shows a clear contrast between the two covariance designs. Under independent responses, competing methods attain better estimation accuracy for $\Ssigma_B$, which is unsurprising since joint modeling offers no additional information in this setting. When \texttt{MOMENT} is fitted marginally, its performance improves substantially: it outperforms \citet{ahn2012moment} and remains only slightly behind \citet{peng2012model} and \citet{hui2017joint} in covariance estimation. By contrast, under correlated responses, \texttt{MOMENT} outperforms all competing methods, highlighting the benefit of explicitly modeling cross-response dependence. 
Across all configurations, \texttt{MOMENT} achieves nearly perfect support recovery for both fixed and random effects, and all methods yield similar accuracy for estimating $\bbeta$. Computationally, \texttt{MOMENT} is also highly efficient. 
Overall, the results indicate that \texttt{MOMENT} offers a strong balance of estimation accuracy, selection performance, and computational efficiency in this setting.

\subsection{$q > n_i$ setting}
\label{subsec:high}
We next examine the setting $q > n_i$, where the competing methods in the previous subsection fail in the sense that the covariance estimates are unidentifiable. Specifically, we set $q = 20$ and $s_{\ell} = 8$ for each response $\ell \in [d]$. Recent methods for this regime often require restrictive structural assumptions, such as $\x_{ij} = \z_{ij}$ \citep{thompson2025scalable} or $\z_{ij}$ being a subset of $\x_{ij}$ \citep{heiling2024glmmpen}. To enable comparison, we modify the simulation setting by setting $\x_{ij} = \z_{ij}$. In Table~\ref{table-sim-high}, we report the comparison between our method and \citet{thompson2025scalable}. We do not include \citet{heiling2024glmmpen} due to its intractable computational cost for estimating the full large-dimensional covariance matrix $\Ssigma_B$. 

\begin{table}[!htbp]
  \centering
  \spacingset{1}
  \footnotesize
  \setlength{\tabcolsep}{3pt}
  \renewcommand{\arraystretch}{0.9}
  \caption{Performance comparison in scenario where $q > n_i$, evaluated in terms of Frobenius norm (F-norm) of the difference between the estimates and the true parameters, F1 score for the selection of relevant fixed effects and random effects, and computational time (Time). Reported values are averaged over 50 replications, with standard deviations shown in parentheses. In each configuration, the best performance is highlighted in bold.}
  \label{table-sim-high}
  \begin{tabular}{@{}lllrrrr@{}}
    \toprule
           &   &     & \multicolumn{2}{c}{Independent Responses} & \multicolumn{2}{c}{Correlated Responses} \\
    \cmidrule(lr){4-5} \cmidrule(lr){6-7}
                 &         &           & \multicolumn{1}{c}{Gaussian} & \multicolumn{1}{c}{Laplace}  & \multicolumn{1}{c}{Gaussian} & \multicolumn{1}{c}{Laplace} \\
    \midrule
    \multirow{5}{*}{\makecell[l]{\texttt{MOMENT}}}
      & \multirow{2}{*}{$\hat{\Ssigma}_B$}
        & F-norm  & 8.11(0.63) & 11.83(3.40)  & \textbf{10.27(2.51)} & \textbf{14.20(3.46)}   \\
      & 
        & F1 score  & 0.93(0.07) & 0.82(0.11)  & 0.89(0.14) & 0.80(0.18)     \\
      & \multirow{2}{*}{$\hat{\bbeta}$}
        & F-norm    & 0.80(0.38) & 0.85(0.54)  & 0.68(0.46) & 0.94(0.61)    \\
      & 
        & F1 score  & \textbf{0.994(0.01)} & 0.99(0.04)  & \textbf{0.997(0.02)} & \textbf{0.993(0.02)} \\
        & \multirow{1}{*}{Time}
    & Seconds & 270(18) & 238(58)  & 217(47) & 165(34)   \\
    \midrule
    \multirow{5}{*}{\makecell[l]{\texttt{MOMENT}\\\texttt{(marginal)}}}
      & \multirow{2}{*}{$\hat{\Ssigma}_B$}
        & F-norm  & \textbf{4.26(0.47)} & \textbf{6.74(2.12)}  & 16.28(0.13) & 16.82(0.59)   \\   
      & 
        & F1 score  & 0.96(0.02) & 0.90(0.04)  & 0.95(0.04) & 0.89(0.07)     \\
      & \multirow{2}{*}{$\hat{\bbeta}$}
        & F-norm    & \textbf{0.69(0.32)} & 0.85(0.60)  & \textbf{0.57(0.16)} & \textbf{0.66(0.31)}    \\
      & 
        & F1 score  & 0.992(0.02) & \textbf{0.99(0.03)}  & 0.996(0.02) & \textbf{0.993(0.02)}    \\
    & \multirow{1}{*}{Time}
    & Seconds & \textbf{117(4)} & \textbf{121(6)}  & 113(12) & 107(11)  \\
    \midrule
    \multirow{5}{*}{\makecell[l]{\citeauthor{thompson2025scalable}\\(\citeyear{thompson2025scalable})}}
      & \multirow{2}{*}{$\hat{\Ssigma}_B$}
        & F-norm    & 9.10(0.23) & 9.20(0.47)  & 16.52(0.03) & 16.55(0.04)   \\
      & 
        & F1 score  & \textbf{0.997(0.005)} & \textbf{0.996(0.007)} & \textbf{0.996(0.007)} & \textbf{0.996(0.007)}\\
      & \multirow{2}{*}{$\hat{\bbeta}$}
        & F-norm    & 0.73(0.11) & \textbf{0.73(0.12)}  & 0.71(0.20) & 0.71(0.17)    \\
      & 
        & F1 score  & 0.74(0.01) & 0.74(0.01)  & 0.74(0.02) & 0.74(0.01)    \\
        & \multirow{1}{*}{Time}
    & Seconds & 119(14) & 121(17)  & \textbf{100(18)} & \textbf{97(19)}   \\
    \bottomrule
  \end{tabular}
\end{table}

The results in Table~\ref{table-sim-high} highlight the trade-off between scalability and covariance modeling. The method of \citet{thompson2025scalable} estimates only the diagonal elements of $\Ssigma_B$, which is sufficient for random-effects selection but does not recover the full covariance structure. As a result, it performs well in support recovery but is limited in covariance estimation accuracy. By contrast, \texttt{MOMENT} estimates the full matrix $\Ssigma_B$ and therefore captures its structure more accurately. Although this richer modeling increases computational cost, the runtime of \texttt{MOMENT} remains comparable to that of \citet{thompson2025scalable}. These results suggest that \texttt{MOMENT} offers a useful balance between statistical accuracy and computational feasibility in moderately large-dimensional settings where full covariance estimation is still of interest.

\subsection{Sensitivity analysis}
The previous experiment, which sets $\x_{ij}=\z_{ij}$, is the ``homecourt'' of \citet{thompson2025scalable}. To assess a more general and practically relevant scenario, we now consider the case in which each entry of $\x_{ij}$ and $\z_{ij}$ is generated independently from a standard Gaussian distribution, and therefore differ substantially. \texttt{MOMENT} can accommodate this setting directly, whereas the method of \citet{thompson2025scalable} must instead practically construct an augmented design matrix, which requires that a concatenated covariate vector $\bm u_{ij} = (\x_{ij}^\top,\z_{ij}^\top)^\top \in \mathbb{R}^{(p+q)}$ must be formed for each observation.

We focus on the Gaussian random-effects setting, and examine two dimensional regimes: $(q,s_{\ell})=(8,4)$ and $(q,s_{\ell})=(15,6)$ for all $\ell \in [d]$. We consider both the independent-response and correlated-response covariance designs, and keep other parameters the same as in previous simulation settings.

Table~\ref{table-sim-sensitivity} shows that \texttt{MOMENT} is particularly advantageous when the covariates for the fixed effects and random effects differ. In the independent-response setting, \texttt{MOMENT (marginal)} achieves the best estimation accuracy for both $\hat{\Ssigma}_B$ and $\hat{\bbeta}$ while maintaining nearly perfect support recovery. In the correlated-response setting, the advantage of the  \texttt{MOMENT} estimator is even more pronounced for estimating $\Ssigma_B$. Although all methods attain near-perfect support recovery for $\hat{\Ssigma}_B$, \citet{thompson2025scalable} is slightly better in F1 score because it targets only the diagonal support, effectively solving a much smaller-dimensional problem. In terms of runtime, \texttt{MOMENT (marginal)} is consistently the fastest method across all configurations. This likely reflects the fact that the augmented design matrix used by \citet{thompson2025scalable} increases the effective problem dimension from $p$ or $q$ to $p+q$, thereby increasing computational cost. Taken together, these results demonstrate that \texttt{MOMENT} retains strong estimation and selection performance in settings that fall outside the structural assumptions required by existing competitors.

\begin{table}[!htbp]
  \centering
  \spacingset{1}
  \footnotesize
  \setlength{\tabcolsep}{3pt}
  \renewcommand{\arraystretch}{0.9}
  \caption{Performance comparison in the setting where $\x_{ij}$ and $\z_{ij}$ substantially differ, evaluated in terms of Frobenius norm (F-norm) of the difference between the estimates and the true parameters, F1 score for the selection of relevant fixed effects and random effects, and computational time (Time). Reported values are averaged over 50 replications, with standard deviations shown in parentheses. In each configuration, the best performance is highlighted in bold.}
  \label{table-sim-sensitivity}
  \begin{tabular}{@{}lllrrrr@{}}
    \toprule
           &   &     & \multicolumn{2}{c}{Independent Responses} & \multicolumn{2}{c}{Correlated Responses} \\
    \cmidrule(lr){4-5} \cmidrule(lr){6-7}
                 &         &           & \multicolumn{1}{c}{$q=8,s=4$} & \multicolumn{1}{c}{$q=15,s=6$}  & \multicolumn{1}{c}{$q=8,s=4$} & \multicolumn{1}{c}{$q=15,s=6$} \\
    \midrule
    \multirow{5}{*}{\makecell[l]{\texttt{MOMENT}}}
      & \multirow{2}{*}{$\hat{\Ssigma}_B$}
        & F-norm  & 3.20(0.32) & 5.48(0.62)  & \textbf{3.43(0.80)} & \textbf{6.36(1.42)}   \\
      & 
        & F1 score  & 0.99(0.02) & 0.96(0.08)  & 0.999(0.008) & 0.98(0.04)     \\
      & \multirow{2}{*}{$\hat{\bbeta}$}
        & F-norm    & 0.26(0.04) & 0.31(0.06)  & \textbf{0.25(0.08)} & 0.30(0.13)    \\
      & 
        & F1 score  & \textbf{1(0)} & \textbf{1(0)}  & \textbf{1(0)} & \textbf{1(0)} \\
        & \multirow{1}{*}{Time}
    & Seconds & 116(17) & 148(10)  & 105(16) & 128(24)   \\
    \midrule
    \multirow{5}{*}{\makecell[l]{\texttt{MOMENT}\\\texttt{(marginal)}}}
      & \multirow{2}{*}{$\hat{\Ssigma}_B$}
        & F-norm  & \textbf{1.73(0.24)} & \textbf{2.89(0.40)}  & 10.87(0.05) & 13.79(0.10)   \\   
      & 
        & F1 score  & \textbf{0.997(0.009)} & 0.98(0.02)  & 0.995(0.02) & 0.98(0.03)     \\
      & \multirow{2}{*}{$\hat{\bbeta}$}
        & F-norm    & \textbf{0.25(0.04)} & \textbf{0.28(0.05)}  & 0.26(0.08) & 0.31(0.12)    \\
      & 
        & F1 score  & \textbf{1(0)} & \textbf{1(0)}  & \textbf{1(0)} & 0.999(0.003)    \\
    & \multirow{1}{*}{Time}
    & Seconds & \textbf{69(3)} & \textbf{85(3)}  & \textbf{65(8)} & \textbf{80(10)}  \\
    \midrule
    \multirow{5}{*}{\makecell[l]{\citeauthor{thompson2025scalable}\\(\citeyear{thompson2025scalable})}}
      & \multirow{2}{*}{$\hat{\Ssigma}_B$}
        & F-norm    & 5.24(0.17) & 7.31(0.19)  & 11.19(0.02) & 14.10(0.02)   \\
      & 
        & F1 score  & 0.996(0.009) & \textbf{0.999(0.004)} & \textbf{0.999(0.003)} & \textbf{0.997(0.008)}\\
      & \multirow{2}{*}{$\hat{\bbeta}$}
        & F-norm    & 0.28(0.06) & 0.29(0.05)  & 0.25(0.10) & \textbf{0.28(0.11)}    \\
      & 
        & F1 score  & 0.98(0.03) & 0.99(0.02)  & 0.993(0.01) & 0.992(0.02)    \\
        & \multirow{1}{*}{Time}
    & Seconds & 89(10) & 185(26)  & 78(17) & 161(28)   \\
    \bottomrule
  \end{tabular}
\end{table}

\section{Hemodialysis clinical and laboratory data analysis}
\label{sec-apply}
We apply \texttt{MOMENT} to a multivariate longitudinal dataset from patients receiving maintenance hemodialysis. Patients undergoing hemodialysis are repeatedly monitored using clinical measurements, dialysis prescription variables, and laboratory markers. These variables are scientifically linked: fluid accumulation between dialysis sessions is related to ultrafiltration during treatment, anemia management involves both erythropoietin and iron supplementation, and inflammatory and nutritional markers often evolve jointly over time. We apply the proposed \texttt{MOMENT} method to explore correlation structures between clinical and laboratory variables.

The original data consist of monthly patient-level summaries over the first five years (months 0 through 60) after initiation of hemodialysis. To illustrate \texttt{MOMENT}’s ability to handle settings with a large number of responses or random effects but a small number of repeated measurements per subject, we selected patients with between 6 and 15 complete monthly observations and retained all available months for each selected patient. The resulting dataset contains $m=2159$ patients and $N=21220$ patient-month observations, with an average of $9.83$ observations per patient. We consider $d=20$ clinical and laboratory variables, using their monthly averages as multivariate responses. Beyond examining the correlation structure among these variables, we study their population-level trajectories and associations with covariates, including age, height, race, ethnicity, gender, diabetes mellitus, and vascular access type. Definitions and summary statistics of all variables are reported in Table~\ref{tab:summary_m2159} of the Supplementary Material~\ref{appendix-data}.


We fit model \eqref{model_eq} with $\z_{ij}=\bm 1$ and
$\x_{ij}
    = (1,
s_1(t_{ij}),\ldots,s_7(t_{ij}),
\text{age}_i,\text{height}_i,
\bm 1_{(\text{race}_i=1)},$ \newline
$\bm 1_{(\text{male}_i=1)},
\bm 1_{(\text{diabetic}_i=1)},
\bm 1_{(\text{hispanic}_i=1)},
\text{AVF}_{ij},
\text{AVG}_{ij}
)^{\top},
$
where $t_{ij}$ is the month of the $j$th observation from patient $i$, $s_1,\ldots,s_7$ are natural cubic spline basis functions of month, constructed using six internal knots placed at the empirical quantiles $\{1/7,\ldots,6/7\}$ of the observed months. The access variables $\text{AVF}_{ij}$ and $\text{AVG}_{ij}$ are the monthly proportions of treatments using arteriovenous fistula and arteriovenous graft access, respectively; catheter access is set as the reference category. Thus, in the notation of model \eqref{model_eq}, $p=16$, $q=1$, and $d=20$, and $\mathbf B_i\in\mathbb R^{1\times 20}$ contains one response-specific random intercept for each patient. The covariance matrix $\boldsymbol\Sigma_B\in\mathbb R^{20\times 20}$ describes patient-level heterogeneity among the response-specific random intercepts, whereas $\boldsymbol\Sigma_\epsilon\in\mathbb R^{20\times 20}$ describes monthly within-patient residual dependence. We imposed an $\ell_1$ penalty on all covariates, while leaving the intercept and spline basis functions unpenalized. The tuning parameters $\lambda$ and $\tau$ were selected via 5-fold cross-validation.

\begin{figure}[h]
    \centering
    \spacingset{1.2}
    \includegraphics[width=\linewidth]{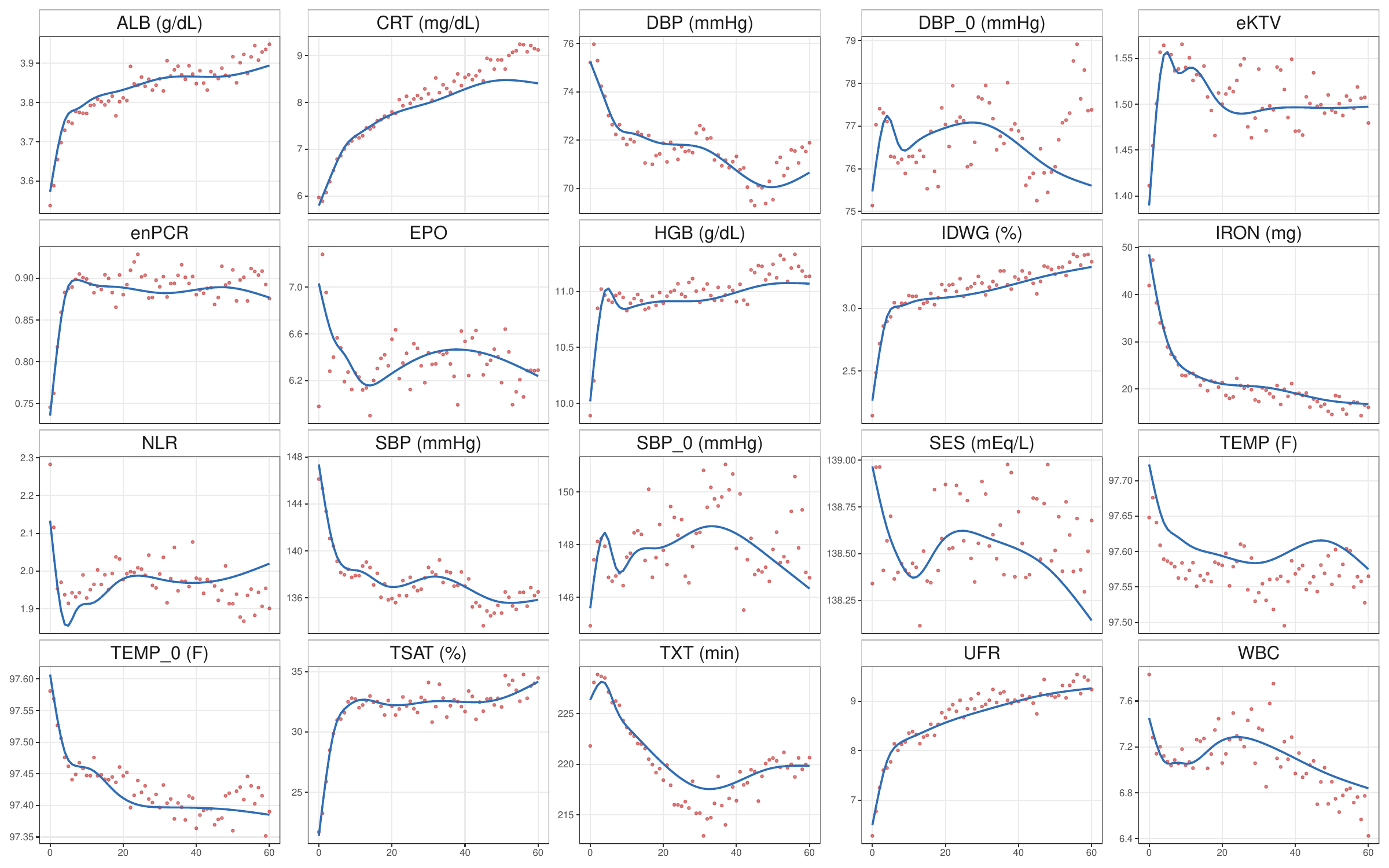}
    \caption{Estimated population trajectories for the $20$ response variables. Red points represent empirical monthly averages of the observed responses, and blue curves represent the estimated trajectories.}
    \label{fig:hemodialysis-fix-trajectories}
\end{figure}

Figure~\ref{fig:hemodialysis-fix-trajectories} shows clear nonlinear population-level changes over the first 60 months of hemodialysis. Several variables change rapidly soon after treatment initiation and then stabilize, including \texttt{HGB}, \texttt{ALB}, \texttt{IDWG}, \texttt{UFR}, \texttt{enPCR}, \texttt{eKTV}, \texttt{TSAT}, and \texttt{CRT}, while post-treatment \texttt{SBP} and \texttt{DBP}, \texttt{EPO}, \texttt{IRON}, \texttt{TEMP}, and \texttt{WBC} generally decrease over time. These patterns suggest that early dialysis management is accompanied by substantial adjustment in anemia-related markers, volume control, nutritional status, dialysis adequacy, and blood pressure, after which the average trajectories become more stable.

\begin{figure}[H]
    \centering
    \spacingset{1.2}
    \includegraphics[width=\linewidth]{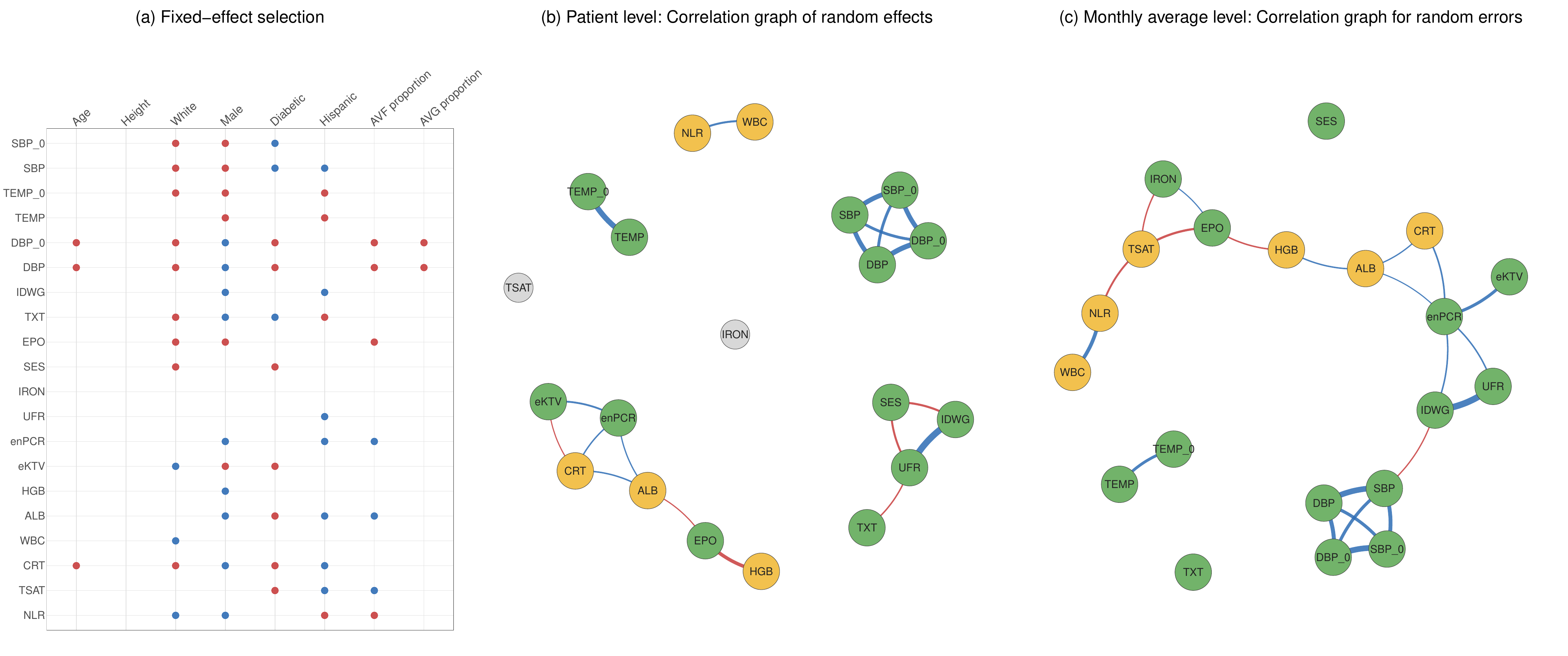}
    \caption{Summary of fixed-effect selection and dependence structure in the hemodialysis multiresponse linear mixed model. Panel (a) shows selected demographic and vascular-access fixed effects for each response; blue and red points indicate positive and negative selected coefficients, respectively. Panel (b) shows the patient-level correlation graph estimated from $\hat \Ssigma_B$, and Panel (c) shows the monthly within-patient residual correlation graph estimated from $\hat \Ssigma_\epsilon$. Blue edges indicate positive correlations and red edges indicate negative correlations, with thicker edges corresponding to stronger absolute correlations. Edges with absolute correlation below 0.2 are omitted. Node colors indicate variable type, with green nodes denoting clinical variables and yellow nodes denoting laboratory variables. In Panel (b), this color scheme is overridden for grey nodes, which correspond to responses whose random-intercept variance is estimated as zero.}
    \label{fig:hemodialysis-selected-fix-random}
\end{figure}

Panel (a) in Figure~\ref{fig:hemodialysis-selected-fix-random} shows that demographic and access covariates contribute heterogeneously across the 20 hemodialysis responses. The covariates \texttt{race}, \texttt{sex}, \texttt{diabetes}, and \texttt{Hispanic} are selected for multiple outcomes, suggesting that part of the between-patient variation in clinical and laboratory measurements is explained by baseline patient characteristics rather than random effects alone. In contrast, \texttt{height} is not selected for any response, and \texttt{age} is selected only for \texttt{DBP\_0}, \texttt{DBP}, and \texttt{CRT}. Access variables are selected for a smaller subset of responses, most notably for blood pressure, anemia-related, nutritional, and inflammatory markers, indicating that vascular access type may capture clinically relevant differences in treatment delivery or patient condition.

Panels (b) and (c) in Figure~\ref{fig:hemodialysis-selected-fix-random} summarize the estimated patient-level and monthly-level dependence structures. The associations identified by \texttt{MOMENT} are clinically meaningful and readily interpretable. At both the patient-level and monthly residual-level, \texttt{IDWG} and \texttt{UFR} are strongly positively correlated, reflecting the relationship between interdialytic fluid gain and fluid removal during dialysis. This agrees with clinical knowledge that \texttt{IDWG} is closely tied to volume management and blood pressure burden in hemodialysis patients~\citep{lopez2005interdialytic,ipema2016causes}. The pre- and post-treatment measurements of the same vital signs are also strongly positively correlated, including \texttt{SBP\_0}--\texttt{SBP}, \texttt{DBP\_0}--\texttt{DBP}, and \texttt{TEMP\_0}--\texttt{TEMP}. In addition, \texttt{WBC} and \texttt{NLR} are positively correlated at both the patient-level and the monthly-average level, consistent with their shared association with inflammatory status. Finally, \texttt{EPO} and \texttt{HGB} are negatively correlated, particularly at the patient level, consistent with the clinical interpretation that patients with lower hemoglobin often require higher erythropoietin doses for anemia management.

Another distinctive finding is that \texttt{IRON} and \texttt{TSAT} are not selected as random-effect-bearing responses in the patient-level covariance component. This should not be interpreted as evidence that iron management is clinically unimportant. Rather, in the fitted multivariate model, \texttt{MOMENT} does not identify additional patient-level random-effect components for these two variables after accounting for the other responses. This is compatible with their clinical roles: \texttt{IRON} is a treatment-management variable whose variation may reflect time-varying treatment decisions, while \texttt{TSAT} is commonly used for monitoring iron status but can exhibit substantial biological and analytical variability~\citep{bailie2013variation,robinson2017evaluating,van2010analytical}. Thus, their relevance may appear primarily through contemporaneous monthly associations with other anemia-related variables rather than through persistent between-patient heterogeneity. This illustrates that the proposed method separates patient-level heterogeneity from within-patient monthly co-movement, a distinction that is difficult to obtain from separate univariate mixed models or from an unpenalized dense covariance estimate.

An additional analysis of the SANAD data is provided in Supplementary Material~\ref{appendix-SANAD}.

\section{Discussion}
\label{sec-discuss}
In this paper, we proposed a unified framework for simultaneous selection and estimation of both fixed and random effects in multivariate mixed-effects models. The proposed procedure is based on moment conditions rather than likelihood evaluation, which makes it computationally attractive and flexible with respect to distributional assumptions. Our theoretical analysis establishes sign-selection consistency for both fixed- and random-effects under joint sub-Weibull tail conditions and explicitly accommodates dependence induced by repeated measurements. The numerical studies further clarify the operating characteristics of the method: when the responses are independent, response-wise methods remain competitive because they are correctly specified, whereas under cross-response dependence, the proposed joint procedure yields clear gains in estimating $\Ssigma_B$ while maintaining accurate recovery of the fixed-effects support.

Several directions for future research are worth pursuing. One is to incorporate more robust loss functions, such as Huber-type or absolute-deviation criteria, into the estimation procedure while preserving its favorable computational properties. Another is to develop tuning-free selection strategies that are fully aligned with the moment-based nature of the framework. A further natural extension is to generalize the methodology to generalized linear mixed-effects models (GLMMs), thereby broadening its applicability to a wider range of longitudinal and clustered data settings.


\section{Disclosure statement}\label{disclosure-statement}

The authors report that there are no conflicts of interest.

\section{Data Availability Statement}\label{data-availability-statement}

The hemodialysis data are available on request from the authors.

\bibliography{bibliography}

\newpage

\phantomsection\label{supplementary-material}
\bigskip

\begin{center}

{\large\bf SUPPLEMENTARY MATERIAL}

\end{center}

The supplementary material contains additional hemodialysis data summaries, the SANAD data analysis, algorithmic details, and theoretical proofs.

\appendix
\section{Comparison table for variable selection methods in LMMs}
\label{appendix-comparison-table}
\begin{table}[htbp]
\centering
\spacingset{1.2}
\scriptsize
\setlength{\tabcolsep}{2pt}
\renewcommand{\arraystretch}{0.95}
\caption{Comparison of variable selection methods in mixed-effects models.}
\label{table-literature}
\begin{tabular}{@{}
    l 
    >{\centering\arraybackslash}m{1.4cm} 
    *{2}{>{\centering\arraybackslash}m{2.2cm}} 
    >{\centering\arraybackslash}m{6.68cm} 
@{}}
\toprule
 & Type & Multi-responses & Large number of random effects & Additional structure assumptions \\
\midrule
\cite{chen2003random} & Likelihood & \xmark & \xmark &   \\
\cite{kinney2007fixed} & Likelihood & \xmark & \xmark &  \\
\cite{bondell2010joint} & Likelihood & \xmark & \xmark &   \\
\cite{ibrahim2011fixed} & Likelihood & \xmark & \xmark &  \\
\cite{peng2012model} & Moment & \xmark & \xmark &   \\
\cite{ahn2012moment} & Moment & \xmark & \xmark &   \\
\cite{hui2017hierarchical} & Likelihood & \xmark & \xmark &   \\
\cite{hui2017joint} & Likelihood & \xmark & \xmark &   \\
\cite{fan2012variable} & Likelihood & \xmark & \cmark & On the proxy matrix  \\
\cite{li2018doubly} & Likelihood & \xmark & \cmark & Hierarchy \\
\cite{heiling2024glmmpen} & Likelihood & \xmark & \cmark & Hierarchy, and diagonal covariance of random effects for scalability\\
\cite{sholokhov2024relaxation} & Likelihood & \xmark & \cmark & Diagonal covariance of random effects and known covariance of random errors  \\
\cite{thompson2025scalable} & Likelihood & \xmark & \cmark & Diagonal covariance of random effects and identical design matrix for fixed and random effects  \\
\textbf{Our Method} & Moment & \cmark & \cmark &   \\
\bottomrule
\end{tabular}
\end{table}

\section{Summary statistics table for hemodialysis data}
\label{appendix-data}

The hemodialysis analysis dataset consists of monthly records from 2,159 patients receiving maintenance hemodialysis during the first five years after treatment initiation, restricted to patients with 6 to 15 complete monthly observations. The dataset contains 21,220 patient-month observations. Our analysis includes monthly averages of the following 20 clinical and laboratory variables as multivariate responses: pre- and post-dialysis systolic blood pressure (\texttt{SBP\_0} and \texttt{SBP}), pre- and post-dialysis diastolic blood pressure (\texttt{DBP\_0} and \texttt{DBP}), pre- and post-dialysis temperature (\texttt{TEMP\_0} and \texttt{TEMP}), interdialytic weight gain (\texttt{IDWG}, expressed as percentage weight gain between dialysis sessions), treatment time (\texttt{TXT}, in minutes), erythropoietin dose (\texttt{EPO}, analyzed on the fourth-root scale following common practice), serum sodium (\texttt{SES}), iron dose (\texttt{IRON}), ultrafiltration rate (\texttt{UFR}, the rate of fluid removal during dialysis), equilibrated normalized protein catabolic rate (\texttt{enPCR}, a marker of dietary protein intake and nutritional status), equilibrated Kt/V (\texttt{eKTV}, a dialysis adequacy measure), hemoglobin (\texttt{HGB}), albumin (\texttt{ALB}), white blood cell count (\texttt{WBC}), creatinine (\texttt{CRT}), transferrin saturation (\texttt{TSAT}, a marker of iron availability), and neutrophil-to-lymphocyte ratio (\texttt{NLR}, an inflammatory marker). These repeated, clinically relevant measurements provide a multivariate longitudinal setting for studying both persistent between-patient heterogeneity and within-patient monthly variation in hemodialysis care. Covariates include demographic and treatment characteristics, namely age at dialysis initiation (\texttt{age}), height (\texttt{height}), sex (\texttt{male}), race (\texttt{race}), ethnicity (\texttt{hispanic}), diabetes status (\texttt{diabetic}), month since dialysis initiation (\texttt{month}), monthly treatment count (\texttt{event\_txtcount}), and vascular access type. The access covariates are monthly average proportions of treatments using arteriovenous fistula (\texttt{ACCESS\_AVF\_avg}), arteriovenous graft (\texttt{ACCESS\_AVG\_avg}), or catheter (\texttt{ACCESS\_CATH\_avg}) access. Summary statistics for all variables are presented in Table~\ref{tab:summary_m2159}.

\begingroup
\spacingset{1.2}
\scriptsize
\setlength{\tabcolsep}{2pt}
\renewcommand{\arraystretch}{0.92}
\setlength{\LTleft}{0pt}
\setlength{\LTright}{0pt}
\setlength{\LTcapwidth}{5.5in}
\begin{longtable}{@{}>{\raggedright\arraybackslash}p{1.70in}>{\raggedright\arraybackslash}p{0.95in}>{\raggedleft\arraybackslash}p{1.50in}>{\raggedleft\arraybackslash}p{1.75in}@{}}

\multicolumn{4}{@{}p{6in}@{}}{\scriptsize \textbf{Table~\thetable}. Summary statistics for the hemodialysis analysis dataset. Patient-level variables are summarized once per patient; patient-month variables are summarized across all observed monthly records. An asterisk (*) indicates variables used as multivariate responses in the hemodialysis analysis. \label{tab:summary_m2159}}\\
\addlinespace[2pt]
\toprule
Variable description & Variable name & Mean (SD) & Median (Q1, Q3) \\
\midrule
\endfirsthead
\multicolumn{4}{@{}l}{\textbf{Table~\thetable} (continued). Summary statistics for the hemodialysis analysis dataset.}\\

\midrule
\endhead
\midrule
\multicolumn{4}{r}{\footnotesize Continued on next page}\\
\endfoot
\bottomrule
\endlastfoot
\multicolumn{4}{l}{\textit{Panel A: Patient-level baseline and design variables ($m=2{,}159$ patients)}}\\
Monthly observations per patient & \texttt{n\_months} & 9.829 (2.792) & 10.00 (7.000, 12.00) \\
Age, years & \texttt{age} & 62.79 (14.88) & 64.00 (54.00, 74.00) \\
Height, cm & \texttt{height} & 167.4 (12.85) & 168.0 (160.0, 175.0) \\
\addlinespace
\multicolumn{4}{l}{\textit{Panel B: Patient-month variables ($N=21{,}220$ patient-month observations)}}\\
Month from treatment start & \texttt{month} & 21.14 (18.08) & 14.00 (6.000, 35.00) \\
Treatment count in month & \texttt{event\_txtcount} & 12.06 (1.869) & 13.00 (12.00, 13.00) \\
AVF access proportion & \texttt{ACCESS\_AVF\_avg} & 0.632 (0.477) & 1.000 (0.000, 1.000) \\
AVG access proportion & \texttt{ACCESS\_AVG\_avg} & 0.167 (0.371) & 0.000 (0.000, 0.000) \\
Catheter access proportion & \texttt{ACCESS\_CATH\_avg} & 0.201 (0.393) & 0.000 (0.000, 0.000) \\
Pre-SBP, mmHg & \texttt{*SBP\_0} & 147.8 (20.42) & 147.4 (133.8, 161.7) \\
Post-SBP, mmHg & \texttt{*SBP} & 138.2 (18.64) & 136.6 (124.8, 150.5) \\
Pre-DBP, mmHg & \texttt{*DBP\_0} & 76.67 (12.63) & 75.62 (67.46, 84.85) \\
Post-DBP, mmHg & \texttt{*DBP} & 72.20 (10.81) & 71.22 (64.43, 79.00) \\
Pre-temperature, F & \texttt{*TEMP\_0} & 97.45 (0.509) & 97.46 (97.12, 97.79) \\
Post-temperature, F & \texttt{*TEMP} & 97.58 (0.480) & 97.60 (97.28, 97.91) \\
Interdialytic weight gain, \% & \texttt{*IDWG} & 3.075 (1.202) & 2.992 (2.225, 3.845) \\
Treatment time, min & \texttt{*TXT} & 222.2 (26.21) & 226.2 (207.9, 241.2) \\
EPO dose, fourth-root scale & \texttt{*EPO} & 6.369 (3.388) & 6.950 (5.035, 8.645) \\
Serum sodium, mEq/L & \texttt{*SES} & 138.6 (3.432) & 139.0 (137.0, 141.0) \\
Iron dose & \texttt{*IRON} & 24.05 (25.50) & 15.38 (0.000, 38.46) \\
Ultrafiltration rate & \texttt{*UFR} & 8.405 (3.353) & 8.050 (6.058, 10.36) \\
enPCR & \texttt{*enPCR} & 0.883 (0.262) & 0.850 (0.700, 1.030) \\
eKTV & \texttt{*eKTV} & 1.520 (0.367) & 1.470 (1.310, 1.660) \\
Hemoglobin, g/dL & \texttt{*HGB} & 10.95 (1.116) & 10.96 (10.30, 11.60) \\
Albumin, g/dL & \texttt{*ALB} & 3.799 (0.432) & 3.800 (3.600, 4.100) \\
White blood cells & \texttt{*WBC} & 7.102 (2.833) & 6.710 (5.400, 8.330) \\
Creatinine, mg/dL & \texttt{*CRT} & 7.585 (3.046) & 7.100 (5.300, 9.400) \\
TSAT, \% & \texttt{*TSAT} & 31.44 (13.43) & 29.00 (22.00, 38.00) \\
Neutrophil-to-lymphocyte ratio & \texttt{*NLR} & 1.969 (0.654) & 1.855 (1.539, 2.259) \\
\addlinespace
\multicolumn{4}{l}{\textit{Panel C: Patient-level categorical variables ($m=2{,}159$ patients)}}\\
 &  & Level 1, No. (\%) & Level 2, No. (\%) \\
Race & \texttt{race} & Non-White: 758 (35.1\%) & White: 1,401 (64.9\%) \\
Sex & \texttt{male} & Female: 924 (42.8\%) & Male: 1,235 (57.2\%) \\
Diabetes status & \texttt{diabetic} & Non-diabetic: 838 (38.8\%) & Diabetic: 1,321 (61.2\%) \\
Ethnicity & \texttt{hispanic} & Non-Hispanic: 1,848 (85.6\%) & Hispanic: 311 (14.4\%) \\
\end{longtable}
\endgroup

\section{Supplementary SANAD data analysis}
\label{appendix-SANAD}
We illustrate the proposed method using the \texttt{epileptic.qol} dataset from the \texttt{joineRML} package \citep{hickey2018joinerml}, which is derived from the SANAD (Standard And New Antiepileptic Drugs) trial conducted in the United Kingdom between 1998 and 2006 \citep{marson2007sanad}. SANAD was a large multicenter randomized trial designed to compare established and newer antiepileptic drugs in routine clinical practice. As part of the trial, a quality-of-life substudy collected repeated questionnaire responses at baseline, 3 months, 1 year, and 2 years after treatment initiation. These repeated multivariate outcomes make the dataset well-suited for a multivariate longitudinal mixed-effects analysis.

The dataset contains 544 subjects and 1852 longitudinal observations, with an average of 3.40 repeated measurements per subject. We focus on three continuous quality-of-life outcomes: \texttt{anxiety}, \texttt{depress}, and \texttt{aep}. Higher values of these variables indicate a worse quality of life. The treatment indicator \texttt{trt} distinguishes Carbamazepine (\texttt{CBZ}), the standard comparator, from Lamotrigine (\texttt{LTG}), a newer antiepileptic agent. To capture the trend over time, we also include \texttt{time}, measured in days since baseline, as a continuous covariate because the actual return times of mailed questionnaires vary around the nominal follow-up schedule.

\begin{figure}[H]
    \centering
    \includegraphics[width=\textwidth]{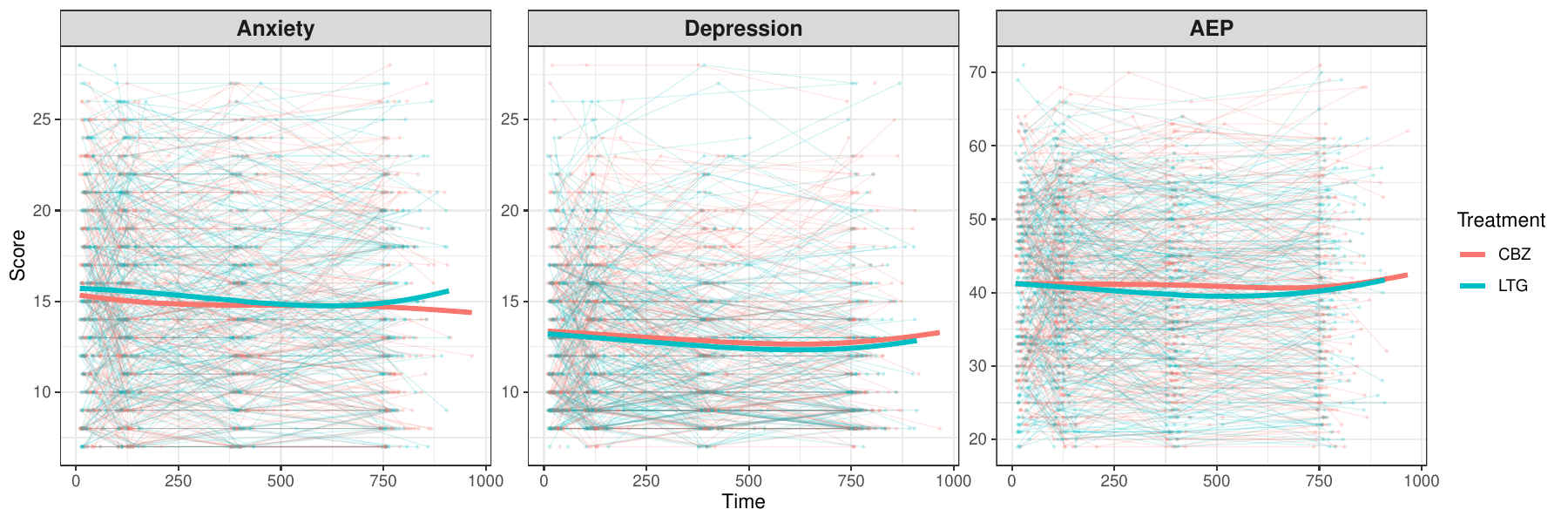}
    \caption{Spaghetti plots of \texttt{anxiety}, \texttt{depress}, and \texttt{aep} score for each subject.}
    \label{fig:spaghetti}
\end{figure}

Figure~\ref{fig:spaghetti} provides an initial exploratory view of the longitudinal quality-of-life trajectories. The treatment groups appear highly overlapping across the three outcomes, and the fitted smooth curves suggest only modest population-level changes over time. At the same time, the individual trajectories exhibit considerable subject-specific variation. This variability is particularly visible for AEP, where patients show more heterogeneous longitudinal patterns than for anxiety and depression. These empirical patterns motivate a multivariate mixed-effects analysis that can separate average treatment and time effects from subject-specific longitudinal heterogeneity.

Our goal is to assess whether treatment and time affect the quality-of-life trajectories of adults with epilepsy. We therefore treat \texttt{anxiety}, \texttt{depress}, and \texttt{aep} as the multivariate response and include \texttt{trt} together with orthogonal polynomial terms of \texttt{time} in the fixed effects. For the random effects, we include only the orthogonal polynomial terms of \texttt{time}. The fitted model is
\begin{equation*}
    \y_{ij} = \bbeta^\top \x_{ij} + \B_i^\top \z_{ij} + \eepsilon_{ij},
\end{equation*}
where
$\y_{ij} = [\texttt{anxiety}_{ij}, \texttt{depress}_{ij}, \texttt{aep}_{ij}]^\top$,
$\x_{ij} = [1,\texttt{trt}_{ij},\texttt{time}_{ij},\texttt{time}_{ij}^2,\texttt{time}_{ij}^3,\texttt{time}_{ij}^4]^\top$,
and
$\z_{ij} = [1,\texttt{time}_{ij},\texttt{time}_{ij}^2,\texttt{time}_{ij}^3,\texttt{time}_{ij}^4]^\top$.
The estimated fixed-effects coefficient matrix is
$$\hat{\bbeta} = \left[ {\begin{array}{ccc}
    15.23 & 13.06 & 40.95 \\
    0 & 0 & 0  \\
    0 & 0 & 0  \\
    0 & 0 & 0  \\
    0 & 0 & 0  \\
    0 & 0 & 0  \\
  \end{array} } \right].$$
Thus, only the intercept terms are retained, while the treatment effect and all time-related fixed effects are shrunk to zero. This suggests that, at the population level, the two treatments do not differ substantially in their effects on these quality-of-life outcomes and that the average trajectories remain essentially flat over the two-year follow-up period. This finding is consistent with the conclusion reported by \citet{jacoby2015quality}.

To examine subject-specific heterogeneity, we also inspect the estimated random-effects covariance matrix $\hat{\Ssigma}_B \in \mathbb{R}^{15 \times 15}$, displayed in Figure~\ref{fig:heatmap}. The heatmap reveals clear differences across the three outcomes. For \texttt{anxiety} and \texttt{depress}, only the random intercept variance is selected, while the random coefficients associated with time are shrunk to zero. This suggests that between-subject variation is concentrated in baseline emotional status, with little evidence of heterogeneous temporal evolution in these two responses. In contrast, the block corresponding to \texttt{aep} exhibits a substantially richer covariance structure, including nonzero higher-order time effects and non-negligible covariances among them. This pattern indicates considerable subject-specific heterogeneity in the longitudinal behavior of adverse-effect burden: some patients improve over time, some worsen, and others follow more complex nonlinear trajectories. Overall, these findings demonstrate that \texttt{MOMENT} can distinguish response-specific heterogeneity in longitudinal dynamics while maintaining an interpretable population-level model.

\begin{figure}[H]
    \centering
    \includegraphics[width=0.4\textwidth]{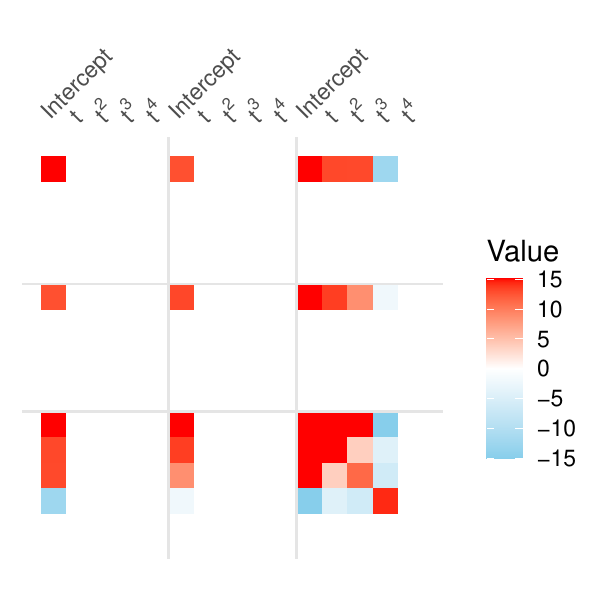}
    \caption{The heatmap of the covariance matrix $\hat{\Ssigma}_B.$ The first block diagonal represents the covariance matrix for response \texttt{anxiety}. The second block diagonal represents the covariance matrix for response \texttt{depress}. The third block diagonal represents the covariance matrix for response \texttt{aep}. The off-diagonal blocks represent the covariance structure among responses.}
    \label{fig:heatmap}
\end{figure}

\section{Proof of Proposition \ref{prop-equivalence}}
Let $\Ssigma \in \R^{qd \times qd}$ be the decision variable and $\Ssigma \succeq 0$ be the constraint. Define the common loss
$$L(\Ssigma)=\frac{1}{2}\sum_{i=1}^m \sum_{\substack{j,k \\ j\neq k}}^{n_i} \left\| (\y_{ij}- \hat{\bbeta}^\top\x_{ij})(\y_{ik}-\hat{\bbeta}^\top\x_{ik})^\top - (\bm{I}_d \otimes \bm{z}_{ij}^\top)\Ssigma (\bm{I}_d \otimes \bm{z}_{ik}) \right\|_F^2.$$
Now, with $\mathcal{P}_B = \lambda \|\W\cdot \mathrm{diag}(\Ssigma_B)\|_1$, define two objective functions
$$F_1(\Ssigma) = L(\Ssigma) + \lambda\  \left\|\W\cdot \diag\left(\Ssigma\right)\right\|_1, \quad F_2(\Ssigma) = L(\Ssigma) + \lambda\ \tr(\W \cdot\Ssigma).$$
For every feasible $\Ssigma \succeq 0,$ we have
$$\lambda \left\|\W\cdot \diag\left(\Ssigma\right)\right\|_1 = \sum_{j=1}^{qd} \W_{jj} \left| \Ssigma_{jj}\right|=\sum_{j=1}^{qd} \W_{jj} \Ssigma_{jj}=\tr(\W \cdot\Ssigma).$$
Hence
$$F_1(\Ssigma) = F_2(\Ssigma) \quad \forall \Ssigma \succeq 0.$$
Define the solution sets to problem \eqref{equ-random-sel} and \eqref{equ-equivalent} as
$$\Set_1 = \argmin_{\Ssigma \succeq 0} F_1(\Ssigma), \quad \Set_2 = \argmin_{\Ssigma \succeq 0} F_2(\Ssigma).$$
Fix any $\Ssigma^\star \in \Set_1,$ then for all feasible $\Ssigma \succeq 0,$ we have
$$F_2(\Ssigma^\star) = F_1(\Ssigma^\star) \leq F_1(\Ssigma) = F_2(\Ssigma).$$
So, $\Ssigma^\star \in \Set_2,$ which means $\Set_1 \subseteq \Set_2.$ The reverse inclusion $\Set_2 \subseteq \Set_1$ also holds by swapping indices 1 and 2. Therefore, we have proved $\Set_1 = \Set_2,$ and thus problem \eqref{equ-random-sel} with $\mathcal{P}_B = \lambda \|\W\cdot \mathrm{diag}(\Ssigma_B)\|_1$ is equivalent to problem \eqref{equ-equivalent}.

\section{Proximal gradient descent algorithm for the selection of $\Sigma_B$}
\label{appendix-algo}
The gradient of the loss function in \eqref{equ-equivalent} without the positive semidefinite constraint is
\begin{equation*}
        \nabla F(\Ssigma_{B})= \sum_{i=1}^m \sum_{\substack{j,k \\ j\neq k}}^{n_i} (\I_d \otimes \z_{ij})\left[ (\I_d \otimes \z_{ij}^\top) \Ssigma_{B} (\I_d \otimes \z_{ik})-(\y_{ij}-\hat{\bbeta}^\top \x_{ij})(\y_{ik}-\hat{\bbeta}^\top \x_{ik})^\top\right] (\I_d \otimes \z_{ik}^\top) + \lambda\ \W.
\end{equation*}
We define the projection operator 
$$\mathrm{proj}_{\SS^{qd}_+}\left(\Ssigma\right):=\sum_{h=1}^{qd} \max(v_h,0)\u_h \u_h^\top,$$
where $\Ssigma=\sum_{h=1}^{qd} v_h \u_h \u_h^\top$ is the eigen-decomposition of a square matrix $\Ssigma.$ Set the step size $\eta = 1/L,$ where L is the largest eigenvalue of the Hessian matrix of the loss function in \eqref{equ-equivalent} without the PSD constraint. We adopt adaptive lasso-type weights constructed from a pilot estimate of the random-effects covariance matrix. Specifically, after obtaining an initial estimate $\tilde{\Ssigma}_B$, we set
$\text{w}_t = (\tilde{\Ssigma}_{B_{tt}} + a)^{-\gamma}$ where $a > 0$ is a small constant to prevent division by zero, and we take $\gamma = 2$ in our implementation.
This choice penalizes small preliminary variance estimates more heavily and reduces shrinkage on larger ones. Then, we give the pseudocode for the proximal gradient descent algorithm for the variable selection of $\Ssigma_B.$ Note that after the convergence of this proximal gradient descent, we can find a reasonable small threshold to hard-threshold the diagonal and select for the active diagonal elements.
\begin{algorithm}
	\renewcommand{\algorithmicrequire}{\textbf{Input:}}
	\renewcommand{\algorithmicensure}{\textbf{Output:}}
	\caption{Proximal Gradient Descent for Selecting $\Ssigma_B$}
	\label{alg-PGD-acc}
	\begin{algorithmic}[1]
		\State Initialization: $\Ssigma_{B}^{[0]} \in \R^{qd \times qd},t=0$
		\Repeat
		\State $t = t + 1$
        \State $\Ssigma_{B}^{[t]}=\mathrm{proj}_{\SS^{qd}_+}\left( \Ssigma_{B}^{[t-1]} - \eta \nabla F (\Ssigma_{B}^{[t-1]}) \right)$
		\Until Convergence
		\Ensure  $\Ssigma_B^{[t]}$
	\end{algorithmic}  
\end{algorithm}

\section{Tuning parameters selection}
We select the tuning parameters using a likelihood-based cross-validation criterion. Specifically, for the purpose of evaluating the validation likelihood, we assume that $B_i$ and $\epsilon_i$ are normally distributed. This assumption is used solely to define the likelihood function and does not affect the validity of the proposed estimation and selection procedure. 

Two tuning parameters are involved in the model: $\lambda$, which controls random effects selection, and $\tau$, which governs the refined estimation of $\bbeta$. We employ a cross-validation procedure to select these parameters sequentially. In particular, we adopt the one-standard-error (1-se) rule when selecting $\lambda$, which favors larger values of $\lambda$ and leads to a more parsimonious random effects structure. The detailed cross-validation procedure is given in Algorithm~\ref{alg-cv}.

\begin{algorithm}[!htbp]
	\renewcommand{\algorithmicrequire}{\textbf{Input:}}
	\renewcommand{\algorithmicensure}{\textbf{Output:}}
	\caption{Procedure of 5-fold cross-validation}
	\label{alg-cv}
    \begin{algorithmic}[1]
    \State Estimate the full $\Ssigma_B$ with $\lambda=0$ to get the weight matrix $\W$.
        \For{$k$ = 1 : 5}
        \State Separate the data into training set $\mathcal{T}_t^{(k)}$ and validation set $\mathcal{T}_v^{(k)}$.
            \ForAll{candidate $\lambda \in \Lambda$}
                \State Compute the $\hat{\Ssigma}_{B}$ for selection and $\hat{\Ssigma}_{\epsilon}$ with the training set $\mathcal{T}_t^{(k)}$.
                \State Compute the log-likelihood $\LL_{\lambda}^{(k)}(\hat{\Ssigma}_{B},\hat{\Ssigma}_{\epsilon})$ on the validation set $\mathcal{T}_v^{(k)}$.
            \EndFor
        \EndFor
        \State Calculate the average of five log-likelihoods: $\LL_{\lambda}=\sum_{k=1}^5 \LL_{\lambda}^{(k)}(\hat{\Ssigma}_{B},\hat{\Ssigma}_{\epsilon}),$ and get the maximal average likelihood $\LL_{\max}.$ For this specific $\lambda_{\max}$ that generates the maximal likelihood, compute the standard error of the likelihood $\text{se}(\LL_{\max})$ over the 5 folds. \newline
         $\lambda_\text{optimal} \leftarrow \max\limits_{\LL_{\tilde \lambda} > \LL_{\max} - \text{se}(\LL_{\max})} \tilde \lambda$
        \State Update $\Ssigma_B$ with $\lambda_\text{optimal}$, then select the relevant random effects.
    \State Refit the model to get $\hat{\Ssigma}_{B}$ and $\hat{\Ssigma}_{\epsilon}.$
    \State Do regular cross-validation to find $\tau_{optimal}$ and $\hat{\bbeta}.$
	\end{algorithmic}
\end{algorithm}

\section{Derivation of problem reformulation as lasso problem}
\label{appendix-reform}
We first vectorize the population's second moment
$$\begin{aligned}
    \mathrm{vec}[(\I_d \otimes \z_{ij}^\top) \Ssigma_B (\I_d \otimes \z_{ik})] &= [(\I_d \otimes \z_{ik}) \otimes (\I_d \otimes \z_{ij})]^\top \vvec(\Ssigma_B)\\
    & = [(\I_d \otimes \z_{ik}) \otimes (\I_d \otimes \z_{ij})]_\mathcal{I}^\top\cdot \ttheta + [(\I_d \otimes \z_{ik}) \otimes (\I_d \otimes \z_{ij})]_{\mathcal{I}^c}^\top\cdot \rrho.
\end{aligned}
$$
So, the original problem \eqref{equ-theory-original} becomes:
\begin{equation*}
\begin{split}
        \min_{\ttheta, \rrho} &\frac{1}{2}\sum_{i=1}^m \sum_{\substack{j,k \\ j\neq k}}^{n_i} \left\| [(\I_d \otimes \z_{ik}) \otimes (\I_d \otimes \z_{ij})]_\mathcal{I}^\top\cdot \ttheta + [(\I_d \otimes \z_{ik}) \otimes (\I_d \otimes \z_{ij})]_{\mathcal{I}^c}^\top\cdot \rrho \right. \\ 
        &\left. \quad \quad \quad \quad \quad -\mathrm{vec}[(\y_{ij}-\bbeta^\top \x_{ij})(\y_{ik}-\bbeta^\top \x_{ik})^\top]\right\|_2^2 + \lambda \|\W \ttheta\|_1,
\end{split}
\end{equation*}
which is equivalent to 
\begin{equation*}
\begin{split}
        \min_{\ttheta, \rrho} &\frac{1}{2}\sum_{i=1}^m \sum_{\substack{j,k \\ j\neq k}}^{n_i} \sum_{\substack{\ell=1 \\ \ell'=1}}^d \left( [(\e_{\ell} \otimes \z_{ik}) \otimes (\e_{\ell'} \otimes \z_{ij})]_\mathcal{I}^\top\cdot \ttheta + [(\e_{\ell} \otimes \z_{ik}) \otimes (\e_{\ell'} \otimes \z_{ij})]_{\mathcal{I}^c}^\top\cdot \rrho \right. \\ 
        &\left. \quad \quad \quad \quad \quad - (y_{ij\ell}-\bbeta_{\ell}^\top \x_{ij})(y_{ik\ell'}-\bbeta_{\ell'}^\top \x_{ik})\right)^2 + \lambda \|\W \ttheta\|_1.
\end{split}
\end{equation*}
This can be further simplified to
$$\min_{\ttheta, \rrho} \frac{1}{2}\|\A_{\mathcal{I}}\ttheta + \A_{\mathcal{I}^c}\rrho - \r\|_2^2 + \lambda \|\W\ttheta\|_1.$$
Recall that $\P_{\A_{\mathcal{I}^c}^\perp} = \I - \A_{\mathcal{I}^c}(\A_{\mathcal{I}^c}^\top \A_{\mathcal{I}^c})^{-1}\A_{\mathcal{I}^c}^\top \in \R^{Md^2\times Md^2}$ is the projection matrix onto the orthogonal complement of the column space of $\A_{\mathcal{I}^c}.$ Hence, the problem is equivalent to an adaptive lasso form of the problem
$$\min_{\ttheta} \frac{1}{2} \|\P_{\A_{\mathcal{I}^c}^\perp}(\A_{\mathcal{I}}\ttheta - \r)\|_2^2 + \lambda\|\W\ttheta\|_1,$$
by appropriate rescaling of the covariates \citep{buhlmann2011statistics}, this adaptive lasso reduces to the standard lasso problem as in 
\eqref{equ-theory-lasso}.

\section{Primal-Dual Witness construction}
With Assumption~\ref{a2} holds, we use PDW to construct a pair $(\check{\ttheta},\check{\v}) \in \R^{qd} \times \R^{qd} $ according to the following steps.
\begin{enumerate}[label=(\arabic*)]
    \item We first solve the restricted problem
    \begin{equation*}
        \check{\ttheta}_S = \argmin_{\ttheta_S \in \R^s} \frac{1}{2} \|\mathcal{Y} - \mathcal{X}_S\ttheta_S\|_2^2 + \lambda\|\ttheta_S\|_1
    \end{equation*}
    to obtain $\check{\ttheta}_S \in \R^s.$ Note that this solution is always unique under the invertibility assumption on $\mathcal{X}_S^\top \mathcal{X}_S$ implied by Assumption~\ref{a2}. We set $\check{\ttheta}_{S^c} = 0.$
    
    \item Then, we take $\check{\v}_S \in \R^s$ to be an arbitrary element of the subdifferential of the $\ell_1$ norm evaluated at $\check{\ttheta}_S.$

    \item We solve for a vector $\check{\v}_{S^c} \in \R^{qd-s}$ satisfying the zero subgradient condition
    \begin{equation*}
        \mathcal{X}^\top \mathcal{X} (\hat{\ttheta} - \ttheta^\star) - \mathcal{X}^\top \oomega + \lambda \hat{\v} = 0,
    \end{equation*}
    where $\hat{\v} \in \partial \|\hat{\ttheta}\|_1$ is a subgradient vector. And we need to check whether the strict dual feasibility $|\check{\v}_j| < 1$ for all $j\in S^c$ is satisfied.

    \item Finally, we verify whether the sign consistency condition $\check{\v}_S = \text{sign}(\ttheta_S^\star)$ is satisfied.
\end{enumerate}

Later in the proof, we will verify the strict dual feasibility in Step 3 and the sign consistency condition in Step 4. Each element $j \in S^c$ in the vector $\check{\v}_{S^c}$ can be written as a scalar random variable
\begin{equation}
\label{rv-vj}
    V_j = \mathcal{X}_j^\top \left[\mathcal{X}_S (\mathcal{X}_S^\top \mathcal{X}_S)^{-1} \check{\v}_S + \P_{\mathcal{X}^\perp_S} \cdot \frac{\oomega}{\lambda}\right],
\end{equation}
where $\P_{\mathcal{X}^\perp_S} = \I - \mathcal{X}_S (\mathcal{X}_S^\top \mathcal{X}_S)^{-1} \mathcal{X}_S^\top \in \R^{Md^2 \times Md^2}$ is an orthogonal projection matrix, and $\check{\v}_S \in \R^s$ is the subgradient vector chosen in Step 2 of PDW construction. On the other hand, for each $k\in S,$ we can define the scalar random variable
\begin{equation}
\label{rv-deltak}
    \Delta_k = \e_k^\top \left(\mathcal{X}_S^\top \mathcal{X}_S \right)^{-1} \left[ \mathcal{X}_S^\top \oomega - \lambda \cdot \text{sign}(\ttheta_S^\star)\right].
\end{equation}
Note that if the problem is sign consistent, then $\Delta_k$ is equal to $\hat{\ttheta}_k - \ttheta^\star_k,$ which is the difference between the lasso solution $\hat{\ttheta}$ and the truth $\ttheta^\star$ at position $k.$ The following proposition reviews the lemmas in \cite{wainwright2009sharp} that are essential for proving the sign-selection consistency.

\begin{proposition}[Lemma 2 and 3 in \cite{wainwright2009sharp}]
\label{prop-wainwright}
    Assume Assumption~\ref{a2} holds, then
    \begin{enumerate}
        \item [(a)] The strict dual feasibility check in Step 3 of the PDW succeeds if and only if 
        $$|V_j| < 1, \quad \forall j \in S^c,$$
        which is a sufficient condition for the problem \eqref{equ-theory-lasso} to have a unique solution $\hat \ttheta$ with $\mathrm{supp}(\hat \ttheta) \subseteq \mathrm{supp}(\ttheta^\star).$

        \item [(b)] The sign-selection consistency in Step 4 of the PDW can be satisfied if and only if
        $$\mathrm{sign}(\ttheta^\star_k + \Delta_k) = \mathrm{sign}(\ttheta^\star_k), \quad \forall k \in S,$$
        which, combined with strict dual feasibility, are sufficient conditions for the problem \eqref{equ-theory-lasso} to have a unique solution $\hat \ttheta$ with correct signed support recovery $\mathrm{supp}_{\pm}(\hat{\ttheta}) = \mathrm{supp}_{\pm}(\ttheta^\star).$
    \end{enumerate}
\end{proposition}
Proposition~\ref{prop-wainwright} reduces signed support recovery of \eqref{equ-theory-lasso} to verification of two conditions arising from the PDW construction: the strict dual feasibility condition on the inactive set and the sign consistency condition on the active set. Therefore, to establish selection consistency, it suffices to control the random quantities $V_j$ for $j \in S^c$ and $\Delta_k$ for $k \in S$. The main theorem below shows that, under the stated assumptions and an appropriate choice of $\lambda$, these two events hold simultaneously with high probability, which in turn yields unique recovery of the signed support $\mathrm{supp}_{\pm}(\hat{\ttheta}) = \mathrm{supp}_{\pm}(\ttheta^\star)$.

\section{Proof of Theorem \ref{thm-sel-consist}}
By the Primal-Dual Witness construction and the lemmas from \cite{wainwright2009sharp}, we need to first establish the strict dual feasibility for random variable $\{V_j, j\in S^c\}$ defined in (\ref{rv-vj}) with high probability to ensure that Step 3 of the PDW construction succeeds. Then, establish the $\ell_{\infty}$ bound for the random variable $\{\Delta_k, k \in S\}$ defined in (\ref{rv-deltak}) so that the sign consistency condition in Step 4 of PDW construction holds.

\paragraph*{Part 1: Establish strict dual feasibility with high probability.}
In this part, we will prove that Step 3 of the PDW construction succeeds with high probability. We can decompose $V_j$ as 
$$V_j = \mu_j + \tilde{V}_j,$$
where $\mu_j = \mathcal{X}_j^\top \mathcal{X}_S (\mathcal{X}_S^\top \mathcal{X}_S)^{-1} \check{\v}_S,$ and $\tilde{V}_j = \mathcal{X}_j^\top \P_{\mathcal{X}^\perp_S} \cdot \frac{\oomega}{\lambda} = \mathcal{X}_j^\top \P_{\mathcal{X}^\perp_S} \P_{\A_{\mathcal{I}^c}^\perp} \cdot \frac{\xxi}{\lambda}.$ Since $\check{\v}_S$ is a subgradient vector for the $\ell_1$ norm, we have $\|\check{\v}_S\|_{\infty} \leq 1.$ Furthermore, by the irrepresentability assumption \ref{a1}, we have $|\mu_j| \leq 1-\nu$ for all $j \in S^c$. Therefore, we have
$$\max_{j \in S^c} |V_j| \leq 1 - \nu + \max_{j \in S^c} |\tilde{V}_j|.$$

Let $\g^{(j)} = \P_{\A_{\mathcal{I}^c}^\perp}\P_{\mathcal{X}^\perp_S} \mathcal{X}_j /\lambda \in \R^{Md^2}.$ We can separate $\g^{(j)}$ in the same way as we separate $\xxi$ such that $\g^{(j)} = [\g_1^{(j)^ \top}, \g_2^{(j)^ \top}, ..., \g_m^{(j)^ \top}]^\top,$ where $\g^{(j)}_i \in \R^{N_i d^2}.$ So, $\tilde{V}_j = \sum_{i=1}^m \g_i^{(j)^ \top} \xxi_i$ is the sum of independent zero-mean sub-Weibull($\frac{\alpha}{2}$) scalar random variables. Let $\g^{(j)^\top}_i = \mathcal{X}_j^\top \P_{\mathcal{X}^\perp_S} \P_{\A_{\mathcal{I}^c}^\perp} \Ppi_i/\lambda,$ where 
\begin{equation}
\label{equ-theory-Pi}
    \Ppi_i = \left[
\begin{array}{@{}c@{}}
    \mathbf{0}\\
    \I\\
    \tilde{\mathbf{0}}
\end{array}
\right]
\in\R^{Md^2\times N_i d^2}
\end{equation}
is a column selector matrix, $\mathbf{0}\in \R^{(\sum_{\ell=0}^{i-1} N_{\ell})d^2 \times N_i d^2}$ and $\tilde{\mathbf{0}}\in \R^{(M - \sum_{\ell=1}^i N_{\ell})d^2 \times N_i d^2}$ are matrices full of zeros, $\I \in \R^{N_i d^2 \times N_i d^2}$ is an identity matrix. Note that $\I$ will be on the top of $\Ppi_i$ if $i=1,$ and will be on the bottom of $\Ppi_i$ if $i=m$. The location of $\I$ in $\Ppi_i$ will change corresponding to $i.$ Then, 
\begin{equation*}
    \begin{aligned}
        \|\g_i^{(j)}\|_2 &= \frac{1}{\lambda} \|\mathcal{X}_j^\top \P_{\mathcal{X}^\perp_S} \P_{\A_{\mathcal{I}^c}^\perp} \Ppi_i\|_2 \\
        &\leq \frac{1}{\lambda} \|\mathcal{X}_j \|_2 \cdot \vertiii{\P_{\mathcal{X}^\perp_S} \P_{\A_{\mathcal{I}^c}^\perp} \Ppi_i}_2 \\
        &\leq \frac{1}{\lambda} \vertiii{\P_{\A_{\mathcal{I}^c}^\perp}}_2 \|\A_{\mathcal{I}_j}\|_2 \vertiii{\P_{\mathcal{X}^\perp_S}}_2 \cdot \vertiii{\P_{\A_{\mathcal{I}^c}^\perp}}_2 \cdot \vertiii{\Ppi_i}_2 \\
        & \leq \frac{\sqrt{M}d}{\lambda }.
    \end{aligned}
\end{equation*}
Here we use the condition $(Md^2)^{-\frac{1}{2}} \max_{j \in S^c} \|\A_{\mathcal{I}_j}\| \leq 1$, and the fact that the projection matrices $\P_{\mathcal{X}^\perp_S}$ and $\P_{\A_{\mathcal{I}^c}^\perp}$ have spectral norm one, and the largest singular value of $\Ppi$ is one. We also repeatedly use the sub-multiplicative property of operator norms. Using the same argument, we can prove that $\|\g^{(j)}\| \leq \frac{\sqrt{M}d}{\lambda}$ as well. Recall that $K_i = \|\xxi_i\|_{J,\psi_{\alpha/2}},$ then $\sum_{i=1}^m \|\g_i^{(j)^\top} \xxi_i\|_{\psi_{\alpha/2}}^2 \leq \sum_{i=1}^m \|\g_i^{(j)}\|_2^2 K_i^2 \leq \|\g^{(j)}\|_2^2 (\max_i K_i)^2 \leq \frac{Md^2}{\lambda^2} (\max_i K_i)^2.$ So, by Theorem 3.1 from \cite{kuchibhotla2022moving}, we have
\begin{equation*}
    \begin{aligned}
        \PP\left(|\tilde{V}_j| \geq t\right) &\leq 2 \exp\left\{-c\cdot \min \left[\frac{t^2}{\sum_{i=1}^m \|\g_i^{(j)^\top} \xxi_i\|_{\psi_{\alpha/2}}^2}, \left(\frac{t}{\max_i \|\g_i^{(j)^\top} \xxi_i\|_{\psi_{\alpha/2}}}\right)^{\alpha/2}\right]\right\} \\
        &\leq 2 \exp\left\{-c\cdot \min \left[\frac{t^2}{\sum_{i=1}^m \left\|\g_i^{(j)}\right\|_2^2 K_i^2}, \left(\frac{t}{\max_i \left\|\g_i^{(j)}\right\|_2 K_i}\right)^{\alpha/2}\right]\right\} \\
        &\leq 2 \exp\left\{-c\cdot \min \left[\frac{t^2 \lambda^2}{Md^2 \cdot (\max_i K_i)^2}, \left(\frac{t \lambda }{\sqrt{M}d \cdot \max_i K_i}\right)^{\alpha/2}\right]\right\}.
    \end{aligned}
\end{equation*}
By union bound, we have
\begin{equation*}
    \begin{aligned}
        \PP\left(\max_{j\in S^c}|\tilde{V}_j| \geq t\right) &\leq \sum_{j \in S^c}\PP\left(|\tilde{V}_j| \geq t\right) \\
        &\leq 2(qd-s) \exp\left\{-c\cdot \min \left[\frac{t^2 \lambda^2}{Md^2 \cdot (\max_i K_i)^2}, \left(\frac{t \lambda }{\sqrt{M}d \cdot \max_i K_i}\right)^{\alpha/2}\right]\right\}.
    \end{aligned}
\end{equation*}
By setting $t = \frac{\nu}{2}$ we will have
$$\PP\left(\max_{j\in S^c}|\tilde{V}_j| \geq \frac{\nu}{2}\right) \leq 2(qd-s) \exp\left\{-c\cdot \min \left[\frac{\nu^2 \lambda^2}{4Md^2(\max_i K_i)^2}, \left(\frac{\nu \lambda }{2\sqrt{M}d\max_i K_i}\right)^{\alpha/2}\right]\right\}.$$
With a slightly different constant $\tilde{c}$, let 
\begin{align*}
    \lambda & \geq \frac{2d\sqrt{M} \max_i K_i}{\tilde{c} \nu}\max\left\{\sqrt{\log\frac{2(qd-s)}{\delta}}, \left(\log\frac{2(qd-s)}{\delta}\right)^{2/\alpha}\right\} \\
    & \geq \max\left\{\sqrt{\frac{4Md^2 (\max_i K_i)^2}{c \nu^2}\log\frac{2(qd-s)}{\delta}}, \frac{2\sqrt{M}d \max_i K_i}{c^{2/\alpha} \nu }\left(\log\frac{2(qd-s)}{\delta}\right)^{2/\alpha}\right\} 
\end{align*}
we will have 
$$\PP\left(\max_{j\in S^c}|\tilde{V}_j| \geq \frac{\nu}{2}\right) \leq \delta.$$
Denote event $\mathscr{B} = \left\{\max_{j\in S^c}|\tilde{V}_j| \geq \frac{\nu}{2}\right\}$ and event $\mathscr{C} = \left\{\max_{j\in S^c}|V_j| \geq 1 - \frac{\nu}{2}\right\}.$ It's trivial to show that $\mathscr{B} ^c \subseteq \mathscr{C}^c,$ so $\mathscr{C} \subseteq \mathscr{B}.$ Therefore,
$$\PP\left(\max_{j\in S^c}|V_j| \geq 1 - \frac{\nu}{2}\right) \leq \PP\left(\max_{j\in S^c}|\tilde{V}_j| \geq \frac{\nu}{2}\right) \leq \delta.$$
By choosing $\nu$ to be small enough, we have proved that the strict dual feasibility in Step 3 of PDW construction succeeds with probability greater than $1-\delta.$

\paragraph*{Part 2: Establish $\ell_\infty$ bound.}
In this part, we will establish an $\ell_{\infty}$ bound on the variables $\{\Delta_k, \ k \in S\}$ to ensure that the sign selection consistency in Step 4 of PDW construction holds with high probability.

By the triangular inequality and the sub-multiplicative property for the matrix infinity norm, we have
\begin{equation*}
    \begin{aligned}
        \max_{k\in S} |\Delta_k| &\leq \left\|\left(\mathcal{X}_S^\top \mathcal{X}_S \right)^{-1}  \mathcal{X}_S^\top \oomega\right\|_{\infty} + \left\|\left(\mathcal{X}_S^\top \mathcal{X}_S \right)^{-1}\text{sign}(\ttheta^\star_S)\right\|_{\infty}\cdot \lambda\\
        &\leq \left\|\left(\mathcal{X}_S^\top \mathcal{X}_S \right)^{-1}  \mathcal{X}_S^\top \oomega\right\|_{\infty} + \bigvertiii{\left(\mathcal{X}_S^\top \mathcal{X}_S \right)^{-1}}_{\infty}\cdot \lambda\\
        &=\left\|\left(\mathcal{X}_S^\top \mathcal{X}_S \right)^{-1}  \mathcal{X}_S^\top \oomega\right\|_{\infty} + \bigvertiii{\left(\frac{1}{Md^2}\mathcal{X}_S^\top \mathcal{X}_S \right)^{-1}}_{\infty}\cdot \frac{\lambda}{Md^2}
    \end{aligned}
\end{equation*}
Since there is no randomness in the second term, we only need to bound the first term. For each $k = 1, 2, ..., s,$ consider the random variable
$$u_k = \e_k^\top \left(\mathcal{X}_S^\top \mathcal{X}_S \right)^{-1} \mathcal{X}_S^\top \oomega = \e_k^\top \left(\mathcal{X}_S^\top \mathcal{X}_S \right)^{-1} \mathcal{X}_S^\top \P_{\A_{\mathcal{I}^c}^\perp} \xxi.$$
Let $\h^{(k)^\top} = \e_k^\top \left(\mathcal{X}_S^\top \mathcal{X}_S \right)^{-1} \mathcal{X}_S^\top \P_{\A_{\mathcal{I}^c}^\perp} \in \R^{Md^2}.$ Again, we can separate it as $\h^{(k)} = [\h_1^{(k)^\top}, \h_2^{(k)^\top}, ..., \h_m^{(k)^\top}]^\top$ where $\h_i^{(k)} \in \R^{N_i d^2}.$ So, $u_k = \sum_{i=1}^m \h_i^{(k)^\top} \xxi_i$ is the sum of independent zero-mean sub-Weibull($\frac{\alpha}{2}$) random variables. Let $\h_i^{(k)^\top} = \e_k^\top \left(\mathcal{X}_S^\top \mathcal{X}_S\right)^{-1}  \mathcal{X}_S^\top \P_{\A_{\mathcal{I}^c}^\perp} \Ppi_i,$ where $\Ppi_i$ is the column selector matrix \eqref{equ-theory-Pi} defined in the previous part of the proof. Then, 
\begin{equation*}
    \begin{aligned}
        \|\h_i^{(k)}\|_2^2 &= \bigvertiii{\Ppi_i^\top \P_{\A_{\mathcal{I}^c}^\perp} \mathcal{X}_S \left(\mathcal{X}_S^\top \mathcal{X}_S\right)^{-1} \e_k \e_k^\top \left(\mathcal{X}_S^\top \mathcal{X}_S\right)^{-1} \mathcal{X}_S^\top \P_{\A_{\mathcal{I}^c}^\perp} \Ppi_i}_2 \\
        &\leq \bigvertiii{\left(\mathcal{X}_S^\top \mathcal{X}_S \right)^{-1} \mathcal{X}_S^\top}_2^2 \bigvertiii{\Ppi_i \Ppi_i^\top}_2 \bigvertiii{\P_{\A_{\mathcal{I}^c}^\perp}}_2^2 \bigvertiii{\e_k \e_k^\top}_2 \\
        &\leq \bigvertiii{\left(\mathcal{X}_S^\top \mathcal{X}_S \right)^{-1}}_2 \\
        &= \frac{1}{Md^2} \bigvertiii{\left(\frac{1}{Md^2}\mathcal{X}_S^\top \mathcal{X}_S \right)^{-1}}_2 \\
        &\leq \frac{1}{Md^2 \cdot C_{\min}}.
    \end{aligned}
\end{equation*}
Here, we repeatedly use the sub-multiplicative property of the operator norm to obtain the first inequality, and we use the assumption \ref{a2} to obtain the last inequality. Again, $\|\h^{(k)} \|_2^2 \leq \frac{1}{Md^2 \cdot C_{\min}}$ as well. Thus, by Theorem 3.1 from \cite{kuchibhotla2022moving}, we have 
\begin{equation*}
    \begin{aligned}
        \PP\left(|u_k| \geq t\right) &\leq 2 \exp\left\{-c\cdot \min \left[\frac{t^2}{\sum_{i=1}^m \|\h_i^{(j)^\top} \xxi_i\|_{\psi_{\alpha/2}}^2}, \left(\frac{t}{\max_i \|\h_i^{(j)^\top} \xxi_i\|_{\psi_{\alpha/2}}}\right)^{\alpha/2}\right]\right\} \\
        &\leq 2 \exp\left\{-c\cdot \min \left[\frac{t^2}{\sum_{i=1}^m \left\|\h_i^{(k)}\right\|_2^2 K_i^2}, \left(\frac{t}{\max_i \left\|\h_i^{(k)}\right\|_2 K_i}\right)^{\alpha/2}\right]\right\} \\
        &\leq 2 \exp\left\{-c\cdot \min \left[\frac{t^2 C_{\min} Md^2}{(\max_i K_i)^2}, \left(\frac{t \sqrt{C_{\min}M}d}{\max_i K_i}\right)^{\alpha/2}\right]\right\}.
    \end{aligned}
\end{equation*}
By union bound, we have
$$\PP\left(\max_{k\in S}|u_k| \geq t\right) \leq 2\cdot s \cdot \exp\left\{-c\cdot \min \left[\frac{t^2 C_{\min} Md^2}{(\max_i K_i)^2}, \left(\frac{t \sqrt{C_{\min}M}d}{\max_i K_i}\right)^{\alpha/2}\right]\right\}.$$
By setting
$$t = \max\left\{\sqrt{\frac{(\max_i K_i)^2}{c C_{\min} Md^2}\log\frac{2s}{\delta}}, \frac{\max_i K_i}{c^{2/\alpha} \sqrt{C_{\min}M}d}\left(\log\frac{2s}{\delta}\right)^{2/\alpha}\right\},$$
with a slightly different constant $\tilde{c}$, we have
\begin{align*}
    \|\hat{\ttheta}_S - \ttheta^\star_S\|_{\infty} &\leq \frac{\lambda}{Md^2} \bigvertiii{\left(\frac{1}{Md^2}\mathcal{X}_S^\top \mathcal{X}_S \right)^{-1}}_{\infty} + \max \left\{\sqrt{\frac{(\max_i K_i)^2}{c C_{\min} Md^2}\log\frac{2s}{\delta}}, \frac{\max_i K_i}{c^{2/\alpha} \sqrt{C_{\min}M}d}\left(\log\frac{2s}{\delta}\right)^{2/\alpha}\right\} \\
    & \leq \underbrace{\frac{\lambda}{Md^2} \bigvertiii{\left(\frac{1}{Md^2}\mathcal{X}_S^\top \mathcal{X}_S \right)^{-1}}_{\infty} + \frac{\max_i K_i}{\tilde{c} \sqrt{C_{\min} M} d}\max \left\{\sqrt{\log\frac{2s}{\delta}}, \left(\log\frac{2s}{\delta}\right)^{2/\alpha}\right\}}_{f(\lambda)}
\end{align*}
holds with probability greater than $1-\delta.$

Combining strict dual feasibility and the $\ell_\infty$ bound, if the minimal signal satisfies $\min_{j \in S} |\theta_j^\star| > f(\lambda)$, support and signed-support recovery follow.

\section{Proof of Corollary \ref{corollary-rate}}
\label{appendix-corollary-rate}
We start with quantifying $K_i= \|\xxi_i\|_{J,\psi_{\alpha/2}}.$ First, we focus on each element in the random vector $\xxi_i.$ By the triangular inequality, we have
\begin{equation*}
    \left\|\xi_{\pi(i,j,k,\ell,\ell')}\right\|_{\psi_{\alpha/2}} \leq \left\|\xi^{(BB)}_{\pi(i,j,k,\ell,\ell')}\right\|_{\psi_{\alpha/2}} + \left\|\xi^{(B\epsilon)}_{\pi(i,j,k,\ell,\ell')}\right\|_{\psi_{\alpha/2}} + \left\|\xi^{(\epsilon B)}_{\pi(i,j,k,\ell,\ell')}\right\|_{\psi_{\alpha/2}} + \left\|\xi^{(\epsilon \epsilon)}_{\pi(i,j,k,\ell,\ell')}\right\|_{\psi_{\alpha/2}},
\end{equation*}
where 
$$
\begin{aligned}
    \xi^{(BB)}_{\pi(i,j,k,\ell,\ell')} &= \B_{i,\ell}^\top \z_{ij} \z_{ik}^\top \B_{i,\ell'} - \A_{\pi(i,j,k,\ell,\ell'),\mathcal{I}} \ttheta^\star - \A_{\pi(i,j,k,\ell,\ell'),\mathcal{I}^C} \rrho^\star = \B_{i,\ell}^\top \z_{ij} \z_{ik}^\top \B_{i,\ell'} - \EE[\B_{i,\ell}^\top \z_{ij} \z_{ik}^\top \B_{i,\ell'}]\\
    \xi^{(B\epsilon)}_{\pi(i,j,k,\ell,\ell')} &= \B_{i,\ell}^\top \z_{ij} \epsilon_{ik\ell'}\\
    \xi^{(\epsilon B)}_{\pi(i,j,k,\ell,\ell')} &= \epsilon_{ij\ell} \z_{ik}^\top \B_{i,\ell'}\\
    \xi^{(\epsilon \epsilon)}_{\pi(i,j,k,\ell,\ell')} &= \epsilon_{ij\ell} \epsilon_{ik\ell'}.
\end{aligned}$$
We bound the Orlicz norm of these four random variables.

\paragraph*{Bound for $\|\xi^{(BB)}_{\pi(i,j,k,\ell,\ell')}\|_{\psi_{\alpha/2}}.$}
    First, we obtain a bound on the Orlicz norm of $\B_{i,\ell}$ by
$$\|\B_{i,\ell}\|_{J,\psi_{\alpha}} = \sup_{\|\eeta\|_2 = 1} \|\eeta^\top \Ssigma_B^{\star(\ell,\ell)^{1/2}} \pphi_{i,\ell}\|_{\psi_{\alpha}} \leq \sup_{\|\eeta\|_2 = 1} \|\Ssigma_B^{\star(\ell,\ell)^{1/2}} \eeta\|_2 \|\pphi_{i,\ell}\|_{J,\psi_{\alpha}} \leq \kappa \sqrt{\vertiii{\Ssigma_B^{\star(\ell,\ell)}}}_2.$$ 
Then, we bound the Orlicz norm of $\B_{i,\ell}^\top \z_{ij} \z_{ik}^\top \B_{i,\ell'}$ by
$$\begin{aligned}
    \left\|\B_{i,\ell}^\top \z_{ij} \z_{ik}^\top \B_{i,\ell'}\right\|_{\psi_{\alpha/2}} &\leq \left\|\z_{ij}^\top \B_{i,\ell}\right\|_{\psi_{\alpha}} \left\|\z_{ik}^\top \B_{i,\ell'}\right\|_{\psi_{\alpha}} \\
    &\leq \|\z_{ij}\|_2  \|\B_{i,\ell}\|_{J,\psi_{\alpha}}  \|\z_{ik}\|_2  \|\B_{i,\ell'}\|_{J,\psi_{\alpha}} \\
    &\leq \|\z_{ij}\|_2 \cdot \kappa \sqrt{\vertiii{\Ssigma_B^{\star(\ell,\ell)}}}_2 \cdot \|\z_{ik}\|_2 \cdot \kappa \sqrt{\vertiii{\Ssigma_B^{\star(\ell',\ell')}}}_2 \\
    &\leq C_z^2 \cdot \kappa^2 \cdot \max_{\ell}\vertiii{\Ssigma_B^{\star(\ell,\ell)}}_2.
\end{aligned}$$
By centering property (Lemma 2.4 \citep{sheu2023matrix}), 
$$\begin{aligned}
    \left\|\xi^{(BB)}_{\pi(i,j,k,\ell,\ell')}\right\|_{\psi_{\alpha/2}} & = \left\|\B_{i,\ell}^\top \z_{ij} \z_{ik}^\top \B_{i,\ell'} - \EE[\B_{i,\ell}^\top \z_{ij} \z_{ik}^\top \B_{i,\ell'}]\right\|_{\psi_{\alpha/2}} \\
    & \leq C \left\|\B_{i,\ell}^\top \z_{ij} \z_{ik}^\top \B_{i,\ell'}\right\|_{\psi_{\alpha/2}} \\
    & \leq C \cdot C_z^2 \cdot \kappa^2 \cdot \max_{\ell}\vertiii{\Ssigma_B^{\star(\ell,\ell)}}_2,
\end{aligned}$$
where C is a constant.

\paragraph*{Bound for $\|\xi^{(B\epsilon)}_{\pi(i,j,k,\ell,\ell')}\|_{\psi_{\alpha/2}}$ and $\|\xi^{(\epsilon B)}_{\pi(i,j,k,\ell,\ell')}\|_{\psi_{\alpha/2}}.$}
    Similarly, we have
\begin{equation*}
    \begin{aligned}
    \left\|\xi^{(B\epsilon)}_{\pi(i,j,k,\ell,\ell')}\right\|_{\psi_{\alpha/2}} & = \left\|\z_{ij}^\top \B_{i,\ell} \epsilon_{ik\ell'}\right\|_{\psi_{\alpha/2}} \\
    & \leq \left\| \z_{ij}^\top \B_{i,\ell} \right\|_{\psi_{\alpha}} \left\|\epsilon_{ik\ell'}\right\|_{\psi_{\alpha}} \\
    & \leq \|\z_{ij}\|_2 \|\B_{i,\ell}\|_{J,\psi_{\alpha}} \|\epsilon_{ik\ell'}\|_{\psi_{\alpha}}  \\
    & \leq C_z \cdot \kappa \sqrt{\vertiii{\Ssigma_B^{\star(\ell,\ell)}}}_2 \cdot \sigma_{\epsilon}.
\end{aligned}
\end{equation*}
And we have the same bound for $\left\|\xi^{(\epsilon B)}_{\pi(i,j,k,\ell,\ell')}\right\|_{\psi_{\alpha/2}}.$

\paragraph*{Bound for $\|\xi^{(\epsilon \epsilon)}_{\pi(i,j,k,\ell,\ell')}\|_{\psi_{\alpha/2}}.$}
    By a similar argument to the previous, we have
$$\begin{aligned}
    \left\|\xi^{(\epsilon \epsilon)}_{\pi(i,j,k,\ell,\ell')}\right\|_{\psi_{\alpha/2}} & = \left\|\epsilon_{ij\ell} \epsilon_{ik\ell'}\right\|_{\psi_{\alpha/2}} \\
    & \leq \left\|\epsilon_{ij\ell}\right\|_{\psi_{\alpha}} \left\|\epsilon_{ik\ell'}\right\|_{\psi_{\alpha}} \\
    & \leq \sigma_{\epsilon}^2.
\end{aligned}$$
Therefore, 
\begin{equation*}
\begin{aligned}
        \left\|\xi_{\pi(i,j,k,\ell,\ell')}\right\|_{\psi_{\alpha/2}} &\leq \left\|\xi^{(BB)}_{\pi(i,j,k,\ell,\ell')}\right\|_{\psi_{\alpha/2}} + \left\|\xi^{(B\epsilon)}_{\pi(i,j,k,\ell,\ell')}\right\|_{\psi_{\alpha/2}} + \left\|\xi^{(\epsilon B)}_{\pi(i,j,k,\ell,\ell')}\right\|_{\psi_{\alpha/2}} + \left\|\xi^{(\epsilon \epsilon)}_{\pi(i,j,k,\ell,\ell')}\right\|_{\psi_{\alpha/2}} \\
        &\leq C\left(C_z \kappa \sqrt{\max_{\ell}\vertiii{\Ssigma_B^{\star(\ell,\ell)}}_2} + \sigma_{\epsilon}\right)^2.
\end{aligned}
\end{equation*}
If $\alpha \geq 2,$ then $\| \cdot \|_{\psi_{\alpha/2}}$ is a genuine norm. So, the triangular inequality holds. Then,
$$\begin{aligned}
    K_i &= \|\xxi_i\|_{J,\psi_{\alpha/2}} \\
    &= \sup_{\|\eeta\|_2 = 1} \|\sum_{h=1}^{N_i d^2} \eta_h \xi_{ih}\|_{\psi_{\alpha/2}} \\
    &\leq \sup_{\|\eeta\|_2 = 1} \sum_{h=1}^{N_i d^2} |\eta_h| \|\xi_{ih}\|_{\psi_{\alpha/2}} \\
    &\leq \sqrt{N_i}d \max_{h}\|\xi_{ih}\|_{\psi_{\alpha/2}} \\
    &\leq C \sqrt{N_i} d \left(C_z \kappa \sqrt{\max_{\ell}\vertiii{\Ssigma_B^{\star(\ell,\ell)}}_2} + \sigma_{\epsilon}\right)^2.
\end{aligned}$$
Otherwise, if $0 < \alpha < 2,$ then $\| \cdot \|_{\psi_{\alpha/2}}$ is only a quasi-norm, which means the triangular inequality does not hold. We need an additional technical lemma to bound $K_i.$

\begin{lemma}
\label{lemma-technical}
    For any $0 < \zeta < 1,$ any scalar $a_1,...,a_m,$ and any random variable $X_1,...,X_m,$
    $$\left\|\sum_{h=1}^m a_h X_h \right\|_{\psi_{\zeta}} \leq \left(\sum_{h=1}^m |a_h|^\zeta \|X_h\|_{\psi_{\zeta}}^{\zeta}\right)^{1/\zeta}.$$
\end{lemma}

\begin{proof}
    Let $L_h$ be any number such that $L_h > \|X_h\|_{\psi_{\zeta}}.$ By the definition of Orlicz norm, we have $\EE \exp \left(\frac{|X_h|^{\zeta}}{L_h^{\zeta}}\right) \leq 2$ for $h\in[m].$ Set $L:=\left(\sum_{h=1}^m |a_h|^{\zeta} L_h^{\zeta} \right)^{1/\zeta}.$ If $L=0,$ the claim is trivial. Otherwise, define $\mu_h := \frac{|a_h|^{\zeta} L_h^{\zeta}}{L^{\zeta}}.$ Note that $\mu_h \geq 0$ and $\sum_{h=1}^m \mu_h = 1.$ Since $0 < \zeta < 1,$ the mapping $x \mapsto x^{\zeta}$ is subadditive. So, $|\sum_{h=1}^m a_h X_h |^{\zeta} \leq \sum_{h=1}^m |a_h|^{\zeta} |X_h|^{\zeta}.$ Therefore, 
    $$\frac{|\sum_{h=1}^m a_h X_h |^{\zeta}}{L^{\zeta}} \leq \sum_{h=1}^m \frac{|a_h|^{\zeta} |X_h|^{\zeta}}{L^{\zeta}} = \sum_{h=1}^m \mu_h \frac{|X_h|^{\zeta}}{L_h^{\zeta}}.$$ 
    Since $x\mapsto e^x$ is convex, by Jensen's inequality and monotonicity of the exponential function, 
    $$\exp\left(\frac{|\sum_{h=1}^m a_h X_h |^{\zeta}}{L^{\zeta}}\right) \leq \exp\left(\sum_{h=1}^m \mu_h \frac{|X_h|^{\zeta}}{L_h^{\zeta}}\right) \leq \sum_{h=1}^m \mu_h \exp\left(\frac{|X_h|^{\zeta}}{L_h^{\zeta}}\right).$$
    Taking the expectations of both sides,
    $$\EE \exp\left(\frac{|\sum_{h=1}^m a_h X_h |^{\zeta}}{L^{\zeta}}\right) \leq \sum_{h=1}^m \mu_h \EE \exp\left(\frac{|X_h|^{\zeta}}{L_h^{\zeta}}\right) \leq 2.$$
    So, by the definition of the Orlicz norm, $\|\sum_{h=1}^m a_h X_h\|_{\psi_{\zeta}} \leq L.$ Let $L_h$ approach $\|X_h\|_{\psi_{\zeta}}$ from above, we have $\left\|\sum_{h=1}^m a_h X_h \right\|_{\psi_{\zeta}} \leq \left(\sum_{h=1}^m |a_h|^\zeta \|X_h\|_{\psi_{\zeta}}^{\zeta}\right)^{1/\zeta}.$
\end{proof}

Using Lemma \ref{lemma-technical} and Jensen's inequality, we have 
$$\begin{aligned}
    K_i &= \|\xxi_i\|_{J,\psi_{\alpha/2}} \\
    &= \sup_{\|\eeta\|_2 = 1} \|\sum_{h=1}^{N_i d^2} \eta_h \xi_{ih}\|_{\psi_{\alpha/2}} \\
    &\leq \sup_{\|\eeta\|_2 = 1} \left(\sum_{h=1}^{N_i d^2} |\eta_h|^{\alpha/2} \|\xi_{ih}\|_{\psi_{\alpha/2}}^{\alpha/2}\right)^{2/\alpha} \\
    &\leq \sup_{\|\eeta\|_2 = 1} \left(\sum_{h=1}^{N_i d^2} |\eta_h|^{\alpha/2}\right)^{2/\alpha} \max_{h}\|\xi_{ih}\|_{\psi_{\alpha/2}} \\
    &\leq \sup_{\|\eeta\|_2 = 1} \left(N_i d^2 \left(\frac{1}{N_i d^2} \sum_{h=1}^{N_i d^2}\eta_h^2\right)^{\alpha/4}\right)^{2/\alpha} \max_{h}\|\xi_{ih}\|_{\psi_{\alpha/2}}\\
    & = \sup_{\|\eeta\|_2 = 1} \left( N_i d^2 \right)^{\frac{2}{\alpha} - \frac{1}{2}} \|\eeta\|_2 \max_{h}\|\xi_{ih}\|_{\psi_{\alpha/2}} \\
    & \leq C \cdot \left( N_i d^2 \right)^{\frac{2}{\alpha} - \frac{1}{2}} \left(C_z \kappa \sqrt{\max_{\ell}\vertiii{\Ssigma_B^{\star(\ell,\ell)}}_2} + \sigma_{\epsilon}\right)^2.
\end{aligned}$$
Now, we have concluded that 
\begin{equation}
\label{appendix-equ-ki}
    K_i \leq \begin{cases}
        C\sqrt{N_i}d\left(C_z \kappa \sqrt{\max_{\ell}\vertiii{\Ssigma_B^{\star(\ell,\ell)}}_2} + \sigma_{\epsilon}\right)^2, & \alpha \geq 2, \\
        C(N_i d^2)^{\frac{2}{\alpha} - \frac{1}{2}}\left(C_z \kappa \sqrt{\max_{\ell}\vertiii{\Ssigma_B^{\star(\ell,\ell)}}_2} + \sigma_{\epsilon}\right)^2, & 0< \alpha < 2,
    \end{cases}
\end{equation}
for constant $C.$ Plugging \eqref{appendix-equ-ki} into the result in Theorem~\ref{thm-sel-consist} yields the results in Corollary~\ref{corollary-rate}.

\section{Proof of Theorem~\ref{thm-fixed-sel-consistency}}
\label{appendix-fix}

\begin{proof}
We divide the proof into two parts. In part 1, we derive the estimated version of Assumption~\ref{a-fix2} and \ref{a-fix3}. In part 2, we construct the Primal-Dual Witness and use the estimated versions of Assumption~\ref{a-fix2} and \ref{a-fix3} to prove the selection consistency.

\paragraph*{Part 1. }
\paragraph*{Estimated version of Assumption~\ref{a-fix3}}
Recall that we have defined the true Gram matrix $\Ggamma = \frac{1}{N}\sum_{i=1}^m (\X_i \otimes \I_d)^\top \Ssigma_i^{\star^{-1}} (\X_i \otimes \I_d) \in \R^{pd \times pd}.$ We now define the estimated Gram matrix $\hat \Ggamma = \frac{1}{N}\sum_{i=1}^m (\X_i \otimes \I_d)^\top \hat \Ssigma_i^{-1} (\X_i \otimes \I_d) \in \R^{pd \times pd}.$ Let $\Xxi = \hat \Ggamma - \Ggamma$ be the difference between the estimated and true Gram matrix, and let event $\Eepsilon = \{\max_i \vertiii{\hat \Ssigma_i^{1/2} (\hat{\Ssigma}_i^{-1} - \Ssigma_i^{\star^{-1}}) \hat \Ssigma_i^{1/2
}}_2 \leq a_m\}$, which holds with probability at least $1 - \delta_{\Sigma}$. Then on event $\Eepsilon,$ for each entry $(j,k),$
\begin{align*}
    |\Xxi_{j,k}| & = \frac{1}{N} \left|\sum_{i=1}^m (\X_i \otimes \I_d)^\top_j (\hat \Ssigma_i^{-1} - \Ssigma_i^{\star^{-1}}) (\X_i \otimes \I_d)^\top_k\right| \\
    & = \frac{1}{N} \left|\sum_{i=1}^m \left[\hat \Ssigma_i^{-1/2}(\X_i \otimes \I_d)_j\right]^\top \hat \Ssigma_i^{1/2} (\hat \Ssigma_i^{-1} - \Ssigma_i^{\star^{-1}}) \hat \Ssigma_i^{1/2} \left[\hat \Ssigma_i^{-1/2} (\X_i \otimes \I_d)_k\right]^\top\right| \\
    & \leq \frac{1}{N} \sum_{i=1}^m \left\|\hat \Ssigma_i^{-1/2}(\X_i \otimes \I_d)_j\right\|_2 \bigvertiii{\hat \Ssigma_i^{1/2} (\hat \Ssigma_i^{-1} - \Ssigma_i^{\star^{-1}}) \hat \Ssigma_i^{1/2}}_2 \left\|\hat \Ssigma_i^{-1/2}(\X_i \otimes \I_d)_k\right\|_2 \\
    & \leq a_m \frac{1}{N} \sum_{i=1}^m \left\|\hat \Ssigma_i^{-1/2}(\X_i \otimes \I_d)_j\right\|_2 \left\|\hat \Ssigma_i^{-1/2}(\X_i \otimes \I_d)_k\right\|_2 \\
    & \leq a_m,
\end{align*}
where the last inequality follows the standardized column design. So $\vertiii{\Xxi}_{\max} := \max_j \max_k |\Xxi_{j,k}| \leq a_m,$ then $\vertiii{\Xxi_{\Stilde \Stilde}}_{\infty} \leq \stilde a_m, \vertiii{\Xxi_{\Stilde^c \Stilde}}_{\infty} \leq \stilde a_m,$ and $\vertiii{\Xxi_{\Stilde\Stilde}}_2 \leq \|\Xxi_{\Stilde\Stilde}\|_F \leq \stilde a_m.$ 


Note that $\hat \Ggamma_{\Stilde \Stilde} = \Ggamma_{\Stilde \Stilde} + \Xxi_{\Stilde \Stilde} = \Ggamma_{\Stilde \Stilde}(\I + \Ggamma_{\Stilde \Stilde}^{-1}\Xxi_{\Stilde \Stilde}),$ then $\hat \Ggamma_{\Stilde \Stilde}^{-1} = (\I + \Ggamma_{\Stilde \Stilde}^{-1}\Xxi_{\Stilde \Stilde})^{-1} \Ggamma_{\Stilde \Stilde}^{-1}$ and $\vertiii{\Ggamma_{\Stilde \Stilde}^{-1}\Xxi_{\Stilde \Stilde}}_{\infty} \leq \vertiii{\Ggamma_{\Stilde \Stilde}^{-1}}_{\infty} \vertiii{\Xxi_{\Stilde \Stilde}}_{\infty} \leq  \frac{\stilde a_m}{C_e} \leq \frac{\nu_x}{8} < \frac{1}{2},$ which follows from the condition stated in Theorem~\ref{thm-fixed-sel-consistency}. By Neumann series approximation,
$$\left(\I + \Ggamma_{\Stilde \Stilde}^{-1}\Xxi_{\Stilde \Stilde}\right)^{-1} = \I - \Ggamma_{\Stilde \Stilde}^{-1}\Xxi_{\Stilde \Stilde} + \left(\Ggamma_{\Stilde \Stilde}^{-1}\Xxi_{\Stilde \Stilde}\right)^2 - \cdots .$$
Then
$$\bigvertiii{\left(\I + \Ggamma_{\Stilde \Stilde}^{-1}\Xxi_{\Stilde \Stilde}\right)^{-1}}_{\infty} \leq \sum_{k=0}^{\infty} \bigvertiii{\Ggamma_{\Stilde \Stilde}^{-1}\Xxi_{\Stilde \Stilde}}_{\infty}^k = \frac{1}{1-\bigvertiii{\Ggamma_{\Stilde \Stilde}^{-1}\Xxi_{\Stilde \Stilde}}_{\infty}}.$$
Recall that $\vertiii{\Ggamma_{\Stilde \Stilde}^{-1}\Xxi_{\Stilde \Stilde}}_{\infty} < \frac{1}{2},$ so
$$\bigvertiii{\hat \Ggamma_{\Stilde \Stilde}^{-1}}_{\infty} \leq \frac{\bigvertiii{\Ggamma_{\Stilde\Stilde}^{-1}}_{\infty}}{1 - \bigvertiii{\Ggamma_{\Stilde \Stilde}^{-1}\Xxi_{\Stilde \Stilde}}_{\infty}} \leq 2 \bigvertiii{\Ggamma_{\Stilde\Stilde}^{-1}}_{\infty} \leq \frac{2}{C_e}.$$

\paragraph*{Estimated version of Assumption~\ref{a-fix2}.}

On the other hand, we have
\begin{align*}
    \hat \Ggamma_{\Stilde^c \Stilde} \hat \Ggamma_{\Stilde \Stilde}^{-1} & = \left(\Ggamma_{\Stilde^c \Stilde} + \Xxi_{\Stilde^c \Stilde}\right)\hat \Ggamma_{\Stilde \Stilde}^{-1} \\
    & = \Ggamma_{\Stilde^c \Stilde} \Ggamma_{\Stilde \Stilde}^{-1} \Ggamma_{\Stilde \Stilde}\hat \Ggamma_{\Stilde \Stilde}^{-1} + \Xxi_{\Stilde^c \Stilde} \hat \Ggamma_{\Stilde \Stilde}^{-1} \\
    & = \Ggamma_{\Stilde^c \Stilde} \Ggamma_{\Stilde \Stilde}^{-1} - \Ggamma_{\Stilde^c \Stilde} \Ggamma_{\Stilde \Stilde}^{-1} \Xxi_{\Stilde \Stilde}\hat \Ggamma_{\Stilde \Stilde}^{-1} + \Xxi_{\Stilde^c \Stilde} \hat \Ggamma_{\Stilde \Stilde}^{-1}.
\end{align*}
Hence, 
\begin{align*}
    \bigvertiii{ \hat \Ggamma_{\Stilde^c \Stilde} \hat \Ggamma_{\Stilde \Stilde}^{-1}}_{\infty} & \leq \bigvertiii{\Ggamma_{\Stilde^c \Stilde} \Ggamma_{\Stilde \Stilde}^{-1}}_{\infty} + \bigvertiii{\Ggamma_{\Stilde^c \Stilde} \Ggamma_{\Stilde \Stilde}^{-1}}_{\infty} \bigvertiii{\Xxi_{\Stilde \Stilde}}_{\infty} \bigvertiii{\hat \Ggamma_{\Stilde \Stilde}^{-1}}_{\infty} + \bigvertiii{\Xxi_{\Stilde^c \Stilde}}_{\infty} \bigvertiii{\hat \Ggamma_{\Stilde \Stilde}^{-1}}_{\infty} \\
    & \leq 1 - \nu_x + (1 - \nu_x)\frac{2\stilde a_m}{C_e} + \frac{2\stilde a_m}{C_e} \\ 
    & \leq 1 - \nu_x + (1 - \nu_x)\frac{\nu_x}{4} + \frac{\nu_x}{4} \\
    & = 1 - \frac{\nu_x}{2} - \frac{\nu_x^2}{4} \\
    & \leq 1 - \frac{\nu_x}{2},
\end{align*}
which is the estimated version of the irrepresentability assumption.

\paragraph*{Part 2.}
\paragraph*{Primal-Dual Witness Construction.}
Define the score vector $\hat \g = \frac{1}{N} \sum_{i=1}^m (\X_i \otimes \I_d)^\top \hat \Ssigma_i^{-1} \bm \varepsilon_i \in \R^{pd}.$ Then, the KKT condition for \eqref{equ-fgls} with $\mathcal{P}_\beta (\bbeta) = \tau \|\vvec(\bbeta^\top)\|_1$ is
\begin{equation}
\label{appendix-equ-kkt}
    \hat \Ggamma\left[\vvec(\hat \bbeta^\top) - \vvec(\bbeta^{\star\top})\right] - \hat \g + \tau \hat \u = 0,
\end{equation}
where $\u \in \partial \|\vvec(\hat \bbeta^\top)\|_1$ is a subgradient vector. We construct a pair $(\vvec(\check \bbeta^\top), \check{\u}) \in \R^{pd} \times \R^{pd}$ according to PDW procedure:
\begin{enumerate}[label=(\arabic*)]
    \item Obtain $\vvec(\check{\bbeta}^\top)_{\Stilde}$ by solving the restricted problem, and set $\vvec(\check{\bbeta})_{\Stilde^c} = \bm 0.$
    \item Choose $\check{\u}_{\Stilde}$ as an element of the subdifferential of the $\ell_1$ norm evaluated at $\vvec(\check{\bbeta})_{\Stilde}.$
    \item Solve for $\check{\u}_{\Stilde^c}$ from the KKT condition \eqref{appendix-equ-kkt}, then check the strict dual feasibility condition $\|\check{\u}_{\Stilde^c}\|_{\infty} < 1.$
    \item Check the sign consistency condition $\check{\u}_{\Stilde} = \mathrm{sign}(\vvec(\bbeta^{\star \top})_{\Stilde}).$
\end{enumerate}
Then, the KKT condition \eqref{appendix-equ-kkt} can be written as
\begin{equation*}
    \begin{bmatrix}
        \hat{\Ggamma}_{\Stilde \Stilde} & \hat{\Ggamma}_{\Stilde \Stilde^c} \\
        \hat{\Ggamma}_{\Stilde^c \Stilde} & \hat{\Ggamma}_{\Stilde^c \Stilde^c}
    \end{bmatrix}
    \begin{bmatrix}
        \vvec(\check{\bbeta}^\top)_{\Stilde} - \vvec(\bbeta^{\star \top})_{\Stilde} \\
        \bm 0
    \end{bmatrix}
    -
    \begin{bmatrix}
        \hat{\g}_{\Stilde} \\
        \hat{\g}_{\Stilde^c}
    \end{bmatrix}
    + \tau
    \begin{bmatrix}
        \check{\u}_{\Stilde} \\
        \check{\u}_{\Stilde^c}
    \end{bmatrix}
    = 0.
\end{equation*}
Solving for the lower block and the upper block gives
\begin{equation}
\label{appendix-equ-fix-u}
    \check{\u}_{\Stilde^c} = \hat{\Ggamma}_{\Stilde^c \Stilde} \hat{\Ggamma}_{\Stilde \Stilde}^{-1} \check{\u}_{\Stilde} + \frac{1}{\tau} \left(\hat{\g}_{\Stilde^c} - \hat{\Ggamma}_{\Stilde^c \Stilde} \hat{\Ggamma}_{\Stilde \Stilde}^{-1} \hat{\g}_{\Stilde}\right),
\end{equation}
and
\begin{equation}
\label{appendix-equ-fix-beta}
    \vvec(\check{\bbeta}^\top)_{\Stilde} - \vvec(\bbeta^{\star \top})_{\Stilde} = \hat{\Ggamma}_{\Stilde \Stilde}^{-1} \left(\hat{\g}_{\Stilde} - \tau \cdot \mathrm{sign} (\vvec(\bbeta^{\star \top}))\right),
\end{equation}
respectively. Before we check the strict dual feasibility and bound $\|\vvec(\check{\bbeta}^\top)_{\Stilde} - \vvec(\bbeta^{\star \top})_{\Stilde}\|_{\infty}$, we first establish the $\ell_{\infty}$ bound for $\hat{\g}.$

\paragraph*{Establish the $\ell_{\infty}$ bound for $\hat{\g}$.}
With a slight abuse of notation, let $$\hat{g}_j = \frac{1}{N}\sum_{i=1}^m (\X_i \otimes \I_d)^\top_j \hat{\Ssigma}_i^{-1} \bm \varepsilon_i := \frac{1}{N} \sum_{i=1}^m \v_{ij}^\top \o_i,$$
for $j \in [pd],$ where $\v_{ij} = \Ssigma_i^{\star^{1/2}} \hat{\Ssigma}_i^{-1} (\X_i \otimes \I_d)_j$ and $\o_i = \Ssigma_i^{\star^{-1/2}} \bm \varepsilon_i.$ The following concentration argument should be understood as a sample-splitting
argument. More precisely, split the original subjects into two independent subsets $\mathcal I_{\Sigma}$ and $\mathcal I_g$. The covariance estimators 
$\hat\Ssigma_B$ and $\hat\Ssigma_{\epsilon}$, and hence 
$\hat\Ssigma_i$ for $i\in\mathcal I_g$, are constructed using only the data in 
$\mathcal I_{\Sigma}$. But the score vector $\hat{\g},$ estimated Gram matrix $\hat{\Ggamma},$ and event $\Eepsilon$ used in the fixed-effects selection step are then 
formed using only the independent subjects in \(\mathcal I_g\). 

To avoid introducing additional notation, we relabel the subjects in
$\mathcal I_g$ as $1,\ldots,m$ and write $N=\sum_{i=1}^m n_i$ for the total number
of observations in this relabeled score sample. With a slight abuse of notation,
the previously defined $\hat g_j$ is therefore interpreted as the
score over this relabeled sample. The same relabeling convention is also used for the estimated Gram matrix
$\hat\Ggamma$ and the event $\Eepsilon$ throughout this proof.
Equivalently, the proof analyzes the sample-splitting version of the problem \eqref{equ-fgls}, while retaining the notation of the full-sample problem for simplicity. If the splitting proportion is fixed, the relabeled $N$ is of the same order
as the original sample size, so the rate is unchanged up to constants.
Let
$$
\mathcal F_{\Sigma}
=
\sigma\{\hat\Ssigma_i:i\in\mathcal I_g\}
$$
be the sigma-field generated by the first-stage covariance estimators. Conditional on 
$\mathcal F_{\Sigma}$, the vectors $\hat\Ssigma_i$ are fixed, while the random errors
$\o_i$'s for $i\in\mathcal I_g$
remain independent, mean-zero, joint sub-Weibull$(\alpha)$ random vectors by Assumption~\ref{a-fix4}. Therefore, 
conditional on $\mathcal F_{\Sigma}$, for each $j\in[pd]$, $\hat{g}_j$
is a sum of independent mean-zero sub-Weibull($\alpha$) scalar random variables with deterministic coefficients $\v_{ij}$. On the event $\Eepsilon,$ the $\ell_2$ norm of $\v_{ij}$ is bounded by
\begin{align*}
    \|\v_{ij}\|_2 & = \|\Ssigma_i^{\star^{1/2}} \hat{\Ssigma}_i^{-1} (\X_i \otimes \I_d)_j\|_2 \\
    & \leq \bigvertiii{\Ssigma_i^{\star^{1/2}} \hat{\Ssigma}_i^{-1/2
    }}_2 \|\hat \Ssigma_i^{-1/2} (\X_i \otimes \I_d)_j\|_2 \\
    & = \sqrt{\bigvertiii{\hat{\Ssigma}_i^{-1/2
    }\Ssigma_i^{\star} \hat{\Ssigma}_i^{-1/2
    }}_2} \|\hat \Ssigma_i^{-1/2} (\X_i \otimes \I_d)_j\|_2\\
    & \leq \frac{1}{\sqrt{1 - a_m}} \|\hat{\Ssigma}_i^{-1/2
    } (\X_i \otimes \I_d)_j\|_2 \\
    & \leq (1 + a_m) \|\hat{\Ssigma}_i^{-1/2
    } (\X_i \otimes \I_d)_j\|_2,
\end{align*}
where the last inequality follows the sufficient condition that $a_m \leq \frac{\sqrt{5} - 1}{2} \approx 0.618.$
Again, Theorem 3.1 from \cite{kuchibhotla2022moving} gives
\begin{align*}
    \PP\left(|\hat{g}_j| \geq t \mid \mathcal{F}_{\Sigma}\right) &\leq 2 \exp\left\{-c\cdot \min \left[\frac{N^2t^2}{\sum_{i=1}^m \|\v_{ij}^\top \o_i\|_{\psi_{\alpha}}^2}, \left(\frac{Nt}{\max_i \|\v_{ij}^\top \o_i\|_{\psi_{\alpha}}}\right)^{\alpha}\right]\right\} \\
    &\leq 2 \exp\left\{-c\cdot \min \left[\frac{N^2t^2}{\sum_{i=1}^m \left\|\v_{ij}\right\|_2^2 \tilde K_i^2}, \left(\frac{Nt}{\max_i \left\|\v_{ij}\right\|_2 \tilde K_i}\right)^{\alpha}\right]\right\} \\
    &\leq 2 \exp\left\{-c\cdot \min \left[\frac{N^2 t^2}{(1+a_m)^2 N (\max_i \tilde K_i)^2}, \left(\frac{Nt }{(1+a_m) \sqrt{N} \max_i \tilde K_i}\right)^{\alpha}\right]\right\}.
\end{align*}
By the union bound, we have
\begin{equation*}
    \PP\left(\|\hat{\g}\|_{\infty} \geq t \mid \mathcal{F}_{\Sigma} \right) \leq 2pd \cdot \exp\left\{-c\cdot \min \left[\frac{N^2 t^2}{(1+a_m)^2 N (\max_i \tilde K_i)^2}, \left(\frac{Nt }{(1+a_m) \sqrt{N} \max_i \tilde K_i}\right)^{\alpha}\right]\right\}.
\end{equation*}
Since $\Eepsilon \in\mathcal F_\Sigma$, setting $t = \frac{\nu_x \cdot \tau}{8},$ we have
\begin{align*}
    \PP\left(\|\hat{\g}\|_{\infty} \geq \frac{\nu_x \cdot \tau}{8}\right) & \leq \PP(\Eepsilon^c) + \PP\left(\left\{\|\hat{\g}\|_{\infty} \geq \frac{\nu_x \cdot \tau}{8}\right\} \cap \Eepsilon\right) \\
    & = \PP(\Eepsilon^c) + \mathbb E\left[\mathbf 1_{\mathcal E}\mathbb P\left(\left\{\|\hat{\g}\|_{\infty} \geq \frac{\nu_x \cdot \tau}{8}\right\}\mid\mathcal F_\Sigma\right) \right] \\
    &\leq \delta_{\Sigma} + 2pd \cdot \exp\left\{-c\cdot \min \left[\frac{N \nu_x^2 \tau^2}{64(1+a_m)^2 (\max_i \tilde K_i)^2}, \left(\frac{\sqrt{N}\nu_x \tau }{8(1+a_m) \max_i \tilde K_i}\right)^{\alpha}\right]\right\}.
\end{align*}
With a slightly different constant $\tilde c$, let 
\begin{equation*}
    \tau \geq \frac{8 (1+a_m)
    \max_i \tilde{K}_i}{\tilde{c} \nu_x \sqrt{N}}\max \left\{\sqrt{\log \frac{2pd}{\delta_{\beta}}}, \left(\log \frac{2pd}{\delta_{\beta}}\right)^{1/\alpha}\right\},
\end{equation*}
then $\|\hat{\g}\|_{\infty} < \frac{\nu_x \cdot \tau}{8}$ holds with probability at least $1-\delta_{\Sigma} - \delta_{\beta}.$

\paragraph*{Check strict dual feasibility.}
We can bound the $\ell_{\infty}$ norm of \eqref{appendix-equ-fix-u} by
\begin{align*}
    \|\check{\u}_{\Stilde^c}\|_{\infty} & \leq \bigvertiii{\hat{\Ggamma}_{\Stilde^c \Stilde} \hat{\Ggamma}_{\Stilde \Stilde}^{-1}}_{\infty} + \frac{1}{\tau} \left(\|\hat{\g}_{\Stilde^c}\|_{\infty} + \bigvertiii{\hat{\Ggamma}_{\Stilde^c \Stilde} \hat{\Ggamma}_{\Stilde \Stilde}^{-1}}_{\infty}\|\hat{\g}_{\Stilde}\|_{\infty}\right) \\
    & \leq 1 - \frac{\nu_x}{2} + \frac{1}{\tau} (2 - \frac{\nu_x}{2}) \|\hat{\g}\|_{\infty} \\
    & \leq 1 - \frac{\nu_x}{4} - \frac{\nu_x^2}{16} \\
    & < 1,
\end{align*}
holds with probability at least $1-\delta_{\Sigma} - \delta_{\beta}.$

\paragraph*{Establish $\ell_{\infty}$ bound for \eqref{appendix-equ-fix-beta}.}
Since the strict dual feasibility is verified, the PDW candidate is the unique global lasso solution and $\check{\bbeta} = \hat{\bbeta}.$ So,
\begin{align*}
    \|\vvec(\hat{\bbeta}^\top)_{\Stilde} - \vvec(\bbeta^{\star \top})_{\Stilde}\|_{\infty} & \leq \bigvertiii{\hat{\Ggamma}_{\Stilde \Stilde}^{-1}}_{\infty} \|\hat{\g}_{\Stilde}\|_{\infty} + \tau \bigvertiii{\hat{\Ggamma}_{\Stilde \Stilde}^{-1}}_{\infty} \\
    & \leq 2  \|\hat{\g}\|_{\infty} / C_e + 2
    \tau / C_e \\
    & \leq 2  (1 + \frac{\nu_x}{8}) \tau / C_e
\end{align*}
holds with probability at least $1-\delta_{\Sigma} - \delta_{\beta}.$
\end{proof}

\section{Estimation error bounds for refitted $\hat{\Ssigma}_B$ in Step 3, $\hat{\Ssigma}_{\epsilon}$ in Step 4, and $\hat{\Ssigma}_i$ in FGLS \eqref{equ-fgls}}
\label{appendix-am}
\begin{assumption}
    \label{a-am-1}
    There exists a constant $C_B>0$ such that 
    $$\Lambda_{\min} \left(\frac{1}{Md^2} \A_\J^\top \A_\J\right) \geq C_B,$$
    where $\A$ is defined in \eqref{equ-theory-linreg}, $\J$ indexes the active block of $\Ssigma_B$ with cardinality be $s^2,$ and $\A_\J \in \R^{Md^2 \times s^2}.$
\end{assumption}

\begin{assumption}
    \label{a-am-2}
    There exists a constant $C_r > 0$ such that each row of the design is bounded by
    $$\max_{\ell} \left\|\A_{\J_{\ell,.}}\right\|_2 \leq C_r,$$
    where $\A_{\J_{\ell,.}} \in \R^{s^2}$ denotes the $\ell$-th row of matrix $\A_\J.$
\end{assumption}

\begin{assumption}
\label{a-am-4}
The random effects $\vvec(\B_{i})$ and random error $\eepsilon_{ij}$ are both joint sub-Gaussian random vectors. 
\end{assumption}

\begin{proof}
In this proof, we first give bounds for $\vertiii{\hat \Ssigma_B - \Ssigma_B^\star}_2$ and $\vertiii{\hat \Ssigma_{\epsilon} - \Ssigma_{\epsilon}^\star}_2$, then we show the bound of $\max_i \bigvertiii{\hat \Ssigma_i^{1/2} (\hat{\Ssigma}_i^{-1} - \Ssigma_i^{\star^{-1}}) \hat \Ssigma_i^{1/2}}_2,$ which is $a_m.$
Throughout this proof, we condition on the event $\{\hat{S} = S\}$, i.e., the random-effects selection in Step 2 is correct, and use the oracle $\bbeta^\star$ to replace the working estimator $\hat \bbeta$ for simplicity. For notational simplicity, we focus on the sub-Gaussian setting for the random effects and random errors. The more general sub-Weibull setting can be handled using the same proof techniques developed in the previous analysis.

\paragraph*{Bound $\vertiii{\hat \Ssigma_B - \Ssigma_B^\star}_2$.}
Conditional on the event $\{\hat{S} = S\}$ that we correctly selected all the true random effects, we consider the refitting problem (Step 3 of \texttt{MOMENT})
\begin{align}
\label{equ-refit}
    &\min_{\Ssigma_B \succeq 0} \frac{1}{2} \sum_{i=1}^m \sum_{j \neq k}^{n_i} \left\| (\y_{ij}- \bbeta^{\star\top}\x_{ij})(\y_{ik}-\bbeta^{\star\top}\x_{ik})^\top - (\bm{I}_d \otimes \bm{z}_{ij}^\top)\Ssigma_B (\bm{I}_d \otimes \bm{z}_{ik}) \right\|_F^2  \\
    &\text{ s.t. } \quad \Ssigma_{B_{tt}} = 0 \quad \text{for all } t \notin \hat{S} \nonumber.
\end{align}
Then the refitting problem \eqref{equ-refit} without the PSD constraint corresponds to an unpenalized linear regression problem
$$\r = \A_{\J} \vvarphi_S^\star + \xxi,$$
where $\r, \A,$ and $\xxi$ are defined in \eqref{equ-theory-linreg}, $\J$ indexes the active block of $\Ssigma_B,$ and $\vvarphi_S = \vvec(\Ssigma_{B_S}) \in \R^{s^2}$ represents the vectorized active block of $\Ssigma_B.$ Therefore, we have
$$\hat \vvarphi_S = (\A_\J^\top \A_\J)^{-1} \A_\J^\top \r, \quad \text{and} \quad 
\hat \vvarphi_S - \vvarphi_S^\star = (\A_\J^\top \A_\J)^{-1} \A_\J^\top \xxi.$$
Hence, by Assumption~\ref{a-am-1},
\begin{align*}
    \bigvertiii{\hat \Ssigma_{B_S} - \Ssigma_{B_S}^\star}_2 &\leq \left\|\hat \Ssigma_{B_S} - \Ssigma_{B_S}^\star\right\|_F \\
    & \leq \left\|\hat \vvarphi_S - \vvarphi_S^\star\right\|_2 \\
    & \leq \left\|\left(\frac{1}{Md^2} \A_\J^\top \A_\J\right)^{-1}\right\|_2 \left\|\frac{1}{Md^2} \A_\J^\top \xxi\right\|_2 \\
    & \leq \frac{1}{C_B} \left\|\frac{1}{Md^2} \A_\J^\top \xxi\right\|_2.
\end{align*}
Then, we use the $\varepsilon$-net argument in Section 4.2 of \cite{vershynin2018high} to bound $\|\frac{1}{Md^2} \A_\J^\top \xxi\|_2.$ To elaborate, an $\varepsilon$-net $\mathcal{N}_{\varepsilon}$ of the unit sphere $\SS^{s^2 - 1} = \{\u \in \R^{s^2}: \|\u\|_2 = 1\}$ is a finite subset of $\SS^{s^2 - 1}$ such that for every $\u \in \SS^{s^2 - 1},$ there exists $\v \in \mathcal{N}_{\varepsilon}$ satisfying $\|\u - \v \|_2 \leq \varepsilon.$ By Proposition 4.2.13 of \cite{vershynin2018high}, the cardinality of $\mathcal{N}_{\varepsilon}$ is upper bounded by $|\mathcal{N}_{\varepsilon}| \leq (\frac{2}{\varepsilon} + 1)^{s^2}.$

Define a unit vector $\tilde{\u} = \frac{1}{Md^2} \A_\J^\top \xxi \big/ \|\frac{1}{Md^2} \A_\J^\top \xxi\|_2,$ then $\|\frac{1}{Md^2} \A_\J^\top \xxi\|_2 = \tilde{\u}^\top \frac{1}{Md^2} \A_\J^\top \xxi.$ Choose $\tilde{\v} \in \mathcal{N}_{\varepsilon}$ such that $\|\tilde{\u} - \tilde{\v}\|_2 \leq \varepsilon,$ then 
\begin{align*}
    \left\|\frac{1}{Md^2} \A_\J^\top \xxi\right\|_2 &= \tilde{\u}^\top \frac{1}{Md^2} \A_\J^\top \xxi \\
    & = \tilde{\v}^\top \frac{1}{Md^2} \A_\J^\top \xxi + (\tilde{\u} - \tilde{\v})^\top \frac{1}{Md^2} \A_\J^\top \xxi \\
    & \leq \left|\tilde{\v}^\top \frac{1}{Md^2} \A_\J^\top \xxi \right| + \left\|\tilde{\u} - \tilde{\v}\right\|_2 \left\|\frac{1}{Md^2} \A_\J^\top \xxi\right\|_2 \\
    & \leq \max_{\v \in \mathcal{N_{\varepsilon}}} \left|\v^\top \frac{1}{Md^2} \A_\J^\top \xxi \right| + \varepsilon \left\|\frac{1}{Md^2} \A_\J^\top \xxi\right\|_2.
\end{align*}
Taking $\varepsilon = \frac{1}{2},$ we have
$$\left\|\frac{1}{Md^2} \A_\J^\top \xxi\right\|_2 \leq \frac{1}{1 - \varepsilon} \max_{\v \in \mathcal{N_{\varepsilon}}} \left|\v^\top \frac{1}{Md^2} \A_\J^\top \xxi \right| = 2 \max_{\v \in \mathcal{N}_{1/2}} \left|\v^\top \frac{1}{Md^2} \A_\J^\top \xxi \right|, \quad \text{and} \quad \left|\mathcal{N}_{1/2}\right| \leq 5^{s^2}.$$
Therefore,
\begin{align*}
    \PP\left(\left\|\frac{1}{Md^2} \A_\J^\top \xxi\right\|_2 \geq t \right) & \leq \PP \left(2 \max_{\v \in \mathcal{N}_{1/2}} \left|\v^\top \frac{1}{Md^2} \A_\J^\top \xxi \right| \geq t\right) \\
    &\leq \sum_{\v \in \mathcal{N}_{1/2}} \PP \left(\left|\v^\top \frac{1}{Md^2} \A_\J^\top \xxi \right| \geq \frac{t}{2}\right) \\
    &\leq 5^{s^2} \sup_{\v \in \mathcal{N}_{1/2}} \PP\left(\left| \sum_{i=1}^m \frac{1}{Md^2} \v^\top \A_{\J,i}^\top \xxi_i \right| \geq \frac{t}{2}\right),
\end{align*}
where $\A_{\J,i} \in \R^{N_i d^2 \times s^2}$ represents the $i$-th row block of the design matrix $\A_\J$ that corresponds to the $i$-th subject, and $\xxi_i \in \R^{N_i d^2}$ are independent sub-Exponential random vectors that are defined in the analysis of random-effects selection. By Assumption~\ref{a-am-2}, the operator norm of $\A_{\J,i}$ is bounded by
\begin{equation*}
    \bigvertiii{\A_{\J,i}}_2 \leq \left\|\A_{\J,i}\right\|_F = \left(\sum_{\ell=1}^{N_i d^2}\left\|\A_{\J_{\ell,.}}\right\|_2^2\right)^{1/2} \leq C_r \sqrt{N_i}d.
\end{equation*}
Then, using Theorem 3.1 from \cite{kuchibhotla2022moving} with $\alpha/2 = 1,$ we have
\begin{align*}
     \PP\left(\left\|\frac{1}{Md^2} \A_\J^\top \xxi\right\|_2 \geq t \right) & \leq 5^{s^2} \sup_{\v \in \mathcal{N}_{1/2}} \PP\left(\left| \sum_{i=1}^m \frac{1}{Md^2} \v^\top \A_{\J,i}^\top \xxi_i \right| \geq \frac{t}{2}\right) \\
     & \leq 2 \times 5^{s^2} \sup_{\v \in \mathcal{N}_{1/2}} \exp\left\{-c \cdot \min \left[\frac{t^2 M^2 d^4}{4 \sum_{i=1}^m \left\| (\A_{\J,i} \v)^\top \xxi_i\right\|_{\psi_1}^2}, \frac{t M d^2}{2 \max_{i}\left\| (\A_{\J,i} \v)^\top \xxi_i\right\|_{\psi_1}} \right]\right\} \\
     & \leq 2 \times 5^{s^2} \exp\left\{-c \cdot \min \left[\frac{t^2 M^2 d^4}{4 (\max_{i} K_i)^2\sum_{i=1}^m \bigvertiii{\A_{\J,i}}_2^2}, \frac{t M d^2}{2 (\max_{i} K_i) \max_{i} \bigvertiii{\A_{\J,i}}_2} \right]\right\}\\
     & \leq 2 \times 5^{s^2} \exp\left\{-c \cdot \min \left[\frac{t^2 M d^2}{4 (\max_{i} K_i)^2 C_r^2}, \frac{t M d}{2 (\max_{i} K_i) (\max_{i} \sqrt{N_i}) C_r} \right]\right\}.
\end{align*}
Let 
$$t = \max\left\{\sqrt{\frac{4 (\max_{i} K_i)^2 C_r^2}{c M d^2} \left(\log \frac{4}{\delta_{\Sigma}} + s^2 \log 5\right)}, \frac{2 (\max_{i} K_i) (\max_{i} \sqrt{N_i}) C_r}{c M d} \left(\log \frac{4}{\delta_{\Sigma}} + s^2 \log 5\right)\right\},$$
then $\|\frac{1}{Md^2} \A_\J^\top \xxi\|_2 \leq t$ holds with probability at least $1 - \frac{\delta_{\Sigma}}{2}.$ Hence, $\vertiii{\hat \Ssigma_{B_S} - \Ssigma_{B_S}^\star}_2 \leq \frac{t}{C_B}$ also holds with probability at least $1 - \frac{\delta_{\Sigma}}{2}.$ 

Although symmetry of the solution is not imposed explicitly after dropping the PSD constraint, it is implied by the structure of the loss. Indeed, the loss is transpose-invariant: for any active-block matrix $\Ssigma \in \R^{s \times s},$ the term indexed by $(j,k,\ell,\ell')$ for $\Ssigma$ coincides with the term indexed by $(k,j,\ell',\ell)$ for $\Ssigma^\top.$ Since the objective sums over all ordered pairs $j \neq k$ and all response pairs $(\ell,\ell')$, we have $L_S(\Ssigma) = L_S(\Ssigma^\top),$ where $L_S$ denotes the objective function of \eqref{equ-refit}. On the other hand, Assumption~\ref{a-am-1} implies that $\A_\J^\top \A_\J$ is nonsingular, so the unconstrained active-block least-squares problem is strictly convex in $\vvec(\Ssigma)$ and therefore has a unique minimizer. Consequently, if $\hat \Ssigma_{B_S}$ is the minimizer, then $\hat \Ssigma_{B_S}^\top$ is also a minimizer by transpose-invariance. Uniqueness then yields $\hat \Ssigma_{B_S} = \hat \Ssigma_{B_S}^\top.$ Thus, even though symmetry is not imposed explicitly in the unconstrained refit, it follows automatically from transpose-invariance of the loss and uniqueness of the least-squares solution.
If the minimum eigenvalue of the true covariance matrix $\Lambda_{\min}(\Ssigma_{B_S}^\star) > \frac{t}{C_B}$ for this specific choice of $t,$ then by Weyl's inequality,
$$\Lambda_{\min}\left(\hat \Ssigma_{B_S}\right) \geq \Lambda_{\min}(\Ssigma_{B_S}^\star) - \bigvertiii{\hat \Ssigma_{B_S} - \Ssigma_{B_S}^\star}_2 > 0,$$
which implies that $\hat{\Ssigma}_{B_S}$ is indeed positive definite. Therefore, it is feasible for the PSD-constrained problem \eqref{equ-refit}.

Since the $K_i$ here follows exactly the same definition as the $K_i$ in Theorem~\ref{thm-sel-consist} except that we specify $\alpha = 2$ in this subsection, by \eqref{appendix-equ-ki}, we further simplify the rate of estimation error bound for $\Ssigma_{B_S}$ as
$$\vertiii{\hat \Ssigma_{B_S} - \Ssigma_{B_S}^\star}_2 = \mathcal{O}_p \left\{ \kappa^2 \max_{\ell} \bigvertiii{\Ssigma_B^{\star^{(\ell,\ell)}}}_2 \left(\sqrt{\frac{s^2}{M \big / \max_{i} N_i}} \vee \frac{s^2 }{M \big / \max_{i} N_i}\right)\right\}.$$

\paragraph*{Bound $\vertiii{\hat \Ssigma_{\epsilon} - \Ssigma_{\epsilon}^\star}_2$.}
Solving the random error covariance estimation problem \eqref{equ-opt-sigmae} without the PSD constraint gives
\begin{equation}
\label{appendix-equ-sigmae-hat}
    \hat{\Ssigma}_\epsilon =\frac{1}{N} \sum_{i=1}^m \sum_{j=1}^{n_i} \left[(\y_{ij}-\bbeta^{\star\top} \x_{ij})(\y_{ij}-\bbeta^{\star\top} \x_{ij})^\top -(\I_d \otimes \z_{ij}^\top)\hat{\Ssigma}_{B}(\I_d \otimes \z_{ij})\right].
\end{equation}
The moment equation \eqref{equ-moment} where $j = k$ gives
\begin{equation}
\label{appendix-equ-sigmae-star}
    \EE \left[(\y_{ij}-\bbeta^{\star\top} \x_{ij})(\y_{ij}-\bbeta^{\star\top} \x_{ij})^\top\right] = (\I_d \otimes \z_{ij}^\top) \Ssigma_B^\star (\I_d \otimes \z_{ij}) + \Ssigma_{\epsilon}^\star.
\end{equation}
Let $\eeta_{ij} = \y_{ij} - \bbeta^{\star\top} \x_{ij},$ \eqref{appendix-equ-sigmae-hat} minus \eqref{appendix-equ-sigmae-star} yields
$$\hat \Ssigma_{\epsilon} - \Ssigma_{\epsilon}^\star = \frac{1}{N} \sum_{i=1}^m \sum_{j=1}^{n_i} \left[\eeta_{ij}\eeta_{ij}^\top - \EE (\eeta_{ij}\eeta_{ij}^\top)\right] - \frac{1}{N} \sum_{i=1}^m \sum_{j=1}^{n_i} \left[(\I_d \otimes \z_{ij}^\top)(\hat{\Ssigma}_{B} - \Ssigma_B^\star)(\I_d \otimes \z_{ij})\right].$$
Recall again that in Corollary~\ref{corollary-rate}, we assume $\max_{i,j} \|\z_{ij}\|_2 \leq C_z,$ then $\frac{1}{N} \sum_{i=1}^m \sum_{j=1}^{n_i} \|\z_{ij}\|_2^2 \leq C_z^2.$ So,
\begin{align*}
    \bigvertiii{\hat \Ssigma_{\epsilon} - \Ssigma_{\epsilon}^\star}_2 &\leq \bigvertiii{\frac{1}{N} \sum_{i=1}^m \sum_{j=1}^{n_i} \left[\eeta_{ij}\eeta_{ij}^\top - \EE (\eeta_{ij}\eeta_{ij}^\top)\right]}_2 + \left(\frac{1}{N} \sum_{i=1}^m \sum_{j=1}^{n_i} \|\z_{ij}\|_2^2\right) \bigvertiii{\hat \Ssigma_{B} - \Ssigma_{B}^\star}_2 \\
    & \leq \bigvertiii{\frac{1}{N} \sum_{i=1}^m \sum_{j=1}^{n_i} \left[\eeta_{ij}\eeta_{ij}^\top - \EE (\eeta_{ij}\eeta_{ij}^\top)\right]}_2 + C_z^2 \bigvertiii{\hat \Ssigma_{B} - \Ssigma_{B}^\star}_2 \\
    & := \bigvertiii{\Ssigma_{\eta}}_2 + C_z^2 \bigvertiii{\hat \Ssigma_{B} - \Ssigma_{B}^\star}_2,
\end{align*}
where we denote $\Ssigma_{\eta} = \frac{1}{N} \sum_{i=1}^m \sum_{j=1}^{n_i} [\eeta_{ij}\eeta_{ij}^\top - \EE (\eeta_{ij}\eeta_{ij}^\top)].$ We use the $\varepsilon$-net argument again to bound $\vertiii{\frac{1}{N} \sum_{i=1}^m \sum_{j=1}^{n_i} [\eeta_{ij}\eeta_{ij}^\top - \EE (\eeta_{ij}\eeta_{ij}^\top)]}_2.$ Let $\mathcal{N}_{1/4}$ be an $\varepsilon$-net of the unit sphere $\SS^{d - 1} = \{\u \in \R^{d}: \|\u\|_2 = 1\}$, then $|\mathcal{N}_{1/4}| \leq 9^d.$ Let $\tilde{\u} \in \SS^{d-1},$ choose $\tilde{\v} \in \mathcal{N}_{1/4}$ such that $\|\tilde{\u} - \tilde{\v}\|_2 \leq \frac{1}{4}.$ Then,
$$\tilde{\u}^\top \Ssigma_{\eta} \tilde{\u} = \tilde{\v}^\top \Ssigma_{\eta} \tilde{\v} + (\tilde{\u} - \tilde{\v})^\top \Ssigma_{\eta} \tilde{\u} + \tilde{\v}^\top \Ssigma_{\eta} (\tilde{\u} - \tilde{\v}).$$
So,
\begin{align*}
    |\tilde{\u}^\top \Ssigma_{\eta} \tilde{\u}| & \leq |\tilde{\v}^\top \Ssigma_{\eta} \tilde{\v}| + \|\tilde{\u} - \tilde{\v}\|_2 \vertiii{\Ssigma_{\eta}}_2 \|\tilde{\u}\|_2 + \|\tilde{\v}\|_2 \vertiii{\Ssigma_{\eta}}_2 \|\tilde{\u} - \tilde{\v}\|_2 \\
    &\leq \max_{\v \in \mathcal{N}_{1/4}} |\v^\top \Ssigma_{\eta} \v| + \frac{1}{2} \vertiii{\Ssigma_{\eta}}_2.
\end{align*}
Taking the supremum over $\tilde{\u} \in \SS^{d-1},$ we have 
$$\vertiii{\Ssigma_{\eta}}_2 \leq \max_{\v \in \mathcal{N}_{1/4}} |\v^\top \Ssigma_{\eta} \v| + \frac{1}{2} \vertiii{\Ssigma_{\eta}}_2,$$
which implies
\begin{equation} \label{appendix-equ-Sigma_eta}
\vertiii{\Ssigma_{\eta}}_2 \leq 2 \max_{\v \in \mathcal{N}_{1/4}} |\v^\top \Ssigma_{\eta} \v|.
\end{equation}
Note that $\eeta_{ij} = \y_{ij} - \bbeta^{\star^\top} \x_{ij} = (\I_d \otimes \z_{ij}^\top) \vvec (\B_i) + \eepsilon_{ij},$ so $\Ssigma_{\eta}$ can be decomposed as
$$\Ssigma_{\eta} = \bm T_{\epsilon \epsilon} + \bm T_{BB} + \bm T_{B \epsilon} + \bm T_{\epsilon B} ,$$
where 
\begin{align*}
    & \bm T_{\epsilon \epsilon} = \frac{1}{N} \sum_{i=1}^m \sum_{j=1}^{n_i} \left\{\eepsilon_{ij} \eepsilon_{ij}^\top - \Ssigma_{\epsilon}^\star\right\} \\
    & \bm T_{BB} = \frac{1}{N} \sum_{i=1}^m \sum_{j=1}^{n_i} \left\{ (\I_d \otimes \z_{ij}^\top) \left[\vvec(\B_i) \vvec(\B_i)^\top - \Ssigma_B^\star\right] (\I_d \otimes \z_{ij})\right\} \\
    & \bm T_{B \epsilon} = \frac{1}{N} \sum_{i=1}^m \sum_{j=1}^{n_i} \left\{(\I_d \otimes \z_{ij}^\top) \vvec(\B_i) \eepsilon_{ij}^\top \right\} \\
    & \bm T_{\epsilon B} = \bm T_{B \epsilon}^\top.
\end{align*}
Following the same $\varepsilon$-net argument, each of these quantities is upper bounded by a similar format of \eqref{appendix-equ-Sigma_eta}. Let $K_{\epsilon} = \sup_{i,j} \|\eepsilon_{ij}\|_{J,\psi_2} = \sup_{i,j} \sup_{\|\v\|=1} \|\v^\top \eepsilon_{ij}\|_{\psi_2}$ and $K_{B} = \sup_{i,j} \|(\I_d \otimes \z_{ij}^\top) \vvec(\B_i)\|_{J,\psi_2} = \sup_{i,j} \sup_{\|\v\|=1} \|\v^\top (\I_d \otimes \z_{ij}^\top) \vvec(\B_i)\|_{\psi_2}.$ By the proof of Corollary~\ref{corollary-rate} (Supplementary Material~\ref{appendix-corollary-rate}), we have $K_{\epsilon} \leq \sqrt{d} \sigma_{\epsilon}$ and $K_B \leq \sqrt{d} C_z \kappa \sqrt{\max_{\ell} \vertiii{\Ssigma_B^{\star^{(\ell,\ell)}}}_2.}$ Then, since $\eepsilon_{ij}$'s are i.i.d. across $i \in [m]$ and $j \in [n_i],$ we have
{\small
\begin{align*}
    & \PP\left(\bigvertiii{\bm T_{\epsilon \epsilon}}_2 \geq t\right) \\
    \leq & 9^d \sup_{\v\in\mathcal{N}_{1/4}} \PP \left(|\v^\top \bm T_{\epsilon \epsilon} \v| \geq \frac{t}{2}\right) \\
    = & 9^d \sup_{\v\in\mathcal{N}_{1/4}} \PP \left(\left|\frac{1}{N} \sum_{i=1}^m \sum_{j=1}^{n_i} \left[(\v^\top\eepsilon_{ij})^2 - \EE (\v^\top\eepsilon_{ij})^2\right]\right| \geq \frac{t}{2}\right) \\
    \leq & 2 \times 9^d \sup_{\v\in\mathcal{N}_{1/4}} \exp\left\{-c \cdot \min\left[\frac{N^2 t^2}{4 \sum_{i=1}^m \sum_{j=1}^{n_i} \left\| \left[(\v^\top\eepsilon_{ij})^2 - \EE (\v^\top\eepsilon_{ij})^2\right]\right\|_{\psi_1}^2}, \frac{N t}{2 \max_i \left\| \left[(\v^\top\eepsilon_{ij})^2 - \EE (\v^\top\eepsilon_{ij})^2\right]\right\|_{\psi_1}} \right]\right\} \\
    \leq & 2 \times 9^d \exp\left\{-c \cdot \min\left[\frac{N t^2}{4 K_{\epsilon}^{4}}, \frac{N t}{2 K_{\epsilon}^{2}} \right]\right\} \\
    \leq & 2 \times 9^d \exp\left\{-c \cdot \min\left[\frac{N t^2}{4 d^2 \sigma_{\epsilon}^4}, \frac{N t}{2 d \sigma_{\epsilon}^2} \right]\right\}.
\end{align*}
}
Let
$$t = \frac{2d \sigma_{\epsilon}^2}{c}\max \left\{\sqrt{ \frac{\log \frac{16}{\delta} + d \log 9}{N}}, \frac{\log \frac{16}{\delta} + d \log 9}{N}\right\},$$
then $\vertiii{\bm T_{\epsilon \epsilon}}_2 \leq t$ with probability at least $1 - \frac{\delta_{\Sigma}}{8}.$ So, 
$$\bigvertiii{\bm T_{\epsilon \epsilon}}_2 = \mathcal{O}_p\left[d \left(\sqrt{\frac{d}{N}} \vee \frac{d}{N} \right)\right].$$
On the other hand, since $\B_i$'s are only independent across $i \in [m],$ we have
{\small
\begin{align*}
    & \PP\left(\bigvertiii{\bm T_{BB}}_2 \geq t\right) \\
    \leq & 9^d \sup_{\v\in\mathcal{N}_{1/4}} \PP \left(|\v^\top \bm T_{BB} \v| \geq \frac{t}{2}\right) \\
    = & 9^d \sup_{\v\in\mathcal{N}_{1/4}} \PP \left(\left|\frac{1}{N} \sum_{i=1}^m \sum_{j=1}^{n_i} \left\{[\v^\top(\I_d \otimes \z_{ij}^\top)\vvec(\B_i)]^2 - \EE [\v^\top(\I_d \otimes \z_{ij}^\top)\vvec(\B_i)]^2\right\}\right| \geq \frac{t}{2}\right) \\
    \leq & 2 \times 9^d \sup_{\v\in\mathcal{N}_{1/4}} \exp\left\{-c \cdot \min\left[\frac{N^2 t^2}{4 \sum_{i=1}^m \left\|\sum_{j=1}^{n_i} \left\{[\v^\top(\I_d \otimes \z_{ij}^\top)\vvec(\B_i)]^2 - \EE [\v^\top(\I_d \otimes \z_{ij}^\top)\vvec(\B_i)]^2\right\}\right\|_{\psi_1}^2}, \right.\right.\\
     &\left.\left. \qquad \qquad \qquad \frac{N t}{2 \max_i \left\|\sum_{j=1}^{n_i} \left\{[\v^\top(\I_d \otimes \z_{ij}^\top)\vvec(\B_i)]^2 - \EE [\v^\top(\I_d \otimes \z_{ij}^\top)\vvec(\B_i)]^2\right\}\right\|_{\psi_1}} \right]\right\} \\
    \leq & 2 \times 9^d \exp\left\{-c \cdot \min\left[\frac{N^2 t^2}{4 \sum_{i=1}^m n_i^2 K_B^{4}}, \frac{N t}{2 \max_i n_i K_B^{2}} \right]\right\} \\
    \leq & 2 \times 9^d \exp\left\{-c \cdot \min\left[\frac{N^2 t^2}{4 \sum_{i=1}^m n_i^2 d^2 C_z^4 \kappa^4 \max_{\ell}\vertiii{\Ssigma_B^{\star^{(\ell,\ell)}}}_2^2}, \frac{N t}{2 \max_i n_i d C_z^2 \kappa^2 \max_{\ell}\vertiii{\Ssigma_B^{\star^{(\ell,\ell)}}}_2} \right]\right\}.
\end{align*}
}
Let 
$$t = \frac{2d C_z^2 \kappa^2 \max_{\ell}\vertiii{\Ssigma_B^{\star^{(\ell,\ell)}}}_2}{c}\max \left\{\sqrt{ \frac{\log \frac{16}{\delta} + d \log 9}{N^2 / \sum_{i=1}^m n_i^2}}, \frac{\log \frac{16}{\delta} + d \log 9}{N / \max_i n_i}\right\},$$
then $\vertiii{\bm T_{BB}}_2 \leq t$ with probability at least $1 - \frac{\delta_{\Sigma}}{8}.$ Hence,
$$\bigvertiii{\bm T_{BB}}_2 = \mathcal{O}_p\left[\kappa^2 \max_{\ell}\vertiii{\Ssigma_B^{\star^{(\ell,\ell)}}}_2 d \left(\sqrt{\frac{d}{N^2 / \sum_{i=1}^m n_i^2}} \vee \frac{d}{N / \max n_i}\right)\right].$$
Moreover,
{\small
\begin{align*}
    & \PP\left(\bigvertiii{\bm T_{B\epsilon}}_2 \geq t\right) \\
    \leq & 9^d \sup_{\v\in\mathcal{N}_{1/4}} \PP \left(|\v^\top \bm T_{B\epsilon} \v| \geq \frac{t}{2}\right) \\
    = & 9^d \sup_{\v\in\mathcal{N}_{1/4}} \PP \left(\left|\frac{1}{N} \sum_{i=1}^m \sum_{j=1}^{n_i} \left[\v^\top(\I_d \otimes \z_{ij}^\top)\vvec(\B_i) \eepsilon_{ij}^\top \v\right]\right| \geq \frac{t}{2}\right) \\
    \leq & 2 \times 9^d \sup_{\v\in\mathcal{N}_{1/4}} \exp\left\{-c \cdot \min\left[\frac{N^2 t^2}{4 \sum_{i=1}^m \left\|\sum_{j=1}^{n_i} \left[\v^\top(\I_d \otimes \z_{ij}^\top)\vvec(\B_i) \eepsilon_{ij}^\top \v\right]\right\|_{\psi_1}^2}, \right.\right.\\
     &\left.\left. \qquad \qquad \qquad \frac{N t}{2 \max_i \left\|\sum_{j=1}^{n_i} \left[v^\top(\I_d \otimes \z_{ij}^\top)\vvec(\B_i) \eepsilon_{ij}^\top \v\right]\right\|_{\psi_1}} \right]\right\} \\
    \leq & 2 \times 9^d \exp\left\{-c \cdot \min\left[\frac{N^2 t^2}{4 \sum_{i=1}^m n_i^2 K_B^{2} K_{\epsilon}^2}, \frac{N t}{2 \max_i n_i K_B K_{\epsilon}} \right]\right\} \\
    \leq & 2 \times 9^d \exp\left\{-c \cdot \min\left[\frac{N^2 t^2}{4 \sum_{i=1}^m n_i^2 d^2 C_z^2 \kappa^2 \max_{\ell}\vertiii{\Ssigma_B^{\star^{(\ell,\ell)}}}_2 \sigma_{\epsilon}^2}, \frac{N t}{2 \max_i n_i d C_z \kappa \max_{\ell}\sqrt{\vertiii{\Ssigma_B^{\star^{(\ell,\ell)}}}_2}\sigma_{\epsilon}} \right]\right\}.
\end{align*}
}
Let 
$$t = \frac{2d C_z \kappa \sigma_{\epsilon} \sqrt{\max_{\ell}\vertiii{\Ssigma_B^{\star^{(\ell,\ell)}}}_2}}{c}\max \left\{\sqrt{ \frac{\log \frac{16}{\delta} + d \log 9}{N^2 / \sum_{i=1}^m n_i^2}}, \frac{\log \frac{16}{\delta} + d \log 9}{N / \max_i n_i}\right\},$$
then $\vertiii{\bm T_{B\epsilon}}_2 = \vertiii{\bm T_{\epsilon B}}_2 \leq t$ with probability at least $1 - \frac{\delta_{\Sigma}}{8}.$ Thus,
$$\bigvertiii{\bm T_{B \epsilon}}_2 = \bigvertiii{\bm T_{\epsilon B}}_2 = \mathcal{O}_p\left[\kappa \sqrt{\max_{\ell}\vertiii{\Ssigma_B^{\star^{(\ell,\ell)}}}_2} d \left(\sqrt{\frac{d}{N^2 / \sum_{i=1}^m n_i^2}} \vee \frac{d}{N / \max n_i}\right)\right]$$
Taking together, we have 
{\small
\begin{align*}
    &\bigvertiii{\hat \Ssigma_{\epsilon} - \Ssigma_{\epsilon}^\star}_2 \\
    \leq & \bigvertiii{\Ssigma_{\eta}}_2 + C_z^2 \bigvertiii{\hat \Ssigma_{B} - \Ssigma_{B}^\star}_2 \\
    \leq & \bigvertiii{\bm T_{BB}}_2 + \bigvertiii{\bm T_{\epsilon \epsilon}}_2 + \bigvertiii{\bm T_{B\epsilon}}_2 + \bigvertiii{\bm T_{\epsilon B}}_2 + C_z^2 \bigvertiii{\hat \Ssigma_{B} - \Ssigma_{B}^\star}_2 \\
    = & \mathcal{O}_p \left\{d \left(\sqrt{\frac{d}{N}} \vee \frac{d}{N} \right) + \left(\kappa^2 \max_{\ell}\vertiii{\Ssigma_B^{\star^{(\ell,\ell)}}}_2 + \kappa \sqrt{\max_{\ell}\vertiii{\Ssigma_B^{\star^{(\ell,\ell)}}}_2} \right) d \left(\sqrt{\frac{d}{N^2 / \sum_{i=1}^m n_i^2}} \vee \frac{d}{N / \max n_i}\right) \right. \\
    & \left. \qquad + \kappa^2 \max_{\ell} \bigvertiii{\Ssigma_B^{\star^{(\ell,\ell)}}}_2 \left(\sqrt{\frac{s^2}{M \big / \max_{i} N_i}} \vee \frac{s^2 }{M \big / \max_{i} N_i}\right) \right\}.
\end{align*}
}

By the similar argument as when bounding $\vertiii{\hat{\Ssigma}_{B_S} - \Ssigma_{B_S}^\star}_2,$ $\hat{\Ssigma}_{\epsilon}$ is also a feasible solution to the PSD-constrainted problem \eqref{equ-opt-sigmae}. 

\paragraph*{Bound $\max_i \vertiii{\hat \Ssigma_i^{1/2} (\hat{\Ssigma}_i^{-1} - \Ssigma_i^{\star^{-1}}) \hat \Ssigma_i^{1/2}}_2.$}

Denote the bounds of $\vertiii{\hat \Ssigma_{B_S} - \Ssigma_{B_S}^\star}_2$ and $\vertiii{\hat \Ssigma_{\epsilon} - \Ssigma_{\epsilon}^\star}_2$ as
$$\bigvertiii{\hat \Ssigma_{B_S} - \Ssigma_{B_S}^\star}_2 \leq b_m, \qquad \bigvertiii{\hat \Ssigma_{\epsilon} - \Ssigma_{\epsilon}^\star}_2 \leq e_m ,$$
with probability tending to $1.$ Denote
$$ a_m := \max\left\{\frac{b_m}{\Lambda_{\min}(\Ssigma_{B_S}^\star)}, \frac{e_m }{\Lambda_{\min}(\Ssigma_{\epsilon}^\star)}\right\}.$$
Recall that $\mathcal{Z}_i = [\I_d \otimes \z_{i1}, ..., \I_d \otimes \z_{i n_i}]^\top \in \R^{n_i d \times qd}$, then 
$$\Ssigma_i^\star = \mathcal{Z}_i \Ssigma_B^\star \mathcal{Z}_i^\top + \I_{n_i} \otimes \Ssigma_{\epsilon}^\star, \quad \hat \Ssigma_i = \mathcal{Z}_i \hat \Ssigma_B \mathcal{Z}_i^\top + \I_{n_i} \otimes \hat \Ssigma_{\epsilon}.$$
Therefore, 
$$\hat \Ssigma_i - \Ssigma_i^\star = \mathcal{Z}_i (\hat \Ssigma_B - \Ssigma_B^\star) \mathcal{Z}_i^\top + \I_{n_i} \otimes (\hat \Ssigma_{\epsilon} - \Ssigma_{\epsilon}^\star).$$
Conditional on the correct random-effects selection $\{\hat{S} = S\},$ we also have
$$\hat \Ssigma_i - \Ssigma_i^\star = \mathcal{Z}_{i_S} (\hat \Ssigma_{B_S} - \Ssigma_{B_S}^\star) \mathcal{Z}_{i_S}^\top + \I_{n_i} \otimes (\hat \Ssigma_{\epsilon} - \Ssigma_{\epsilon}^\star).$$
Since $\vertiii{\hat \Ssigma_{B_S} - \Ssigma_{B_S}^\star}_2 \leq b_m$, we have $-b_m \I_s \preceq \hat \Ssigma_{B_S} - \Ssigma_{B_S}^\star \preceq b_m \I_s.$ Also since $\Ssigma_{B_S}^\star \succeq \Lambda_{\min}(\Ssigma_{B_S}^\star) \I_s,$ we have $\I_s \preceq \Lambda_{\min}(\Ssigma_{B_S}^\star)^{-1} \Ssigma_{B_S}^\star.$ Therefore
$$-\frac{b_m}{\Lambda_{\min}(\Ssigma_{B_S}^\star)} \Ssigma_{B_S}^\star \preceq \hat \Ssigma_{B_S} - \Ssigma_{B_S}^\star \preceq \frac{b_m}{\Lambda_{\min}(\Ssigma_{B_S}^\star)} \Ssigma_{B_S}^\star.$$
Similarly, for $\Ssigma_{\epsilon},$ we have
$$-\frac{e_m}{\Lambda_{\min}(\Ssigma_{\epsilon}^\star)} \Ssigma_{\epsilon}^\star \preceq \hat \Ssigma_{\epsilon} - \Ssigma_{\epsilon}^\star \preceq \frac{e_m}{\Lambda_{\min}(\Ssigma_{\epsilon}^\star)} \Ssigma_{\epsilon}^\star.$$
By definition, $a_m = \max\left\{\frac{b_m}{\Lambda_{\min}(\Ssigma_{B_S}^\star)}, \frac{e_m }{\Lambda_{\min}(\Ssigma_{\epsilon}^\star)}\right\}$, so
$$-a_m \Ssigma_{B_S}^\star \preceq \hat \Ssigma_{B_S} - \Ssigma_{B_S}^\star \preceq a_m \Ssigma_{B_S}^\star,$$
and 
$$-a_m \Ssigma_{\epsilon}^\star \preceq \hat \Ssigma_{\epsilon} - \Ssigma_{\epsilon}^\star \preceq a_m  \Ssigma_{\epsilon}^\star.$$
From these two inequalities, we have
$$-a_m \mathcal{Z}_{i_S} \Ssigma_{B_S}^\star \mathcal{Z}_{i_S}^\top \preceq \mathcal{Z}_{i_S} (\hat \Ssigma_{B_S} - \Ssigma_{B_S}^\star) \mathcal{Z}_{i_S}^\top \preceq a_m \mathcal{Z}_{i_S} \Ssigma_{B_S}^\star \mathcal{Z}_{i_S}^\top,$$
and
$$-a_m (\I_{n_i} \otimes \Ssigma_{\epsilon}^\star) \preceq [\I_{n_i} \otimes (\hat \Ssigma_{\epsilon} - \Ssigma_{\epsilon}^\star)] \preceq a_m  (\I_{n_i} \otimes \Ssigma_{\epsilon}^\star).$$
Adding these two inequalities gives
$$-a_m \Ssigma_i^\star \preceq \hat{\Ssigma}_i - \Ssigma_i^\star \preceq a_m \Ssigma_i^\star,$$
which implies
$$(1-a_m) \I_{n_i d} \preceq \Ssigma_i^{\star^{-1/2}}\hat{\Ssigma}_i \Ssigma_i^{\star^{-1/2}} \preceq (1+a_m) \I_{n_i d}.$$
Since matrices $\Ssigma_i^{\star^{-1/2}}\hat{\Ssigma}_i \Ssigma_i^{\star^{-1/2}}$ and $\hat \Ssigma_i^{1/2}\Ssigma_i^{\star^{-1}} \hat \Ssigma_i^{1/2}$ have the same eigenvalues, we have
\begin{align*}
    &\bigvertiii{\hat \Ssigma_i^{1/2} (\hat \Ssigma_i^{-1} - \Ssigma_i^{\star^{-1}}) \hat \Ssigma_i^{1/2}}_2 \\
    =& \bigvertiii{\I_{n_i d} - \hat \Ssigma_i^{1/2}\Ssigma_i^{\star^{-1}} \hat \Ssigma_i^{1/2}}_2 \\
    \leq & a_m \\
    = & \max\left\{\frac{b_m}{\Lambda_{\min}(\Ssigma_{B_S}^\star)}, \frac{e_m }{\Lambda_{\min}(\Ssigma_{\epsilon}^\star)}\right\} \\
    =& \mathcal{O}_p \left\{d \left(\sqrt{\frac{d}{N}} \vee \frac{d}{N} \right) + \left(\kappa^2 \max_{\ell}\vertiii{\Ssigma_B^{\star^{(\ell,\ell)}}}_2 + \kappa \sqrt{\max_{\ell}\vertiii{\Ssigma_B^{\star^{(\ell,\ell)}}}_2} \right) d \left(\sqrt{\frac{d}{N^2 / \sum_{i=1}^m n_i^2}} \vee \frac{d}{N / \max n_i}\right) \right. \\
    & \left. \qquad + \kappa^2 \max_{\ell} \bigvertiii{\Ssigma_B^{\star^{(\ell,\ell)}}}_2 \left(\sqrt{\frac{s^2}{M \big / \max_{i} N_i}} \vee \frac{s^2 }{M \big / \max_{i} N_i}\right) \right\}\\
\end{align*}
We can simplify the rate of $a_m$ in the balanced case where $n_i = n$ for all $i\in[m],$
\begin{align*}
    a_m = & \mathcal{O}_p \left\{d \left(\sqrt{\frac{d}{N}} \vee \frac{d}{N} \right) + \left(\kappa^2 \max_{\ell}\vertiii{\Ssigma_B^{\star^{(\ell,\ell)}}}_2 + \kappa \sqrt{\max_{\ell}\vertiii{\Ssigma_B^{\star^{(\ell,\ell)}}}_2} \right) d \left(\sqrt{\frac{d}{m}} \vee \frac{d}{m}\right) \right. \\
    & \left. \qquad \qquad + \kappa^2 \max_{\ell} \bigvertiii{\Ssigma_B^{\star^{(\ell,\ell)}}}_2 \left(\sqrt{\frac{s^2}{m}} \vee \frac{s^2 }{m}\right) \right\}.
\end{align*}
Furthermore, we assume $m > d$, $m > s^2$, and random effect $\B_{i,\ell}$ is joint Gaussian distributed for all $\ell \in [d]$ so that $\kappa$ is a constant. In such a case, the rate of $a_m$ can be simplified as
\begin{align*}
    a_m = \mathcal{O}_p \left\{\sqrt{\frac{d^3}{N}} + \left( \max_{\ell}\vertiii{\Ssigma_B^{\star^{(\ell,\ell)}}}_2 +  \sqrt{\max_{\ell}\vertiii{\Ssigma_B^{\star^{(\ell,\ell)}}}_2} \right) \sqrt{\frac{d^3}{m}} + \max_{\ell} \bigvertiii{\Ssigma_B^{\star^{(\ell,\ell)}}}_2 \sqrt{\frac{s^2}{m}} \right\}.
\end{align*}
\end{proof}

\end{document}